\documentclass[12pt]{scrartcl}
\usepackage{a4}
\usepackage{amsthm}
\usepackage{amsmath}
\usepackage{amssymb}
\usepackage{amsfonts}
\usepackage{mathrsfs}
\usepackage{dsfont}
\usepackage{latexsym}
\usepackage{xcolor}
\usepackage{bbm,exscale}
\definecolor{Myblue}{rgb}{0,0,0.6}
\usepackage[colorlinks,citecolor=Myblue,linkcolor=Myblue,urlcolor=Myblue,pdfpagemode=UseNone]{hyperref}
\usepackage[square,numbers,sort&compress]{natbib}
\usepackage[all,cmtip]{xy}
\usepackage{tikz}
\usetikzlibrary{decorations.pathmorphing}
\usetikzlibrary{calc}
\usetikzlibrary{decorations.markings}
\usetikzlibrary{fadings,decorations.pathreplacing}
\usetikzlibrary{matrix,arrows}
\usepackage{bbding}
\usetikzlibrary{arrows,calc,decorations.pathreplacing,decorations.markings,shapes.geometric,shadows}

\tikzset{
    string/.style={draw=#1, postaction={decorate}, decoration={markings,mark=at position .51 with {\arrow[draw=#1]{>}}}},
    costring/.style={draw=#1, postaction={decorate}, decoration={markings,mark=at position .51 with {\arrow[draw=#1]{<}}}},
    ostring/.style={draw=#1, postaction={decorate}, decoration={markings,mark=at position .47 with {\arrow[draw=#1]{>}}}},
    ustring/.style={draw=#1, postaction={decorate}, decoration={markings,mark=at position .56 with {\arrow[draw=#1]{>}}}},
    oostring/.style={draw=#1, postaction={decorate}, decoration={markings,mark=at position .43 with {\arrow[draw=#1]{>}}}},
    uustring/.style={draw=#1, postaction={decorate}, decoration={markings,mark=at position .59 with {\arrow[draw=#1]{>}}}},
    directed/.style={string=blue!50!black}, 
    odirected/.style={ostring=blue!50!black}, 
    udirected/.style={ustring=blue!50!black}, 
    oodirected/.style={oostring=blue!50!black}, 
    uudirected/.style={uustring=blue!50!black},     
    redirected/.style={costring= blue!50!black},
    redirectedgreen/.style={costring= green!50!black},
    directedgreen/.style={string= green!50!black},
    redirectedlightgreen/.style={costring= green!65!black},
    directedlightgreen/.style={string= green!65!black},
}

\tikzset{-dot-/.style={decoration={
  markings,
  mark=at position 0.5 with {\fill circle (1.875pt);}},postaction={decorate}}}

\tikzset{
	Fdot/.style={circle, draw, fill, inner sep=0pt}, 
	Odot/.style={circle, draw, inner sep=0.1pt, minimum size=0.1cm}
	}

\def\nicehalfpalecolourscheme{\shadedraw[top color=blue!22, bottom color=blue!22, draw=white]}
\def\nicenotpalecolourscheme{\shadedraw[top color=blue!32, bottom color=blue!32, draw=white]}
\def\nicecolourscheme{\shadedraw[top color=blue!22, bottom color=blue!22, draw=blue!22]}
\def\nicepalecolourscheme{\shadedraw[top color=blue!12, bottom color=blue!12, draw=white]}

  \vfuzz \hfuzz
\makeatletter
\newcommand{\raisemath}[1]{\mathpalette{\raisem@th{#1}}}
\newcommand{\raisem@th}[3]{\raisebox{#1}{$#2#3$}}
\makeatother

\newcommand{\E}{\text{e}}
\newcommand{\I}{\text{i}}
\newcommand{\B}{\mathcal{B}}

\newcommand{\Beq}{\B_{\mathrm{eq}}}
\newcommand{\C}{\mathds{C}}

\newcommand{\Z}{\mathds{Z}}
\newcommand{\IP}{\mathds{P}}
\def\1{\ifmmode\mathrm{1\!l}\else\mbox{\(\mathrm{1\!l}\)}\fi}

\newcommand{\be}{\begin{equation}}
\newcommand{\ee}{\end{equation}}
\newcommand{\bes}{\begin{equation*}}
\newcommand{\ees}{\end{equation*}}

\newcommand{\Hom}{\operatorname{Hom}}
\newcommand{\End}{\operatorname{End}}
\newcommand{\modu}{\operatorname{mod}}
\def\LG{\mathcal{LG}}

\newcommand{\hmf}{\operatorname{hmf}}

\newcommand{\ev}{\operatorname{ev}}
\newcommand{\tev}{\widetilde{\operatorname{ev}}}
\newcommand{\coev}{\operatorname{coev}}
\newcommand{\tcoev}{\widetilde{\operatorname{coev}}}
\def\lra{\longrightarrow}
\def\lmt{\longmapsto}
\DeclareMathOperator{\tr}{tr}
\DeclareMathOperator{\str}{str}

\DeclareMathOperator{\Res}{Res}

\newcommand{\HIA}{\Hom(I,A)}
\newcommand{\ZA}{Z_A(\Hom(I,A))}
\newcommand{\gZA}{\!Z_A^\gamma(\Hom(I,A))}
\newcommand{\gA}{{}_{\gamma_A}A}
\newcommand{\Aginv}{A_{\gamma_A^{-1}}}
\newcommand{\picc}{\pi^{(\text{c,c})}_A}
\newcommand{\pirr}{\pi^{\text{RR}}_A}
 
\newcommand{\im}{\operatorname{im}}
\DeclareMathOperator*{\eq}{=}
\DeclareMathOperator*{\congscript}{\cong}
\newcommand{\specflow}{\mathcal U_{-\frac{1}{2},-\frac{1}{2}}}
\newcommand{\Hil}{\mathcal{H}}

\newcommand{\Hcc}{\mathcal{H}_{\text{(c,c)}}^A}
\newcommand{\Hrr}{\mathcal{H}_{\text{RR}}^A}
\newcommand{\HccnoA}{\mathcal{H}_{\text{(c,c)}}}
\newcommand{\HrrnoA}{\mathcal{H}_{\text{RR}}}

\def\Cong{C_g}
\def\Centg{N_g}
\def\alphaKK{\alpha^{\{K\}}}
\def\alphagKg{\alpha_g^{K(g)}}

\newcommand{\dX}{{}^\dagger\hspace{-1.8pt}X}
\newcommand{\dY}{{}^\dagger\hspace{-0.3pt}Y}

\newcommand{\dPhi}{{}^\dagger\hspace{-0.9pt}\Phi}

\newcommand\arxiv[2]      {\href{http://arXiv.org/abs/#1}{#2}}
\newcommand\doi[2]        {\href{http://dx.doi.org/#1}{#2}}
\newcommand\httpurl[2]    {\href{http://#1}{#2}}
\newcommand\bhref[3]   {\href{#1}{#2} \href{#1}{#3}}

\newcommand\bdoi[3]   {\doi{#1}{#2} \doi{#1}{#3}}
\newcommand\bhttpurl[3]   {\httpurl{#1}{#2} \httpurl{#1}{#3}}

\allowdisplaybreaks

\deffootnote[1em]{1em}{1em}
{\textsuperscript{\thefootnotemark}}

\theoremstyle{definition}
\newtheorem{definition}{Definition}
\newtheorem{proposition}[definition]{Proposition}

\newtheorem{lemma}[definition]{Lemma}
\newtheorem{corollary}[definition]{Corollary}
\newtheorem{remark}[definition]{Remark}

\newtheorem{example}[definition]{Example}

\numberwithin{equation}{section}
\numberwithin{definition}{section}
\numberwithin{figure}{section}

\newcommand\void[1]{}

\begin{document}

\title{Orbifolds and topological defects}

\author{Ilka Brunner$^*$ \quad Nils Carqueville$^\dagger$ \quad Daniel Plencner$^*$
\\[0.5cm]
 \normalsize{\tt \href{mailto:ilka.brunner@physik.uni-muenchen.de}{ilka.brunner@physik.uni-muenchen.de}} \\
  \normalsize{\tt \href{mailto:nils.carqueville@scgp.stonybrook.edu}{nils.carqueville@scgp.stonybrook.edu}} \\
  \normalsize{\tt \href{mailto:daniel.plencner@physik.uni-muenchen.de}{daniel.plencner@physik.uni-muenchen.de}}\\[0.1cm]
  {\normalsize\slshape $^*$Arnold Sommerfeld Center for Theoretical Physics, LMU M\"unchen}\\[-0.1cm]
  {\normalsize\slshape $^*$Excellence Cluster Universe, Garching}\\[-0.0cm]
  {\normalsize\slshape $^\dagger$Simons Center for Geometry and Physics, Stony Brook}\\[-0.1cm]
}
\date{}
\maketitle

\vspace{-11.8cm}
\hfill {\scriptsize LMU-ASC 49/13}

\vspace{11cm}

\begin{abstract}
We study orbifolds of two-dimensional topological field theories using defects. If the TFT arises as the twist of a superconformal field theory, 
we recover results on the Neveu-Schwarz and Ramond sectors of the orbifold theory as well as bulk-boundary correlators from a novel, universal perspective. 
This entails a structure somewhat weaker than ordinary TFT, which however still describes a sector of the underlying conformal theory.
The case of B-twisted Landau-Ginzburg models is discussed in detail, where we compute charge vectors and superpotential terms for B-type branes. 

Our construction also works in the absence of supersymmetry and for generalised ``orbifolds'' that need not arise from symmetry groups. In general this involves a natural appearance of Hochschild (co)homology in a 2-categorical setting, in which among other things we provide simple presentations of Serre functors and a further generalisation of the Cardy condition. 
\end{abstract}

\thispagestyle{empty}
\newpage

\tableofcontents

\section{Introduction}\label{sec:introduction}

Orbifolds of two-dimensional quantum field theories can naturally be understood in terms of defect operators. More precisely, for every orbifold group~$G$ there is an associated symmetry defect~$A_G$, and correlators in a $G$-orbifolded theory are computed in the unorbifolded theory by covering every worldsheet with a trivalent network of $A_G$-defects. This construction is under good control for rational conformal field theories and topological field theories, carried out in detail in~\cite{ffrs0909.5013} and~\cite{cr1210.6363}, respectively. Furthermore, these works show that orbifolds have a natural generalisation: one may abstract from the symmetry group~$G$ and replace the defect~$A_G$ by any other defect~$A$ such that the correlators do not depend on the choice of triangulation by the $A$-network. 

One of the purposes of the present paper is to apply these ideas in the context of theories with $\mathcal N=(2,2)$ supersymmetry. 
For this we develop the theory of equivariant completion of~\cite{cr1210.6363} further (without assuming any prior familiarity with it), augmenting it to encode information about both Neveu-Schwarz and Ramond sectors in the associated CFT, and initiating a study of its relation to fully extended TFT. Our motivation are Landau-Ginzburg models, whose conventional orbifold theory we recover from the universal defect perspective, leading to new and efficient techniques to compute orbifold bulk/boundary correlators and hence RR brane charges and superpotential terms. Furthermore, we explain how to compute defect actions on twisted bulk fields. On the other hand, our constructions are not limited to Landau-Ginzburg models but are naturally described in a general 2-categorical setting which applies equally well to A- or B-twisted sigma models, among others. At this level our results are best phrased as statements about Hochschild (co)homology, Serre functors, and related algebraic objects. 

In the bulk of the paper we separate the more concrete constructions of Landau-Ginzburg models in Section~\ref{sec:ordinaryLGorbs} from the development of the general theory in Section~\ref{sec:GOs}. The former also explains some of the intuition and applications and may appeal more to physicists, while the latter might be received more favourably by mathematicians. With very few exceptions, both sections can be read independently of one another. However, also in the case of orbifold theory a symbiosis provides more than the sum of two parts.

\medskip

As argued in~\cite{IntriligatorVafa1990} the $\mathcal N = (2,2)$ superconformal field theories associated via renormalisation group flow to affine Landau-Ginzburg models contain spectral flow as a field of the theory, leading to an isomorphism between Neveu-Schwarz and Ramond sectors.\footnote{Spectral flow is always an isomorphism on the level of representations of the $\mathcal N = (2,2)$ Virasoro algebra, but a choice of modular invariant for a CFT may render this isomorphism void.} 
In particular the state space of the topologically twisted Landau-Ginzburg model which describes the associated (c,c) fields may equivalently be viewed as describing RR ground states. 
This is no longer true for arbitrary $G$-orbifolds, and the question arises what information about the Ramond sector is still encoded in the unorbifolded topological theory with defects. 

To answer this question we recall that before orbifold projection the $g$-twisted bulk sector can be identified with the space of defect junction fields between the invisible defect~$I$ and a certain defect~${}_g I$ for every $g \in G$. The defect~${}_g I$ is obtained from~$I$ by ``twisting'' the latter by~$g$, so the action of~$g$ on a bulk field is implemented by wrapping~${}_g I$ around its insertion. The defect encoding the whole group~$G$ is simply $A_G = \bigoplus_{g\in G} {}_g I$, and a superposition of fields in the twisted sectors is an element~$\alpha$ of $\Hom(I,A_G)$, i.\,e.~a field living at the endpoint of the symmetry defect~$A_G$. 

We will see (in~\eqref{eq:piccAG} and~\eqref{eq:piRRAG} below) that there are essentially two natural ways of wrapping the defect~$A_G$ around a twist field~$\alpha$. More importantly we shall find that these two potentially different actions of~$G$ on $\Hom(I,A_G)$ precisely project to (c,c) fields and RR ground states, in complete agreement with~\cite{IntriligatorVafa1990}. In other words, these two sectors are obtained by two actions of one and the same defect~$A_G$. We also give a simple criterion for when these two spaces are isomorphic, extend the construction to the boundary and defect sectors, and apply our results to compute RR brane charges, complementing the work of~\cite{w0412274}. All this is done using defect constructions in the \textsl{unorbifolded} theory, where we have conceptual clarity and full computational control. 

As already mentioned our construction does not depend on the fact that~$A_G$ arises via a group~$G$. In fact any defect~$A$ that comes with the structure of a separable Frobenius algebra in a suitable 2-category can be used as the only input to our construction -- such defects take over and extend the role of symmetry groups. At this level of generality we will realise our generalised (c,c) fields and RR ground states as (new presentations of) Hochschild cohomology and homology of~$A$, respectively, and the physical picture will translate into various relations between them. We also investigate in detail to what extent such orbifolds give rise to the full structure of open/closed TFTs, in particular providing a simple proof of yet another generalisation of the Cardy condition. 

\medskip

The rest of this paper is organised as follows. We begin with a concise review of the standard approach to orbifold Landau-Ginzburg models in Section~\ref{subsec:conventionalorbi}. Section~\ref{subsec:adjointdefects} collects the necessary background on orientation reversal of defects in such theories; the associated adjunctions are one of the main ingredients in our construction, and we illustrate how to use them in practice. The centre of this part of the paper is Section~\ref{subsec:defectorbi} where we explain how to describe Neveu-Schwarz and Ramond sectors via defect actions and compute brane charges. Some technical details are divested to Appendix~\ref{TheAppendix}, while Section~\ref{subsec:CFTandEffectiveFieldTheory} offers a discussion of our construction from the perspective of $\mathcal N = (2,2)$ CFT. 

Section~\ref{sec:GOs} develops the theory on an abstract 2-categorical level, for which we gather the relevant background in Section~\ref{subsec:BicategoricalAlgebra}. This is followed by an extention of our earlier results in Section~\ref{subsec:gentwist}, where we construct generalised Neveu-Schwarz and Ramond sectors, discuss their relation and various characterisations, including Hochschild (co)homology. The algebra involved is fairly simple, and we find this point of view conceptually pleasing and very clear. Section~\ref{subsec:ocTFT} explains the relation to open/closed TFT which in particular entails a construction of Serre functors as simple twists and a generalisation of the Cardy condition. We also sketch a relation to extended TFT in the sense of Lurie, offering an interpretation of our construction in the absence of supersymmetry. Finally in Section~\ref{subsec:DefectsFunctoriality} we discuss the compatibility of defect fusion (or tensor products) and action on bulk fields (or Hochschild 
homology) in our generalised orbifolds.

\section{Ordinary Landau-Ginzburg orbifolds via defects}\label{sec:ordinaryLGorbs}
 
In this section we discuss Landau-Ginzburg orbifolds by finite symmetry groups. We start in Section~\ref{subsec:conventionalorbi} with a concise review of the conventional approach, following~\cite{IntriligatorVafa1990, add0401, br0712.0188, cr1006.5609}. In Section~\ref{subsec:adjointdefects} we collect relevant facts about orientation reversal and adjunction for defects, and we offer a first taste of their natural diagrammatic calculus. This allows us to recast conventional orbifolds solely in universal defect language, as we explain in Section~\ref{subsec:defectorbi} (and in much greater generality in Section~\ref{sec:GOs}). As a byproduct we derive a general expression for RR brane charges, and in Section~\ref{subsec:CFTandEffectiveFieldTheory} we discuss further relations to conformal field theory and effective four-dimensional field theory.

\subsection{Review of conventional description}\label{subsec:conventionalorbi} 

Landau-Ginzburg models are two-dimensional $\mathcal N=(2,2)$ supersymmetric quantum field theories that admit a Lagrangian description. A lot of the properties of the conformal fixed point of their renormalisation group flow are believed to be directly encoded in a single ingredient, to wit the F-term potential which we take to be a polynomial $W\in R = \C[x_1,\ldots,x_n]$. For example, if~$W$ is quasi-homogeneous then the Jacobian 
\be\label{eq:Jacobian}
H_e = R/(\partial W)
\ee 
is believed to be isomorphic to the ring of chiral primary (c,c) fields of the associated $\mathcal N=(2,2)$ conformal field theory. 
In this subsection we will assume some familiarity with basic facts about Landau-Ginzburg models and superconformal field theory as contained e.\,g.~in~\cite{lvw1989, mirror}, 
and we shall sometimes follow the custom in parts of the literature of blurring the lines between Landau-Ginzburg models and their conformal fixed points. 

One of the central notions in the study of $\mathcal N=(2,2)$ CFT is that of spectral flow $\mathcal U_{\theta,\bar\theta}$, which is a family of isomorphisms between representations of the superconformal algebra in the Neveu-Schwarz and Ramond sectors. In particular, $\specflow$ maps (c,c) fields to RR ground states. Hence these two spaces are isomorphic if $\specflow$ is a local operator, i.\,e.~if it corresponds to a vector that is actually present in the state space of the CFT under consideration. This is always true for unorbifolded Landau-Ginzburg models, and as a consequence in this case (c,c) fields and RR ground states can be respectively represented as linear combinations of
$$
\prod_{i=1}^n x_i^{l_i} | 0 \rangle 
\, , \quad
\prod_{i=1}^n x_i^{l_i} | 0 \rangle_{\text{RR}} 
\quad \text{mod } \partial W \, , 
$$
where $| 0 \rangle$ is the NS vacuum and $| 0 \rangle_{\text{RR}} = \specflow | 0 \rangle$.

\subsubsection*{Bulk sector} 

We now turn to the discussion of orbifolds by a finite symmetry group~$G$. For any two-dimensional quantum field theory its orbifold bulk space decomposes into a direct sum of the $G$-invariant (or ``untwisted'') part of the original bulk space and ``$g$-twisted sectors'', i.\,e.~$g$-invariant subspaces of additional representations of the symmetry algebra quotiented by~$G$. On the level of vertex operators~$\phi$, $g$-twisted fields are those with the property $\phi(\E^{2\pi\I} \sigma) = g \phi(\sigma)$. Since every $h\in G$ acts as an isomorphism and $h\phi(\E^{2\pi\I} \sigma) = hg \phi(\sigma) = hg h^{-1}h \phi(\sigma)$ we find that the $g$-twisted sector is isomorphic to the $hg h^{-1}$-twisted sector. Hence in order to avoid redundancies one may choose a set of representatives $\{ g \}$ for the conjugacy classes~$C_g$ in~$G$ and define the orbifold bulk space to be 
\be\label{eq:orbiprojPrime}
\Hil' = \bigoplus_{\{ g \}} P'_{g} H_{g} 
\quad\text{with}\quad 
P'_g = \frac{1}{|N_g|} \sum_{h\in N_g} h \, , 
\ee
where $H_g$ is the $g$-twisted sector, $N_g = \{ h\in G \, | \, gh=hg \}$ is the normaliser, and~$P'_g$ projects to $g$-invariant states. 

Alternatively, one may prefer not to choose representatives for conjugacy classes and consider the bigger space 
\be\label{eq:allTwistedSectorsBeforeOrbifoldProjection}
H = \bigoplus_{g\in G} H_g
\ee
instead. Applying the projector
\be\label{eq:orbifoldProjectorP}
P = \sum_{g\in G} P_g 
\, , \quad 
P_g = \frac{1}{|G|} \sum_{h\in G} h \big|_{H_g}
\ee
we obtain the orbifold bulk space 
\be\label{eq:orbiproj}
\Hil = PH 
\ee
which is isomorphic to~$\Hil'$ in~\eqref{eq:orbiprojPrime} if~$G$ acts unitarily. For convenience we provide the details of this isomorphism in Appendix~\ref{app:HHPrime}. 

\medskip

In the case of Landau-Ginzburg models the above construction was first carried out in~\cite{v1989, IntriligatorVafa1990}; we shall briefly review the results relevant for our purposes. For a given potential $W \in R = \C[x_1,\ldots,x_n]$ let~$G$ be a finite symmetry, i.\,e.~a subgroup of $\{ g\in \operatorname{Aut}(R) \, | \, g(W) = W \}$. Following~\cite{IntriligatorVafa1990} we further assume that the action of~$G$ is diagonal in the sense that there are rational numbers~$\Theta_i^g$ such that $g\in G$ sends~$x_i$ to $\E^{2\pi\I \Theta_i^g} x_i$. (In later sections we will drop this assumption.) Then the $g$-twisted sectors of (c,c) fields and RR ground states are respectively spanned by
\be\label{eq:gtwistedLG}
\prod_{\Theta_i^g \in \Z} x_i^{l_i} | 0 \rangle_{\text{(c,c)}}^g
\, , \quad
\prod_{\Theta_i^g \in \Z} x_i^{l_i} | 0 \rangle_{\text{RR}}^g
\quad \text{mod } \partial \overline W \, , 
\ee
where $| 0 \rangle_{\text{(c,c)}}^g$ is a state of minimal conformal weight in the $g$-twisted NS sector, $| 0 \rangle_{\text{RR}}^g = \specflow | 0 \rangle_{\text{(c,c)}}^g$, and~$\overline W$ is obtained from~$W$ be setting all variables to zero which are not $g$-invariant, i.\,e.~$\overline W$ depends only on those~$x_i$ with $\Theta_i^g \in \Z$. To determine the orbifold projected spaces $\HccnoA$ and $\HrrnoA$ of~\eqref{eq:orbiproj} one has to work out the actions $\rho_{\text{(c,c)}}$ and $\rho_{\text{RR}}$ of group elements on the states~\eqref{eq:gtwistedLG}. Using standard supersymmetry arguments and modular invariance of the Ramond sector partition function together with spectral flow, these actions were found to be 
\begin{align}
\rho_{\text{(c,c)}} (h) \prod_{\Theta_i^g \in \Z} x_i^{l_i} | 0 \rangle_{\text{(c,c)}}^g 
& = \det (h)^{-1} \det (h|_g) \, \E^{2\pi\I \sum_{\Theta_i^g \in \Z} \Theta_i^h l_i} \cdot \prod_{\Theta_i^g \in \Z} x_i^{l_i} | 0 \rangle_{\text{(c,c)}}^g \, , \label{eq:hongcc} \\ 
\rho_{\text{RR}} (h) \prod_{\Theta_i^g \in \Z} x_i^{l_i} | 0 \rangle_{\text{RR}}^g
& = \det (h|_g) \, \E^{2\pi\I \sum_{\Theta_i^g \in \Z} \Theta_i^h l_i} \cdot \prod_{\Theta_i^g \in \Z} x_i^{l_i} | 0 \rangle_{\text{RR}}^g \label{eq:hongRR} 
\end{align}
where $\det (h)$ (respectively $\det (h|_g)$) is the determinant of the matrix representing the action of~$h$ on all variables~$x_i$ (respectively only on $g$-invariant variables).\footnote{Note that it was found in~\cite{IntriligatorVafa1990} that a general $G$-orbifold involves a choice of discrete torsion $\varepsilon(h,g)$ and the data~$K_g$ of a $(-1)^{F_{\text{s}}}$-orbifold. If these are nontrivial then~\eqref{eq:hongcc}, \eqref{eq:hongRR} must be augmented to~\cite[(3.20), (3.21)]{IntriligatorVafa1990}. We shall discuss this along with the general theory of discrete torsion for defects in Landau-Ginzburg models in a separate paper~\cite{BCPdiscretetorsion}.}

Using the above actions the orbifold projectors~$P_g$ of~\eqref{eq:orbifoldProjectorP} can be straightforwardly applied to obtain the spaces $\HccnoA$, $\HrrnoA$ as the images of the projectors
\be\label{eq:ccAndRRprojectors}
\sum_{g\in G} \frac{1}{|G|} \sum_{h\in G} \rho_{\text{(c,c)}}(h)
\, , \quad
\sum_{g\in G} \frac{1}{|G|} \sum_{h\in G} \rho_{\text{RR}}(h) \, , 
\ee
respectively. In contrast to the case of the unorbifolded theory there is however no general reason for them to be isomorphic. Indeed, this is precisely the case if $| 0 \rangle_{\text{RR}} = \specflow | 0 \rangle$ survives the orbifold projection, and one finds
\be\label{eq:HccVersusHRR}
\HccnoA \cong \HrrnoA 
\quad\text{if}\quad 
\det (g) = 1 \quad\text{for all } g \in G \, . 
\ee
This condition is necessary for the conformal field theory to be a string vacuum, which is often presented as a Calabi-Yau compactification.

\subsubsection*{Boundary and defect sector} 

Again we start with a brief review of the unorbifolded case. Boundary conditions in B-twisted Landau-Ginzburg models are described by matrix factorisations~\cite{kl0210, bhls0305, l0312}. By a matrix factorisation of a potential $W \in R = \C[x_1,\ldots,x_n]$ we mean a finitely generated $\Z_2$-graded free $R$-module~$Q$ together with an odd $R$-linear operator~$d_Q$ such that $d_Q^2 = W \cdot 1_Q$. These form the objects of a category $\hmf(R,W)$ whose morphisms describe boundary operators. Given two matrix factoriations $Q,P$ a map in $\hmf(R,W)$ is a class in the BRST cohomology of the differential that sends $\Phi \in \Hom_R(Q,P)$ to $d_P \Phi - (-1)^{|\Phi|} \Phi d_Q$. 

Similarly, defects between Landau-Ginzburg models with potentials $W\in R$ and $V \in S = \C[z_1,\ldots,z_m]$ are described by matrix factorisations of $V-W$~\cite{br0707.0922}. The fusion of defects $X\in \hmf(S\otimes_\C R, V-W)$ and $Y \in \hmf(T\otimes_\C S, U-V)$ is given by
\be\label{eq:MFfusion}
Y \otimes X \equiv Y \otimes_S X 
\, , \quad
d_{Y\otimes X} = d_Y \otimes_S 1_X + 1_Y \otimes_S d_X \, . 
\ee
Note that while $(d_{Y\otimes X})^2 = (U-W) \cdot 1_{Y\otimes X}$, the module $Y\otimes X$ is of infinite rank over $T \otimes_\C R$ unless $S=\C$. However, $Y\otimes X$ is isomorphic to a finite-rank matrix factorisation in $\hmf(T\otimes_\C R, U-W)$ as explained in~\cite{dm1102.2957} (see~\cite{khovhompaper} for a computer implementation), so that fusion is well-defined after idempotent completion. 

The unit for the fusion product is the invisible defect. For a Landau-Ginzburg potential $W\in R = \C[x]$ it is given by the Koszul matrix factorisation
\be\label{eq:IW}
I_W = \bigwedge \Big( \bigoplus_{i=1}^n \C[x,x'] \theta_i \Big) 
\, , \quad 
d_{I_W} = \sum_{i=1}^n \left( (x_i - x'_i) \cdot \theta_i^* + \partial_{[i]}^{x,x'} W \cdot \theta_i \wedge (-) \right)
\ee
where by definition 
$$
\partial_{[i]}^{x,x'} W = \frac{W(x'_1,\ldots,x'_{i-1}, x_i, \ldots, x_n) - W(x'_1,\ldots,x'_{i}, x_{i+1}, \ldots, x_n)}{x_i - x'_i}
$$
is the $i$-th difference quotient, and the anticommuting variables~$\theta_i$ are generators of the exterior algebra~$I_W$. As expected the endomorphisms of the invisible defect in $\hmf(\C[x,x'], W(x)-W(x'))$ coincide with the space of bulk fields~\eqref{eq:Jacobian}. 

That the invisible defect is the fusion unit means that for every defect $X\in \hmf(S\otimes_\C R, V-W)$ there are natural isomorphisms 
\be\label{eq:lambdarho}
\lambda_X : I_V \otimes X \lra X 
\, , \quad
\rho_X : X\otimes I_W \lra X \, . 
\ee
Explicitly they are given by first projecting to the $\theta$-degree zero component, followed by the multiplication in~$S$ and~$R$, respectively. Of course as $(S\otimes_\C R)$-linear maps~$\lambda_X$ and~$\rho_X$ are not isomorphisms, but they are up to homotopy, i.\,e.~in $\hmf(S\otimes_\C R, V-W)$ where their inverses read~\cite{cm1208.1481}
\begin{align}
\lambda^{-1}_X(e_i) 
& = \sum_{l \geqslant 0} \sum_{a_1 < \cdots < a_l} \sum_j
\theta_{a_1} \ldots \theta_{a_l} 
\left\{ \partial^{z,z'}_{[a_l]}d_X \ldots \partial^{z,z'}_{[a_1]}d_X \right\}_{ji} \otimes e_{j}
\, , \label{eq:lambdainverse}
\\
\rho^{-1}_X(e_i) 
& = \sum_{l \geqslant 0} \sum_{a_1 < \cdots < a_l} \sum_j (-1)^{\binom{l}{2} + l|e_i|} e_j \otimes \left\{ \partial^{x,x'}_{[a_1]}d_X \ldots \partial^{x,x'}_{[a_l]}d_X \right\}_{ji} \theta_{a_1} \ldots \theta_{a_l} \label{eq:rhoinverse}
\end{align}
with $\{e_i\}$ a basis for~$X$. 

\medskip

Boundary conditions and defects of $G$-orbifolds are given by $G$-equivariant matrix factorisations~\cite{add0401, br0712.0188}. To describe them, let us write~${}_g (-)$ for the functor that twists the $R$-action on a left $R$-module~$X$ by twisting with the map $g^{-1} \in \operatorname{Aut}(R)$, i.\,e.~the action becomes $(r,m) \mapsto g^{-1}(r) \ldotp m$ for $r\in R$, $m\in X$; for example, applied to the unit~\eqref{eq:IW} it produces a matrix factorisation with the same underlying module and twisted differential
\be\label{eq:gIdefect}
d_{{}_g (I_W)} = \sum_{i=1}^n \left( (g(x_i) - x'_i) \cdot \theta_i^* + \partial_{[i]}^{x,x'} W(g(x),x') \cdot \theta_i \wedge (-) \right) . 
\ee
With this notation objects in the $G$-equivariant category $\hmf(R,W)^G$ are matrix factorisations $(Q,d_Q)$ of~$W$ together with a set of isomorphisms $\{ \gamma_g : {}_g Q \rightarrow Q \}_{g\in G}$ such that $\gamma_e = 1_Q$ and 
$$
%%%%%%%%%%%%%%%%%%%%%%%
\begin{tikzpicture}[
			     baseline=(current bounding box.base), 
			     >=stealth,
			     descr/.style={fill=white,inner sep=2.5pt}, 
			     normal line/.style={->}
			     ] 
\matrix (m) [matrix of math nodes, row sep=3em, column sep=2.5em, text height=1.5ex, text depth=0.25ex] {%
{}_{gh} Q && {}_g Q && Q \\
};
\path[font=\scriptsize] (m-1-1) edge[->] node[auto] {$ {}_g (\gamma_h) $} (m-1-3)
				  (m-1-3) edge[->] node[auto] {$ \gamma_g $} (m-1-5); 
\path[font=\scriptsize] 
				 (m-1-1) edge[->, bend right=35] node[auto] {$ \gamma_{gh} $} (m-1-5); 
\end{tikzpicture}
%%%%%%%%%%%%%%%%%%%%%%%
$$
commutes.\footnote{This is simply another way of expressing the familiar condition $\gamma(g) \, d_Q(g(x)) \, \gamma^{-1}(g) = d_Q(x)$ where~$\gamma$ is a representation of~$G$ on~$Q$.} 
Morphisms $\Phi:Q\rightarrow P$ in $\hmf(R,W)^G$ are maps of matrix factorisations that satisfy 
\be\label{eq:equivmorphs}
\Phi = \gamma_g^{(P)} \circ {}_g \Phi \circ (\gamma_g^{(Q)})^{-1} \, .
\ee

Defects between orbifold Landau-Ginzburg models with potentials~$W$ and~$V$ are defined similarly, namely as matrix factorisations of $W-V$ that are equivariant with respect to the group actions on either side. We do not need the details (for which we refer to~\cite{br0712.0188}), later in Remark~\ref{rem:AGmodules} we shall however offer a concise equivalent description. For the moment we only recall the fact that
\be\label{eq:earlyAG}
A_G = \bigoplus_{g\in G} {}_g (I_W)
\ee
is the invisible defect in the $G$-orbifold of a Landau-Ginzburg model with potential~$W$.

\subsection{Adjunctions between defects}\label{subsec:adjointdefects}

All the defect lines we consider come with an orientation. This orientation also induces a direction of defects as maps between theories. We will view a defect~$X$ of the form 
\be\label{eq:XfromWtoV}
%%%%%%%%%%%%%%%%%%%%%%
\begin{tikzpicture}[very thick,scale=0.7,color=blue!50!black, baseline=0cm]
\nicepalecolourscheme (0,-1.25) rectangle (2.0,1.25);
\nicehalfpalecolourscheme (-2.0,-1.25) rectangle (0,1.25);
\draw[line width=0] 
(0,1.25) node[line width=0pt] (A) {}
(0,-1.25) node[line width=0pt] (A2) {}; 
\draw[
	decoration={markings, mark=at position 0.55 with {\arrow{>}}}, postaction={decorate}
	]
 (0,-1.25) -- (0,1.25); 
 \draw[line width=0] 
(0,-1) node[line width=0pt, right] (Xbottom) {{\small $X$}}
(-1.4,0) node[line width=0pt, right] (Xbottom) {{\small $V$}}
(0.5,0) node[line width=0pt, right] (Xbottom) {{\small $W$}};
\end{tikzpicture}
%%%%%%%%%%%%%%%%%%%%%%
\ee
as mapping \textsl{from} theory~$W$ \textsl{to} theory~$V$. In particular this means that we read diagrams as the above from right to left, which also coincides with the direction of the fusion product: in the setting of the previous subsection we have 
$$
%%%%%%%%%%%%%%%%%%%%%%
\begin{tikzpicture}[very thick,scale=0.7,color=blue!50!black, baseline=0cm]
\nicenotpalecolourscheme (-4,-1.25) rectangle (-2.0,1.25);
\nicepalecolourscheme (0,-1.25) rectangle (2.0,1.25);
\nicehalfpalecolourscheme (-2.0,-1.25) rectangle (0,1.25);
\draw[line width=0] 
(0,1.25) node[line width=0pt] (A) {}
(0,-1.25) node[line width=0pt] (A2) {}; 
\draw[
	decoration={markings, mark=at position 0.55 with {\arrow{>}}}, postaction={decorate}
	]
 (0,-1.25) -- (0,1.25); 
\draw[
	decoration={markings, mark=at position 0.55 with {\arrow{>}}}, postaction={decorate}
	]
 (-2,-1.25) -- (-2,1.25); 
 \draw[line width=0] 
(0,-1) node[line width=0pt, right] (Xbottom) {{\small $X$}}
(-2,-1) node[line width=0pt, right] (Xbottom) {{\small $Y$}}
(-3.4,0) node[line width=0pt, right] (Xbottom) {{\small $U$}}
(-1.4,0) node[line width=0pt, right] (Xbottom) {{\small $V$}}
(0.5,0) node[line width=0pt, right] (Xbottom) {{\small $W$}};
\end{tikzpicture}
%%%%%%%%%%%%%%%%%%%%%%
= 
%%%%%%%%%%%%%%%%%%%%%%
\begin{tikzpicture}[very thick,scale=0.7,color=blue!50!black, baseline=0cm]
\nicepalecolourscheme (0,-1.25) rectangle (2.0,1.25);
\nicenotpalecolourscheme (-2.0,-1.25) rectangle (0,1.25);
\draw[line width=0] 
(0,1.25) node[line width=0pt] (A) {}
(0,-1.25) node[line width=0pt] (A2) {}; 
\draw[
	decoration={markings, mark=at position 0.55 with {\arrow{>}}}, postaction={decorate}
	]
 (0,-1.25) -- (0,1.25); 
 \draw[line width=0] 
(0,-1) node[line width=0pt, right] (Xbottom) {{\small $Y \otimes X$}}
(-1.4,0) node[line width=0pt, right] (Xbottom) {{\small $U$}}
(0.5,0) node[line width=0pt, right] (Xbottom) {{\small $W$}};
\end{tikzpicture}
%%%%%%%%%%%%%%%%%%%%%%
\, . 
$$
This identity holds locally on any worldsheet when we compute correlators. Of course when computing correlators we also want to consider field insertions. In general they are placed at defect junctions, and we write their operator product vertically, read from bottom to top. For example, the special fields~\eqref{eq:lambdarho} and their inverses describing the unit action are depicted as
$$
%%%%%%%%%%%%%%%%%%%%%%
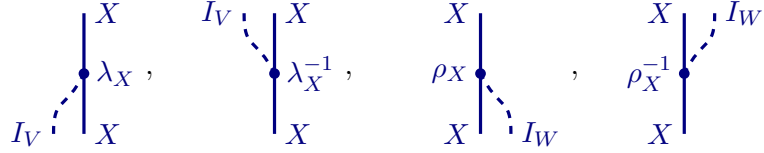
\begin{tikzpicture}[very thick,scale=1.0,color=blue!50!black, baseline=0cm]

\fill (0,0) circle (2.0pt) node[right] {{\small $\lambda_X$}};

\draw[dashed] (-0.4,-0.8) .. controls +(0,0.5) and +(-0.25,-0.5) .. (0,0);

\draw[line width=0] 
(0,-0.8) node[line width=0pt, right] (Xbottom) {{\small $X$}}
(0,0.8) node[line width=0pt, right] (Xtop) {{\small $X$}}
(-0.4,-0.8) node[line width=0pt, left] (I) {{\small $I_V$}};
\draw (0,-0.8) -- (0,0.8); 

\end{tikzpicture}
%%%%%%%%%%%%%%%%%%%%%% 
, \quad
%%%%%%%%%%%%%%%%%%%%%%
\begin{tikzpicture}[very thick,scale=1.0,color=blue!50!black, baseline=0cm]

\fill (0,0) circle (2.0pt) node[right] {{\small $\lambda_X^{-1}$}};

\draw[dashed] (-0.4,0.8) .. controls +(0,-0.5) and +(-0.25,0.5) .. (0,0);

\draw[line width=0] 
(0,-0.8) node[line width=0pt, right] (Xbottom) {{\small $X$}}
(0,0.8) node[line width=0pt, right] (Xtop) {{\small $X$}}
(-0.4,0.8) node[line width=0pt, left] (I) {{\small $I_V$}};
\draw (0,-0.8) -- (0,0.8); 

\end{tikzpicture}
%%%%%%%%%%%%%%%%%%%%%% 
, \qquad 
%%%%%%%%%%%%%%%%%%%%%%
\begin{tikzpicture}[very thick,scale=1.0,color=blue!50!black, baseline=0cm]

\fill (0,0) circle (2.0pt) node[left] {{\small $\rho_X$}};

\draw[dashed] (0.4,-0.8) .. controls +(0,0.5) and +(0.25,-0.5) .. (0,0);

\draw[line width=0] 
(0,-0.8) node[line width=0pt, left] (Xbottom) {{\small $X$}}
(0,0.8) node[line width=0pt, left] (Xtop) {{\small $X$}}
(0.4,-0.8) node[line width=0pt, right] (I) {{\small $I_W$}};
\draw (0,-0.8) -- (0,0.8); 

\end{tikzpicture}
%%%%%%%%%%%%%%%%%%%%%% 
, \quad 
%%%%%%%%%%%%%%%%%%%%%%
\begin{tikzpicture}[very thick,scale=1.0,color=blue!50!black, baseline=0cm]

\fill (0,0) circle (2.0pt) node[left] {{\small $\rho_X^{-1}$}};

\draw[dashed] (0.4,0.8) .. controls +(0,-0.5) and +(0.25,0.5) .. (0,0);

\draw[line width=0] 
(0,-0.8) node[line width=0pt, left] (Xbottom) {{\small $X$}}
(0,0.8) node[line width=0pt, left] (Xtop) {{\small $X$}}
(0.4,0.8) node[line width=0pt, right] (I) {{\small $I_W$}};
\draw (0,-0.8) -- (0,0.8); 

\end{tikzpicture}
%%%%%%%%%%%%%%%%%%%%%% 
$$
and for the product of $\Phi: X \rightarrow X'$ and $\Psi: X' \rightarrow X''$ we have 
$$
%%%%%%%%%%%%%%%%%%%%%%
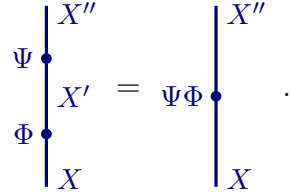
\begin{tikzpicture}[very thick,scale=1.0,color=blue!50!black, baseline=0cm]
\fill (0,-0.5) circle (2.0pt) node[left] {{\small $\Phi$}};
\fill (0,0.5) circle (2.0pt) node[left] {{\small $\Psi$}};
\draw[line width=0] 
(0,-1.1) node[line width=0pt, right] (X) {{\small $X$}}
(0,0) node[line width=0pt, right] (Xs) {{\small $X'$}}
(0,1.1) node[line width=0pt, right] (Xs) {{\small $X''$}};
\draw (0,-1.2) -- (0,1.2); 
\end{tikzpicture}
%%%%%%%%%%%%%%%%%%%%%% 
=
%%%%%%%%%%%%%%%%%%%%%%
\begin{tikzpicture}[very thick,scale=1.0,color=blue!50!black, baseline=0cm]
\fill (0,0) circle (2.0pt) node[left] {{\small $\Psi\Phi$}};
\draw[line width=0] 
(0,-1.1) node[line width=0pt, right] (X) {{\small $X$}}
(0,1.1) node[line width=0pt, right] (Xs) {{\small $X''$}};
\draw (0,-1.2) -- (0,1.2); 
\end{tikzpicture}
%%%%%%%%%%%%%%%%%%%%%% 
\, . 
$$

In Landau-Ginzburg models an upwards-oriented defect~$X$ as in~\eqref{eq:XfromWtoV} is described by a matrix factorisation of $V-W$, and one may ask what the operation of reversing its orientation means algebraically. This must be a matrix factorisation of $W-V$ which turns out to be the (right) adjoint 
$$
X^\dagger \cong X^\vee[n] 
\, , \quad 
d_{X^\vee} : \nu \lmt (-1)^{|\nu|+1} \nu \circ d_X
$$
where $\nu \in X^\vee = \Hom_{S\otimes_\C R}(X,S\otimes_\C R)$, $n$ is the number of variables in $R\ni W$, and~$[1]$ denotes the shift functor. 

A defect line may take a U-turn, 
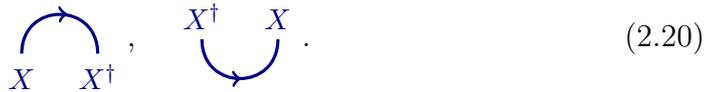
\be\label{eq:Uturns1}
%%%%%%%%%%%%%%%%%%%%%%
\begin{tikzpicture}[very thick,scale=1.0,color=blue!50!black, baseline=.4cm]
\draw[line width=0pt] 
(3,0) node[line width=0pt] (D) {{\small$X^\dagger$}}
(2,0) node[line width=0pt] (s) {{\small$X\vphantom{X^\dagger}$}}; 
\draw[redirected] (D) .. controls +(0,1) and +(0,1) .. (s);
\end{tikzpicture}
%%%%%%%%%%%%%%%%%%%%%%
, \quad
%%%%%%%%%%%%%%%%%%%%%%
\begin{tikzpicture}[very thick,scale=1.0,color=blue!50!black, baseline=-.4cm,rotate=180]
\draw[line width=0pt] 
(3,0) node[line width=0pt] (D) {{\small$X^\dagger$}}
(2,0) node[line width=0pt] (s) {{\small$X\vphantom{X^\dagger}$}}; 
\draw[directed] (D) .. controls +(0,1) and +(0,1) .. (s);
\end{tikzpicture}
%%%%%%%%%%%%%%%%%%%%%%
. 
\ee
In the spirit of our earlier diagrams we interpret these pictures as special junction fields
$$
\tev_X: X \otimes X^\dagger \lra I_V
\, , \quad
\tcoev_X: I_W \lra X^\dagger \otimes X
$$
mapping to and from the (now truly) invisible defect. In Section~\ref{subsec:BicategoricalAlgebra} we will put such maps into a broader context, for now we only recall that they were constructed in~\cite{cr1006.5609, cm1208.1481} and have the following explicit presentations: 
\begin{align}
\widetilde\ev_X( e_j p \otimes e_i^* ) 
& = \sum_{l \geqslant 0} \sum_{a_1 < \cdots < a_l} (-1)^{l + (n+1)|e_j| + n} \, \theta_{a_1} \ldots \theta_{a_l} \nonumber \\
& \qquad \cdot \Res \left[ \frac{ p\, \big\{ \partial^{z,z'}_{[a_l]} d_X  \ldots \partial^{z,z'}_{[a_1]} d_X \, \partial_{x_1}d_X\ldots \partial_{x_n}d_X \big\}_{ij} \, \operatorname{d}\!x }{\partial_{x_1}W, \ldots, \partial_{x_n} W} \right] , \label{eq:evaltilde} \\
\widetilde\coev_X( \bar\gamma ) 
& = \sum_{i,j} (-1)^{(\bar r+1)|e_j| + s_n} \left\{ \partial^{x,x'}_{[\bar b_{\bar r}]}d_X \ldots \partial^{x,x'}_{[\bar b_1]}d_X \right\}_{ji}  e_i^* \otimes e_j \label{eq:coevtilde}
\end{align}
where the integers~$\bar b_i$ are such that $\bar b_1 < \cdots < \bar b_{\bar r}$, and $\bar{\gamma} \theta_{\bar b_1} \ldots \theta_{\bar b_{\bar r}} = (-1)^{s_n} \theta_1 \ldots \theta_n$, and $p\in R = \C[x]$. 

Of course one may equally well consider the orientation-reversed variant of~\eqref{eq:Uturns1}, 
\be\label{eq:Uturns2}
\ev_X = \!
%%%%%%%%%%%%%%%%%%%%%%
\begin{tikzpicture}[very thick,scale=1.0,color=blue!50!black, baseline=.4cm]
\draw[line width=0pt] 
(3,0) node[line width=0pt] (D) {{\small$X\vphantom{X^\dagger}$}}
(2,0) node[line width=0pt] (s) {{\small$\dX$}}; 
\draw[directed] (D) .. controls +(0,1) and +(0,1) .. (s);
\end{tikzpicture}
%%%%%%%%%%%%%%%%%%%%%%
, \quad
\coev_X =  \!
%%%%%%%%%%%%%%%%%%%%%%
\begin{tikzpicture}[very thick,scale=1.0,color=blue!50!black, baseline=-.4cm,rotate=180]
\draw[line width=0pt] 
(3,0) node[line width=0pt] (D) {{\small$X\vphantom{X^\dagger}$}}
(2,0) node[line width=0pt] (s) {{\small$\dX$}}; 
\draw[redirected] (D) .. controls +(0,1) and +(0,1) .. (s);
\end{tikzpicture}
%%%%%%%%%%%%%%%%%%%%%% 
.
\ee
As is implicit in the notation, there is a slight subtlety here as one has to use the (left) adjoint $\dX \cong X^\vee[m]$ which is isomorphic to~$X^\dagger$ only if the number~$m$ of variables in $S\ni V$ has the same parity as~$n$. Explicitly, the maps~\eqref{eq:Uturns2} are given by (with $q\in S = \C[z]$): 
\begin{align}
\ev_X( e_i^* \otimes q e_j )
&  = \sum_{l \geqslant 0} \sum_{a_1 < \cdots < a_l} (-1)^{\binom{l}{2}+l|e_j|} \, \theta_{a_1} \ldots \theta_{a_l}  \nonumber \\
& \qquad \cdot \Res \left[ \frac{ q\, \big\{ \partial_{z_1}d_X\ldots \partial_{z_m}d_X \, \partial^{x,x'}_{[a_1]} d_X  \ldots \partial^{x,x'}_{[a_l]} d_X \big\}_{ij}  \, \operatorname{d}\!z }{\partial_{z_1}V, \ldots, \partial_{z_m} V} \right] , \label{eq:eval} \\
\coev_X(\gamma) 
& = \sum_{i,j} (-1)^{\binom{r+1}{2} + mr + s_m} \left\{ \partial^{z,z'}_{[b_1]}d_X \ldots \partial^{z,z'}_{[b_r]}d_X \right\}_{ij} e_{i} \otimes e_j^* \, . \label{eq:coev}
\end{align}

\medskip

As an application of the above formulas we use them to compute the disc correlator. For this let $Q\in \hmf(R,W)$ be a boundary condition, $\Psi: Q\rightarrow Q$ a boundary operator, and $\phi\in R/(\partial W)$ a bulk field (viewed as a defect field $I_W \rightarrow I_W$). In this case $\tcoev_Q$ is simply a map from the complex numbers to $Q^\dagger \otimes Q$, and $\ev_Q$ maps back to~$\C$. Thus from the general rule of reading all diagrams from bottom to top and from right to left we immediately find 
\be\label{eq:discCorrelator}
%%%%%%%%%%%%%%%%%%%%%%
\begin{tikzpicture}[very thick,scale=0.75,color=blue!50!black, baseline]
\nicepalecolourscheme (0,0) circle (1.5);
\fill (-0.8,-0.8) circle (0pt) node[white] {{\small$W$}};
\draw (0,0) circle (1.5);
\fill (-45:1.55) circle (0pt) node[right] {{\small$Q$}};
\draw[<-, very thick] (0.100,-1.5) -- (-0.101,-1.5) node[above] {}; 
\draw[<-, very thick] (-0.100,1.5) -- (0.101,1.5) node[below] {}; 
\fill (135:0) circle (2.5pt) node[left] {{\small$\phi$}};
\fill (0:1.5) circle (2.5pt) node[left] {{\small$\Psi$}};
\end{tikzpicture} 
%%%%%%%%%%%%%%%%%%%%%% 
=
\Res \left[ \frac{\phi(x) \str\big( \Psi \, \partial_{x_1} d_{Q}\ldots \partial_{x_n} d_{Q} \big) \operatorname{d}\! x}{\partial_{x_1} W, \ldots, \partial_{x_n} W} \right]
\ee
which is indeed the Kapustin-Li disc correlator of~\cite{kl0305, hl0404}. In the next subsection a slight variant of this argument will allow us to compute correlators and RR brane charges in Landau-Ginzburg orbifolds. 

\medskip

The expression~\eqref{eq:discCorrelator} is a special case of the (left) bulk action of the boundary condition~$Q$, viewed as a special defect between the trivial theory and~$W$. In general for a defect $X\in \hmf(S\otimes_\C R, V-W)$ its left and right bulk actions 
\be\label{eq:MotherDefectActionOnBulkFields}
\mathcal D_{\text{l}}^\Phi(X) (\psi) = 
%%%%%%%%%%%%%%%%%%%%%%
\begin{tikzpicture}[very thick,scale=0.75,color=blue!50!black, baseline]
\nicepalecolourscheme (0,0) circle (2.0);
\fill (-1.175,-1.175) circle (0pt) node[white] {{\small$W$}};
\nicecolourscheme (0,0) circle (1.25);
\fill (-0.65,-0.65) circle (0pt) node[white] {{\small$V$}};
\draw (0,0) circle (1.25);
\fill (45:1.3) circle (0pt) node[right] {{\small$X$}};
\draw[<-, very thick] (0.100,-1.25) -- (-0.101,-1.25) node[above] {}; 
\draw[<-, very thick] (-0.100,1.25) -- (0.101,1.25) node[below] {}; 
\fill (135:0) circle (2.5pt) node[left] {{\small$\psi$}};
\fill (0:1.25) circle (2.5pt) node[left] {{\small$\Phi$}};
\end{tikzpicture} 
%%%%%%%%%%%%%%%%%%%%%% 
\, , \quad
\mathcal D_{\text{r}}^\Phi (X) (\phi) = 
%%%%%%%%%%%%%%%%%%%%%%
\begin{tikzpicture}[very thick,scale=0.75,color=blue!50!black, baseline]
\nicecolourscheme (0,0) circle (2.0);
\fill (1.175,-1.175) circle (0pt) node[white] {{\small$V$}};
\nicepalecolourscheme (0,0) circle (1.25);
\fill (0.65,-0.65) circle (0pt) node[white] {{\small$W$}};
\draw (0,0) circle (1.25);
\fill (135:1.3) circle (0pt) node[left] {{\small$X$}};
\draw[->, very thick] (0.100,-1.25) -- (-0.101,-1.25) node[above] {}; 
\draw[->, very thick] (-0.100,1.25) -- (0.101,1.25) node[below] {}; 
\fill (135:0) circle (2.5pt) node[right] {{\small$\phi$}};
\fill (180:1.25) circle (2.5pt) node[right] {{\small$\Phi$}};
\end{tikzpicture} 
%%%%%%%%%%%%%%%%%%%%%% 
\ee 
are given by~\cite[Prop.\,8.2]{cm1208.1481}
\begin{align*}
\mathcal D_{\text{l}}^\Phi(X) (\psi)
& = 
(-1)^{\binom{n+1}{2}}\Res \left[ \frac{\psi(z) \str\big( \Phi \, \partial_{x_1} d_{X}\ldots \partial_{x_n} d_{X} \,  \partial_{z_1} d_{X}\ldots \partial_{z_m} d_{X}\big) \operatorname{d}\! z}{\partial_{z_1} V, \ldots, \partial_{z_m} V} \right] , 
\\
\mathcal D_{\text{r}}^\Phi (X) (\phi)
& =
(-1)^{\binom{m+1}{2}}\Res \left[ \frac{\phi(x) \str\big( \Phi \, \partial_{x_1} d_{X}\ldots \partial_{x_n} d_{X} \,  \partial_{z_1} d_{X}\ldots \partial_{z_m} d_{X}\big) \operatorname{d}\! x}{\partial_{x_1} W, \ldots, \partial_{x_n} W} \right] .
\end{align*}
One checks that defect actions have the expected properties~\cite[Prop.\,8.5]{cm1208.1481}. In particular they are compatible with fusion, e.\,g.~$\mathcal D^\Psi_{\text{r}}(Y) \circ \mathcal D^\Phi_{\text{r}}(X) = \mathcal D^{\Psi \otimes \Phi}_{\text{r}}(Y\otimes X)$, and the invisible defect acts as the identity on bulk fields.

\subsection{Defect description of orbifolds}\label{subsec:defectorbi}

Let $W\in R= \C[x_1,\ldots,x_n]$ be a Landau-Ginzburg potential with invisible defect $I \equiv I_W$ and finite symmetry group~$G$. In this section we will begin to study $G$-orbifolds from a perspective that emphasises the role of the natural symmetry defect
\be\label{eq:AG}
A_G = \bigoplus_{g\in G} {}_g I
\ee
which we already encountered in~\eqref{eq:earlyAG}. We will see how the $g$-twisted invisible defect~${}_g I$ naturally implements the action of~$g$ in the Landau-Ginzburg model. Everything we reviewed in Sections~\ref{subsec:conventionalorbi} and~\ref{subsec:adjointdefects} (and more) can be recovered by merely thinking this basic idea through in a rather pedestrian way. Later in Section~\ref{sec:GOs} we shall explain the more conceptual aspects of this construction, along with a generalisation to ``orbifolds'' that do not need a group as input.

\subsubsection*{Bulk sector} 

We claim that in defect language the bulk space before orbifold projection~\eqref{eq:allTwistedSectorsBeforeOrbifoldProjection} is given by
\be\label{eq:HomIAG}
H = \Hom(I, A_G) = \bigoplus_{g\in G} \Hom(I,{}_g I)
\ee
where the Hom spaces are those in the category $\hmf(\C[x,x'], W(x)-W(x'))$. This means that we identify bulk fields in the $g$-twisted sector with defect fields~$\alpha$ between the invisible defect~$I$ and its $g$-twist ${}_g I$: 
$$
\Hom(I,{}_g I) \ni \alpha = 
%%%%%%%%%%%%%%%%%%%%%%
\begin{tikzpicture}[very thick,scale=0.75,color=green!65!black, baseline]
\fill (0,0) circle (2.5pt) node[right] (D) {{\small $\alpha$}};
\fill (0,0.7) circle (0) node[left] (u) {{\small ${}_g I$}};
\fill (0,-0.7) circle (0) node[left] (d) {{\small $I$}};
\draw (0,0) -- (0,0.7); 
\draw[dashed] (0,0) -- (0,-0.7); 
\end{tikzpicture} 
%%%%%%%%%%%%%%%%%%%%%% 
\equiv  
%%%%%%%%%%%%%%%%%%%%%%
\begin{tikzpicture}[very thick,scale=0.75,color=green!65!black, baseline]
\fill (0,-0.5) circle (2.5pt) node[right] (D) {{\small $\alpha$}};
\draw (0,-0.5) -- (0,0.6); 
\end{tikzpicture} 
%%%%%%%%%%%%%%%%%%%%%%  
. 
$$
That the component spaces $\Hom(I,{}_g I)$ in~\eqref{eq:HomIAG} are indeed isomorphic to the $g$-twisted sectors before orbifold projection spanned by~\eqref{eq:gtwistedLG} in the conventional description was shown in~\cite[App.\,A]{br0707.0922}. 

To obtain the bulk space after orbifold projection we have to understand the projector~\eqref{eq:orbifoldProjectorP} in terms of the defects~${}_h I$ and their action on $\alpha \in \Hom(I,{}_g I)$. In the previous section we saw that in the unorbifolded theory a defect acts on a bulk field simply by forming a loop around its insertion point, see~\eqref{eq:MotherDefectActionOnBulkFields}. Since~$\alpha$ has an outgoing nontrivial defect this is no longer an option in our current situation. However, there are two natural candidates for the process of wrapping the defect~$A_G$ around a field $\alpha \in \Hom(I,{}_g I) \subset \Hom(I,A_G)$, giving rise to the following two projectors: 
\begin{align}
\pi^{\text{(c,c)}} : 
%%%%%%%%%%%%%%%%%%%%%%
\begin{tikzpicture}[very thick,scale=0.75,color=green!50!black, baseline]
\fill (0,-0.5) circle (2.5pt) node[left] (D) {{\small $\alpha$}};
\draw (0,-0.5) -- (0,0.6); 
\end{tikzpicture} 
%%%%%%%%%%%%%%%%%%%%%% 
&
\lmt 
%%%%%%%%%%%%%%%%%%%%%%
\begin{tikzpicture}[very thick,scale=0.75,color=green!50!black, baseline]
\draw (0,0) -- (0,1.25);
\fill (0,0) circle (2.5pt) node[left] {{\small $\alpha$}};
\draw (0,0.8) .. controls +(-0.9,-0.3) and +(-0.9,0) .. (0,-0.8);
\draw (0,-0.8) .. controls +(0.9,0) and +(0.7,-0.1) .. (0,0.4);
\fill (0,-0.8) circle (2.5pt) node {};
\fill (0,0.4) circle (2.5pt) node {};
\fill (0,0.8) circle (2.5pt) node {};
\draw (0,-1.2) node[Odot] (unit) {};
\draw (0,-0.8) -- (unit);
\end{tikzpicture}
%%%%%%%%%%%%%%%%%%%%%%
=
\frac{1}{|G|} 
\sum_{h\in G}
%%%%%%%%%%%%%%%%%%%%%%
\begin{tikzpicture}[very thick,scale=0.95,color=green!65!black, baseline]
\draw (0,-0.3) -- (0,1.25);
\fill (0,-0.3) circle (2.25pt) node[left] {{\small $\alpha$}};
\draw (0,0.8) .. controls +(-0.9,-0.3) and +(-0.9,0) .. (0,-0.8);
\draw (0,-0.8) .. controls +(0.9,0) and +(0.7,-0.1) .. (0,0.4);
\fill (0,-0.8) circle (2.25pt) node {};
\fill (0,0.4) circle (2.25pt) node {};
\fill (0,0.8) circle (2.25pt) node {};
\draw (0,-1.2) node[Odot] (unit) {};
\draw (0,-0.8) -- (unit);
\fill (-0.55,0.1) circle (0pt) node[left] {{\tiny $h\vphantom{h^{-1}}$}};
\fill (0.1,0.1) circle (0pt) node[left] {{\tiny $g\vphantom{h^{-1}}$}};
\fill (1.3,0.1) circle (0pt) node[left] {{\tiny $h^{-1}$}};
\fill (0.9,0.65) circle (0pt) node[left] {{\tiny $gh^{-1}$}};
\fill (1.05,1.1) circle (0pt) node[left] {{\tiny $hgh^{-1}$}};
\end{tikzpicture}
%%%%%%%%%%%%%%%%%%%%%%
\, , \label{eq:piccAG} \\ 
\pi^{\text{RR}} : 
%%%%%%%%%%%%%%%%%%%%%%
\begin{tikzpicture}[very thick,scale=0.75,color=green!50!black, baseline]
\fill (0,-0.5) circle (2.5pt) node[left] (D) {{\small $\alpha$}};
\draw (0,-0.5) -- (0,0.6); 
\end{tikzpicture} 
%%%%%%%%%%%%%%%%%%%%%% 
& 
\lmt 
%%%%%%%%%%%%%%%%%%%%%%
\begin{tikzpicture}[very thick,scale=0.75,color=green!50!black, baseline=0.2cm]
\draw (0,0) -- (0,1.3);
\fill (0,0) circle (2.5pt) node[left] {{\small$\alpha$}};
\draw (0,0.8) .. controls +(-0.9,-0.3) and +(-0.9,0) .. (0,-0.8);
\draw[
	decoration={markings, mark=at position 0.83 with {\arrow{<}}}, postaction={decorate}
	]
	 (0,-0.8) .. controls +(0.9,0) and +(0.9,0.8) .. (0,0.4);
\draw[->] (0.01,-0.8) -- (-0.01,-0.8);
\fill (0,0.4) circle (2.5pt) node {};
\fill (0,0.8) circle (2.5pt) node {};
\end{tikzpicture}
%%%%%%%%%%%%%%%%%%%%%%
=
\frac{1}{|G|} 
\sum_{h\in G}
%%%%%%%%%%%%%%%%%%%%%%
\begin{tikzpicture}[very thick,scale=0.95,color=green!65!black, baseline=0.2cm]
\draw (0,-0.3) -- (0,1.3);
\fill (0,-0.3) circle (2.25pt) node[left] {{\small$\alpha$}};
\draw (0,0.8) .. controls +(-0.9,-0.3) and +(-0.9,0) .. (0,-0.8);
\draw[
	decoration={markings, mark=at position 0.8 with {\arrow{<}}}, postaction={decorate}
	]
	 (0,-0.8) .. controls +(0.9,0) and +(0.9,0.8) .. (0,0.1);
\draw[->] (0.01,-0.8) -- (-0.01,-0.8);
\fill (0,0.1) circle (2.25pt) node {};
\fill (0,0.8) circle (2.25pt) node {};
\fill (-0.55,-0.05) circle (0pt) node[left] {{\tiny $h\vphantom{h^{-1}}$}};
\fill (0.1,-0.05) circle (0pt) node[left] {{\tiny $g\vphantom{h^{-1}}$}};
\fill (1.4,-0.05) circle (0pt) node[left] {};
\fill (0.9,0.55) circle (0pt) node[left] {{\tiny $gh^{-1}$}};
\fill (1.05,1.1) circle (0pt) node[left] {{\tiny $hgh^{-1}$}};
\end{tikzpicture}
%%%%%%%%%%%%%%%%%%%%%%
. \label{eq:piRRAG}
\end{align}
In these diagrams 
$
%%%%%%%%%%%%%%%%%%%%%%
\begin{tikzpicture}[very thick,scale=0.4,color=green!50!black, baseline=-0.12cm]
\draw (0,-0.5) node[Odot] (D) {}; 
\draw (D) -- (0,0.6); 
\end{tikzpicture} 
%%%%%%%%%%%%%%%%%%%%%% 
$ 
is the embedding $I = {}_e I \rightarrow A_G$, the trivalent vertices are given by the twisted left action of the invisible defect and its inverse~\eqref{eq:lambdainverse}, 
$$
%%%%%%%%%%%%%%%%%%%%%%
\begin{tikzpicture}[very thick,scale=0.75,color=green!50!black, baseline=0.4cm]
\draw[-dot-] (3,0) .. controls +(0,1) and +(0,1) .. (2,0);
\draw (2.5,0.75) -- (2.5,1.5); 
\fill (2,0) circle (0pt) node[left] (D) {{\small $g\vphantom{gh}$}};
\fill (3,0) circle (0pt) node[right] (D) {{\small $h\vphantom{gh}$}};
\fill (2.5,1.5) circle (0pt) node[right] (D) {{\small $gh$}};
\end{tikzpicture} 
%%%%%%%%%%%%%%%%%%%%%%  
= 
{}_g (\lambda_{{}_h I}) 
\, , \quad
%%%%%%%%%%%%%%%%%%%%%%
\begin{tikzpicture}[very thick,scale=0.75,color=green!50!black, baseline=-0.6cm, rotate=180]
\draw[-dot-] (3,0) .. controls +(0,1) and +(0,1) .. (2,0);
\draw (2.5,0.75) -- (2.5,1.5); 
\fill (2,0) circle (0pt) node[right] (D) {{\small $g\vphantom{gh}$}};
\fill (3,0) circle (0pt) node[left] (D) {{\small $h\vphantom{gh}$}};
\fill (2.5,1.5) circle (0pt) node[right] (D) {{\small $hg$}};
\end{tikzpicture} 
%%%%%%%%%%%%%%%%%%%%%%  
= 
{}_h (\lambda^{-1}_{{}_g I})
\, , 
$$
while
$
%%%%%%%%%%%%%%%%%%%%%%
\begin{tikzpicture}[very thick,scale=0.5,color=green!50!black, baseline=0.16cm]
\draw[line width=0pt] 
(3,0) node[line width=0pt] (D) {}
(2,0) node[line width=0pt] (s) {}; 
\draw[redirectedgreen] (D) .. controls +(0,1) and +(0,1) .. (s);
\end{tikzpicture} 
%%%%%%%%%%%%%%%%%%%%%% 
$ 
and
$
%%%%%%%%%%%%%%%%%%%%%%
\begin{tikzpicture}[very thick,scale=0.5,color=green!50!black, baseline=-0.38cm]
\draw[line width=0pt] 
(3,0) node[line width=0pt] (D) {}
(2,0) node[line width=0pt] (s) {}; 
\draw[directedgreen] (D) .. controls +(0,-1) and +(0,-1) .. (s);
\end{tikzpicture} 
%%%%%%%%%%%%%%%%%%%%%% 
$ 
are the adjunction maps~\eqref{eq:evaltilde}, \eqref{eq:coev}. 

The operators~\eqref{eq:piccAG} and~\eqref{eq:piRRAG} are our proposals for the projectors~\eqref{eq:ccAndRRprojectors} to (c,c) fields and RR ground states in the $G$-orbifold. That they are indeed projectors is a direct consequence of our more general discussion in Section~\ref{subsec:gentwist} where we will also recover the relation~\eqref{eq:HccVersusHRR} as a special case. 

To show that the images of $\pi^{\text{(c,c)}}$ and $\pi^{\text{RR}}$ are indeed isomorphic to the spaces $\HccnoA$ and $\HrrnoA$ in the conventional description of Section~\ref{subsec:conventionalorbi} one has to check that the above defect actions reproduce the formulas~\eqref{eq:hongcc} and~\eqref{eq:hongRR}. This is a straightforward exercise which we carry out for Landau-Ginzburg models of Fermat type in Appendix~\ref{app:RRprojDetails}.

\subsubsection*{Boundary and defect sector} 

In Section~\ref{subsec:conventionalorbi} we reviewed that boundary and defect conditions in the orbifold theory are described by $G$-equivariant matrix factorisations. In terms of the symmetry defect~\eqref{eq:AG} this translates into these two sectors being given by modules and bimodules over~$A_G$, viewed as a certain algebra~\cite[Sect.\,7.1]{cr1210.6363}. We shall explain this construction in detail in Section~\ref{sec:GOs}; see in particular Remark~\ref{rem:twistedpivotal}. 

For now we focus on the implications of our identification of $g$-twisted bulk fields with elements in $\Hom(I, {}_g I)$. This perspective in particular provides us with a general method to compute disc correlators in the orbifold theory. Indeed, for a $G$-equivariant matrix factorisation $(Q,d_Q,\{ \gamma_g \})$ of~$W$, the disc correlator of a bulk field $\alpha \in \Hom(I,{}_g I)$ and a boundary field $\Psi \in \hmf(R,W)^G$ is
\be\label{eq:orbiDisc}
%%%%%%%%%%%%%%%%%%%%%%
\begin{tikzpicture}[very thick,scale=0.75,color=blue!50!black, baseline]
\nicepalecolourscheme (0,0) circle (1.5);
\fill (-0.8,-0.8) circle (0pt) node[white] {{\small$W$}};
\draw (0,0) circle (1.5);
\fill (-45:1.55) circle (0pt) node[right] {{\small$Q$}};
\draw[<-, very thick] (0.100,-1.5) -- (-0.101,-1.5) node[above] {}; 
\draw[<-, very thick] (-0.100,1.5) -- (0.101,1.5) node[below] {}; 
\fill[color=green!65!black] (135:0) circle (2.5pt) node[left] {{\small$\alpha$}};
\fill (0:1.5) circle (2.5pt) node[left] {{\small$\Psi$}};
\draw[color=green!65!black] (0,0) .. controls +(0,0.6) and +(-0.4,-0.4) .. (45:1.5);
\fill[color=green!65!black] (45:1.5) circle (2.5pt) node[right] {{\small$\gamma_g\circ {}_g(\lambda_Q)$}};
\end{tikzpicture} 
%%%%%%%%%%%%%%%%%%%%%% 
\ee
with the evaluation and coevaluation maps as in Section~\ref{subsec:adjointdefects}. The expression~\eqref{eq:orbiDisc} is the natural generalisation of the disc correlator in the unorbifolded theory~\eqref{eq:discCorrelator}, which is the special case $\alpha = \phi \in \Hom(I,I)$. Note that we are forced to include all the canonical maps that come with the data of an equivariant boundary condition: first the map ${}_g (\lambda_Q)$ fuses the defect~${}_g I$ with the boundary~$Q$ to produce~${}_g Q$, and then $\gamma_g : {}_g Q \cong Q$ has to be applied so that the boundary can be consistently labelled by~$Q$. We also emphasise that the correlator~\eqref{eq:orbiDisc} can straightforwardly be evaluated in any given model, thanks to the explicit expressions~\eqref{eq:lambdainverse}, \eqref{eq:coevtilde}, and \eqref{eq:eval} for its constituents. 

In our disc correlator~\eqref{eq:orbiDisc}~$\alpha$ may be either a (c,c) field or an RR ground state. In particular, we can consider the special case where $\alpha = \phi_g$ is in the $g$-twisted RR sector and~$\Psi$ is the identity field. Then~\eqref{eq:orbiDisc} is the one-point correlator $\langle \phi_g \rangle_Q$ which computes the $g$-th RR charge of the brane~$Q$. We emphasise that in the defect approach the expression~\eqref{eq:orbiDisc} for correlators and brane charges follows most naturally and can be considered a derivation from first principles. 

In~\cite[Sect.\,5]{w0412274} it was proposed that 
\be\label{eq:WalcherProposal}
\langle \phi_g \rangle_Q = 
\Res \left[ \frac{\varphi_g \str\big( \gamma \, \partial_{\overline x_1} d_{\overline Q} \ldots \partial_{\overline x_r} d_{\overline Q} \big) \operatorname{d}\! \overline x}{\partial_{\overline x_1} \overline W, \ldots, \partial_{\overline x_r} \overline W} \right]
\ee
where $\{ \overline x_i \}_{i\in \{1,\ldots, r\}}$ are the $g$-invariant variables, $\overline W, d_{\overline Q}$ are obtained from $W, d_Q$ by setting the remaining variables to zero, and $\varphi_g \in \C[\overline x_1, \ldots, \overline x_r]/(\partial_{\overline x_i} \overline W)$ represents~$\phi_g$ as in~\eqref{eq:gtwistedLG}. This proposal was checked by comparing it with known values of RR charges in the case of minimal models and their tensor products. Furthermore in the special case of no untwisted variables ($r=0$) it was shown to be compatible with the orbifold Cardy condition. 

Below in Section~\ref{subsec:ocTFT} we will prove the Cardy condition in the general case (extending earlier results of~\cite{pv1002.2116} and~\cite{cr1210.6363}), and we will see that our proposal is inherently compatible with it. In addition in Appendix~\ref{app:RRchargeComputations} we explicitly compute $\langle \phi_g \rangle_Q$ in the case of tensor products of minimal model branes, generalised permutation branes~\cite{bg0503207, cfg0511078}, and linear matrix factorisations~\cite{err0508}. We find our results to be consistent with the proposal~\eqref{eq:WalcherProposal} in all these cases. 

\medskip

We close this section with a discussion of the defect sector. Since a ``defect description of the defect sector'' is rather self-referencing we can be brief. Furthermore, all the facts about defects in ordinary orbifold theories are special cases of the general results to be discussed in the more conceptual context of Section~\ref{subsec:DefectsFunctoriality}, where we will also have developed a vocabulary that makes the proofs much simpler than the direct computations we would have to resort to at the present stage. 

As mentioned at the end of Section~\ref{subsec:conventionalorbi}, a defect between two orbifold Landau-Ginzburg models is a suitably equivariant matrix factorisation of the difference of the potentials, in particular the invisible defect in the $G$-orbifold is the symmetry defect~$A_G$ of~\eqref{eq:earlyAG}. For details we refer to~\cite[Sect.\,4.1]{br0712.0188} or the equivalent description of Section~\ref{sec:GOs}. There we will also explain fusion in the orbifold theory as well as the defect actions~\eqref{eq:leftdefectorbiaction}, \eqref{eq:rightdefectorbiaction} on bulk fields $\alpha \in \Hom(I,A_G)$ (corresponding to the case~\eqref{eq:MotherDefectActionOnBulkFields} in the unorbifolded theory). Furthermore we will again find compatibility between defect actions and both fusion and bulk correlators.

\subsection{Relation to conformal and effective field theory}\label{subsec:CFTandEffectiveFieldTheory}

Many two-dimensional $\mathcal N=(2,2)$ supersymmetric field theories, including all such superconformal theories, can be topologically twisted to obtain TFTs, see e.\,g.~\cite{mirror}. There are two possible twists, the topological A- or B-twist, and superconformal A- or B-boundary conditions and defects are compatible with the respective twists. After twisting they become the relevant boundary conditions and defects of the topological theory. In the bulk sector, the fields surviving the topological twist are the (a,c) and (c,c) chiral rings for the A- and B-twist, respectively. These fields have nonsingular operator product expansions already on the level of the full CFT, as well as regular behaviour when brought to a compatible boundary or defect line. Quite generally, it is of interest to do computations in the simpler topological theory, if one can argue that the results still have an interpretation on the level of the conformal field theory.

In this paper we used the B-twisted Landau-Ginzburg model with defects to compute correlators in superconformal orbifold theories. Our approach in particular has the advantage that it is not necessary to formulate a consistent topological \textsl{orbifold} theory (in the sense of \cite{l0010269,ms0609042}). Instead a somewhat weaker structure suffices, where only a consistent unorbifolded parent theory is required. Also, we can compute disc one-point functions that are not part of the topological structure.

Let us review the necessary background in somewhat more detail. Starting with an $\mathcal N=(2,2)$ CFT, the topologically twisted theory does not always satisfy all the axioms of a TFT; in particular, two-point correlators of the topological theory can be degenerate. The reason for this is simple: on the one hand, we can always perform a topological A- or B-twist on the level of the symmetry algebra. On the other hand, topological correlators are obtained from CFT correlators by inserting a background charge. The two-point correlator of two elements of the chiral ring in the topological theory translates to a two-point function of RR ground states in the CFT. The space of ground states in a consistent conformal field theory is always equipped with a nondegenerate pairing, and this pairing translates to a nondegenerate pairing in the topological theory -- provided that the spectral flow operator relating chiral ring elements to RR ground states is a symmetry of the spectrum of the CFT. 

Viewed from a slightly different perspective, nondegeneracy of the pairing would follow from a type of operator-state correspondence, meaning a one-to-one correspondence between ground states and chiral fields. Such a correspondence does not always exist starting from arbitrary $(2,2)$ theories. Relevant for us is the fact that the correspondence always exists for B-twisted unorbifolded Landau-Ginzburg theories, but may not be present for their orbifolds. Within the class of all Landau-Ginzburg orbifolds, we can single out the special case where spectral flow survives the orbifold projection. This subclass contains all Landau-Ginzburg models corresponding to Calabi-Yau compactifications. In other cases, there is generically no consistent topological orbifold theory, and the nondegeneracy of two-point correlators will be violated both in the bulk and in the boundary or defect sector.

However, as we will describe in Section~\ref{subsec:ocTFT} there is still a somewhat weaker structure which takes into account that on the level of the twisted orbifold theory, there is no operator-state correspondence in the strict sense. 
In this weaker structure, the two-point correlator is again nondegenerate. A key fact underlying the analysis in Section~\ref{subsec:ocTFT} is that on the level of the full CFT two-point functions between RR ground states are nondegenerate; hence we can achieve nondegeneracy also on the level of the twisted theory by identifying the full set of RR ground states. As we will explain in detail, the resulting structure contains nondegenerate two-point correlators in an RR-like sector of the topological theory, whereas consistent products between operators are formulated in an NSNS-like sector. Unlike for ordinary open/closed TFTs, there is no operator-state correspondence that relates the two.

On the other hand, orbifolds on the level of the \textsl{superconformal} field theory can be consistent even if the spectral flow operator does not survive the orbifold projection.
The natural question is therefore whether our constructions can indeed yield correlation functions in a subsector of a consistent conformal orbifold theory.

This can be argued on the level of the unorbifolded parent theory. Of course, defects can also be used to formulate orbifold theories on the level of the full CFT, as described in~\cite{ffrs0909.5013} for rational conformal field theory. The superposition~$A_G$ of symmetry defects~${}_g I$ preserves the full symmetry algebra realised in the bulk, which in our case is the $\mathcal N=(2,2)$ Virasoro algebra. $A_G$ is topological already on the level of the full CFT, so there is a direct relation to the topologically twisted theory. In particular, twisted sector fields arise as defect changing operators between symmetry defects and the identity defect. If the spectrum of the original theory respects a certain spectral flow symmetry, this is inherited by the spectrum between symmetry defects. If now the unorbifolded parent CFT is B-twistable, we can map any correlation function of chiral ring elements, also those involving defect changing operators between symmetry defects, to a correlator in the associated 
topological theory and compute it there.
Our procedure to obtain correlation functions of an orbifold CFT is therefore to first interpret them as correlators in the unorbifolded CFT using defects, then do the topological twist and carry out the actual computation on the level of TFT with defects.

The simplest example where this procedure can be made concrete are $\mathcal N=(2,2)$ minimal models at level~$k$, where it is well-known that the diagonal and charge-conjugated modular invariants are related by a $\Z_{k+2}$-orbifold \cite{Gepner:1986hr}. The diagonal minimal model can consistently be B-twisted, giving rise to a TFT with nondegenerate pairing. The mirror theory with a charge conjugation modular invariant can be consistently A-twisted, but not B-twisted. Nonetheless, we can compute certain correlators for the charge-conjugated theory as correlators of defect changing operators in the diagonal minimal model.

\subsubsection*{RR charges}

Of particular interest in  string theory are D-brane charges. In CFT, they correspond to disc one-point functions of RR fields. B-type boundary conditions require that the disc one-point function of any bulk field is  nonvanishing only if the $U(1)$ charges satisfy $q_L=-q_R$. In the original unorbifolded CFT corresponding to a Landau-Ginzburg model, the modular invariant is diagonal and imposes $q_L=q_R$. As a consequence, D-branes can only be charged if there are RR ground states  with $q_L=q_R=0$. Under symmetric spectral flow, these RR ground states are mapped to chiral primary states in the NSNS sector with $q_L=q_R=c/6$. These states are elements of the chiral ring, represented as polynomials in the Landau-Ginzburg model. To be more precise, the relevant monomials must have half the charge of the unique element of maximal charge in the chiral ring. Hence the selection rule is quite restrictive. For example, B-type branes in minimal models at odd level can never carry RR charge. 

This is different in the orbifold theory, where we are interested in D-brane charges in the twisted sector. As described above, we compute them as one-point correlators of defect changing fields between the identity and symmetry defects. On the level of the full CFT, we consider RR states satisfying $q_L+q_R=0$, the charge selection rule for B-type boundary conditions. Since the initial theory was spectral flow symmetric, we can apply a half-unit flow for any such RR state to find the corresponding defect changing field in the NSNS sector. It has charge $(q_L+c/6, -q_L+c/6)$, where $q_L$ is the charge of the left movers of the initial RR ground state. This state is now realised as a defect changing field in the associated B-twisted Landau-Ginzburg model. It automatically satisfies the charge selection rule for the topological theory: on the disc the total charge of operator insertions has to be $c/3$. We can now use our formalism to compute the disc one-point correlator of a defect changing operator. In CFT 
it corresponds to the disc one-point function of a RR field, hence in string theory to the RR charge.

Note that the Landau-Ginzburg twist field whose one-point function we compute does not necessarily survive the orbifold projection. This is in particular not the case if the spectral flow operator does not survive the orbifold projection. However, the corresponding RR ground state does survive, and this is the state that we are eventually interested in.

Consider  as an example again the case of supersymmetric minimal models with diagonal modular invariant, Landau-Ginzburg potential $W=x^{k+2}$ and orbifold group $\Z_{k+2}$. In every twisted sector of the final projected orbifold theory, there is a single RR ground state whose total charge vanishes. In addition, there is an (a,c) field to which this state flows by asymmetric left-right spectral flow. The twisted sector (c,c) fields, to which the RR ground states flow by symmetric spectral flow, are removed by the orbifold projection implemented on twisted sector states. On the other hand, before projection they are defect changing fields between appropriate symmetry defects, and we can compute the disc one-point function in the Landau-Ginzburg framework, see Appendix~\ref{app:RRchargeComputations}.

\subsubsection*{Superpotential terms}

A further quantity of interest for string compactifications with D-branes is the effective superpotential. In Calabi-Yau compactifications the behaviour of B-type branes under complex structure deformations is encoded in the superpotential of the B-twisted topological theory. Its first order terms can be computed as the three-point function of three chiral boundary superfields or one (c,c) bulk field and one boundary field. The results of this paper can be used to explicitly compute such terms in case the bulk insertion is a twist field. As before, the twist field is realised as a defect changing operator, and our formalism yields a concrete formula for such terms.

As an example, consider hypersurfaces of degree~8 in $\IP_{(1,1,2,2,2)}$. For these Calabi-Yau spaces we have $h^{2,1}=86$, and three of the complex structure deformations are not realised by polynomials. At the Fermat point of the hypersurface, there exists a family of D2-branes, wrapping the exceptional $\IP_1$s introduced by the resolution of the $\Z_2$-orbifold singularity of the ambient weighted projective space. The open string moduli space is one-dimensional, corresponding to different points on the curve
\begin{equation}\label{sgcurve}
x_3^4+x_4^4+x_5^4=0 \, .
\end{equation}
Geometrically, one expects that obstructions arise when perturbing with the nonpolynomial deformations, such that the family of D2-branes reduces to a finite number of holomorphically embedded $\IP_1$s, corresponding to  supersymmetric D2-branes. 

Physically, the obstruction is encoded in a superpotential, where the first order term consists of a two-point function of the marginal operator generating the family, and the obstructing bulk field. This superpotential has been discussed from a geometric point of view in \cite{Kachru:2000an}. In the Landau-Ginzburg framework,  it was determined in \cite{Baumgartl:2012uh}, where it was computed to first order in the bulk and all orders in the boundary couplings. These computations made use of the image of the twist fields under the bulk-boundary map, motivated by charge conservation, and then used the Kapustin-Li correlator for boundary fields. 

We can give an alternative derivation of the correlation functions using defects. One starts with the Landau-Ginzburg potential
$$
W=x_1^8+x_2^8+x_3^4+x_4^4+x_5^4 
$$
and orbifold group $G=\Z_8$. The nonpolynomial complex structure deformations correspond to three states in the 4-th twisted sector. We realise them as defect changing operators between the identity defect~$I$ and~${}_4 I$. This symmetry defect can be obtained by a tensor product construction of the defects discussed in Appendix~\ref{app:RRchargeComputations} for single minimal models. The defect changing field consists of two fermionic pieces (see the appendix for concrete formulas)  in the first two minimal models of charge $q=3/4$, multiplied by $x_3$, $x_4$ or $x_5$. The total charge of the twist field is then $(1,1)$, such that these fields are good marginal operators. The D2-branes correspond to Koszul-type matrix factorisations given by
$$
J_1=x_1\, , \quad J_2=x_2 \, , \quad J_3=a x_4- bx_3 \, , \quad J_4=cx_3-a x_5 
$$
where $a,b,c\in\C$ parametrise the open string moduli space and have to satisfy~\eqref{sgcurve}. At specific points in the moduli space, these branes reduce to the tensor product and minimal model branes described in the appendix. One can then use the formulas given there to verify the computation of
 \cite{Baumgartl:2012uh}. The expected result is that after perturbation the one-dimensional family of D2-branes is reduced to four supersymmetric vacua.

\section{Generalised orbifolds}\label{sec:GOs}

We now wish to formalise and generalise our discussion of Section~\ref{sec:ordinaryLGorbs}. For this an algebro-diagrammatic language is particularly suitable as we shall review in Section~\ref{subsec:BicategoricalAlgebra}. This allows us to develop the theory for \textsl{any} two-dimensional TFT, not only Landau-Ginzburg models; in fact the only input we require is a suitable bicategory which may e.\,g.~also arise from sigma models or WZW models. 

In Section~\ref{subsec:gentwist} we introduce ``generalised'' chiral primary fields and Ramond ground states for any defect~$A$ that has the structure of a separable Frobenius algebra (a notion that we will explain below).\footnote{In particular, we do not assume~$A$ to be symmetric, hence the associated orbifolds have a weaker structure than that of a TFT with defects, see~\cite[Sect.\,3]{cr1210.6363} and Remark~\ref{rem:extendedTFT}. We continue to somewhat vaguely refer to this as ``generalised orbifolds''.} 
In Section~\ref{subsec:ocTFT} we prove several general properties of this construction and explore to what extent generalised orbifolds give rise to new open/closed TFTs, and in Section~\ref{subsec:DefectsFunctoriality} we discuss the defect sector. 

A special case of our central player~$A$ is the symmetry defect~$A_G$ of~\eqref{eq:AG}. Already in this case some of our results go beyond what was established in the literature, but the more exciting questions concern examples that do not originate from classical orbifold groups. Both classes are treated on the same footing in our approach.

\subsection{Bicategorical algebra}\label{subsec:BicategoricalAlgebra}

\subsubsection*{Bicategories with adjoints}

We start with the observation that the properties of Landau-Ginzburg models collected in Sections~\ref{subsec:conventionalorbi} and~\ref{subsec:adjointdefects} give rise to the structure of a bicategory with adjoints. Recall that a \textsl{bicategory} (or \textsl{weak 2-category})~$\B$  is made of a class of \textsl{objects} $a\in\B$, and for all pairs $a,b\in\B$ there are categories $\B(a,b)$ whose objects and arrows are called \textsl{1-morphisms} and \textsl{2-morphisms}, respectively. There are \textsl{tensor product} functors $\otimes: \B(b,c) \times \B(a,b) \rightarrow \B(a,c)$ which are associative in the sense that there exist natural 2-isomorphisms $(X\otimes Y)\otimes Z \cong X \otimes (Y\otimes Z)$ for all composable 1-morphisms $X,Y,Z$. Furthermore, for every $a\in\B$ we have the \textsl{unit 1-morphism} $I_a \in \B(a,a)$ which comes with natural left and right actions 
$$
\lambda_X : 
I_b \otimes X \stackrel{\cong}{\longrightarrow} X 
\, , \quad 
\rho_X : 
X \otimes I_a \stackrel{\cong}{\longrightarrow} X 
$$
for all $X\in\B(a,b)$. These data have to satisfy certain properties which are e.\,g.~explained in~\cite[(7.18),\,(7.19)]{bor94}, but we will not have to directly refer to them in the following. 

We think of objects, 1- and 2-morphisms of a bicategory as the (bulk) theories, defects and local operators of a two-dimensional TFT. For example, topological Landau-Ginzburg models form a bicategory~$\LG$ with the data of Section~\ref{sec:ordinaryLGorbs} as explained in~\cite{cr0909.4381, cm1208.1481}: objects $a=W$ are potentials, the tensor product describes fusion and~$I_a$ is interpreted as the invisible defect. 

The diagrammatic language of Section~\ref{subsec:adjointdefects} is also applicable for general bicategories: for a 2-morphism we have 
$$
\Hom(X_1\otimes \ldots \otimes X_r, Y_1 \otimes \ldots \otimes Y_s) \ni \Phi
\equiv
%%%%%%%%%%%%%%%%%%%%%%
\begin{tikzpicture}[very thick, scale=1.0, color=blue!50!black, baseline=-1.3]
\draw[
	decoration={markings, mark=at position 0.5 with {\arrow{>}}}, postaction={decorate}
	]
	(0,0) -- (120:1);
\draw[
	decoration={markings, mark=at position 0.5 with {\arrow{>}}}, postaction={decorate}
	]
	(0,0) -- (60:1);
\draw[dotted] (110:0.8) arc (110:70:0.8);
\fill (120:1) circle (0pt) node[left] {{\small $Y_1$}};
\fill (60:1) circle (0pt) node[right] {{\small $Y_s$}};
\draw[
	decoration={markings, mark=at position 0.5 with {\arrow{>}}}, postaction={decorate}
	]
	(-50:1) -- (0,0);
\draw[
	decoration={markings, mark=at position 0.5 with {\arrow{>}}}, postaction={decorate}
	]
	(-130:1) -- (0,0);
\fill (-130:1) circle (0pt) node[left] {{\small $X_1$}};
\fill (-50:1) circle (0pt) node[right] {{\small $X_r$}};
\draw[dotted] (-120:0.8) arc (-120:-60:0.8);
\fill (0,0) circle (2.5pt) node[left] {{\small $\Phi$}};
\end{tikzpicture}
%%%%%%%%%%%%%%%%%%%%%%
\, , 
$$
ordinary concatenation is depicted vertically and the tensor product is understood horizontally. Hence any composition of 2-morphisms can be computed as the associated \textsl{string diagram} by reading the latter from bottom to top (operator product) and from right to left (fusion). The ``value'' of a string diagram depends only on its planar isotopy class, thus providing us with a convenient, rigorous graphical calculus in any bicategory; for details we refer to the original papers~\cite{JSGoTCI,JSGoTCII} and the review~\cite{ladia}. 

\medskip

A bicategory~$\B$ \textsl{has right adjoints} if for every $X\in\B(a,b)$ there is a $X^\dagger \in \B(b,a)$ together with maps
$$
\tev_X = \!
%%%%%%%%%%%%%%%%%%%%%%
\begin{tikzpicture}[very thick,scale=1.0,color=blue!50!black, baseline=.4cm]
\draw[line width=0pt] 
(3,0) node[line width=0pt] (D) {{\small$X^\dagger$}}
(2,0) node[line width=0pt] (s) {{\small$X\vphantom{X^\dagger}$}}; 
\draw[redirected] (D) .. controls +(0,1) and +(0,1) .. (s);
\end{tikzpicture}
%%%%%%%%%%%%%%%%%%%%%%
\!\!\! : X \otimes X^\dagger \lra I_b 
\, , \quad 
\tcoev_X = \!
%%%%%%%%%%%%%%%%%%%%%%
\begin{tikzpicture}[very thick,scale=1.0,color=blue!50!black, baseline=-.4cm,rotate=180]
\draw[line width=0pt] 
(3,0) node[line width=0pt] (D) {{\small$X^\dagger$}}
(2,0) node[line width=0pt] (s) {{\small$X\vphantom{X^\dagger}$}}; 
\draw[directed] (D) .. controls +(0,1) and +(0,1) .. (s);
\end{tikzpicture}
%%%%%%%%%%%%%%%%%%%%%%
\!\! : I_a \lra X^\dagger \otimes X
$$
that must satisfy the \textsl{Zorro moves}
$$
%%%%%%%%%%%%%%%%%%%%%%
\begin{tikzpicture}[very thick,scale=0.85,color=blue!50!black, baseline=0cm]
\draw[line width=0] 
(1,1.25) node[line width=0pt] (A) {{\small $X$}}
(-1,-1.25) node[line width=0pt] (A2) {{\small $X$}}; 
\draw[redirected] (0,0) .. controls +(0,1) and +(0,1) .. (-1,0);
\draw[redirected] (1,0) .. controls +(0,-1) and +(0,-1) .. (0,0);
\draw (-1,0) -- (A2); 
\draw (1,0) -- (A); 
\end{tikzpicture}
%%%%%%%%%%%%%%%%%%%%%%
=
%%%%%%%%%%%%%%%%%%%%%%
\begin{tikzpicture}[very thick,scale=0.85,color=blue!50!black, baseline=0cm]
\draw[line width=0] 
(0,1.25) node[line width=0pt] (A) {{\small $X$}}
(0,-1.25) node[line width=0pt] (A2) {{\small $X$}}; 
\draw[
	decoration={markings, mark=at position 0.5 with {\arrow{<}}}, postaction={decorate}
	]
(A) -- (A2); 
\end{tikzpicture}
%%%%%%%%%%%%%%%%%%%%%%
\, , \qquad
%%%%%%%%%%%%%%%%%%%%%%
\begin{tikzpicture}[very thick,scale=0.85,color=blue!50!black, baseline=0cm]
\draw[line width=0] 
(-1,1.25) node[line width=0pt] (A) {{\small $X^\dagger$}}
(1,-1.25) node[line width=0pt] (A2) {{\small $X^\dagger$}}; 
\draw[redirected] (0,0) .. controls +(0,-1) and +(0,-1) .. (-1,0);
\draw[redirected] (1,0) .. controls +(0,1) and +(0,1) .. (0,0);
\draw (-1,0) -- (A); 
\draw (1,0) -- (A2); 
\end{tikzpicture}
%%%%%%%%%%%%%%%%%%%%%%
=
%%%%%%%%%%%%%%%%%%%%%%
\begin{tikzpicture}[very thick,scale=0.85,color=blue!50!black, baseline=0cm]
\draw[line width=0] 
(0,1.25) node[line width=0pt] (A) {{\small $X^\dagger$}}
(0,-1.25) node[line width=0pt] (A2) {{\small $X^\dagger$}}; 
\draw[
	decoration={markings, mark=at position 0.5 with {\arrow{<}}}, postaction={decorate}
	]
 (A2) -- (A); 
\end{tikzpicture}
%%%%%%%%%%%%%%%%%%%%%%
. 
$$
Here we label upwards-oriented lines with~$X$ and downwards-oriented lines with~$X^\dagger$ (but soon we will only imagine such labels if they are clear from the context), and we do no longer show the lines pertaining to the units $I_a, I_b$ -- after all, they describe invisible defects. For a 2-morphism $\Phi \in \Hom(X,Y)$ we have
\be\label{eq:rightPhi}
%%%%%%%%%%%%%%%%%%%%%%
\begin{tikzpicture}[very thick,scale=0.75,color=blue!50!black, baseline=.4cm]
\fill (3,0.5) circle (2.5pt) node[right] {{\small $\Phi^\dagger$}};
\draw (3,0) -- (3,1.5)
node[above] {{{\small$X^\dagger$}}};
\draw (3,0) -- (3,-0.5)
node[below] {{{\small$Y^\dagger$}}};
\end{tikzpicture} 
%%%%%%%%%%%%%%%%%%%%%%
=
%%%%%%%%%%%%%%%%%%%%%%
\begin{tikzpicture}[very thick,scale=0.75,color=blue!50!black, baseline=.4cm]
\draw[line width=0pt] 
(3,0.5) node[line width=0pt] (D) {}
(2,0.5) node[line width=0pt] (s) {}; 
\draw[redirected] (D) .. controls +(0,1) and +(0,1) .. (s);
\fill (2,0.5) circle (2.5pt) node[left] {{\small $\Phi$}};
\draw (2,0.32) -- (2,0.68);
\draw[line width=0pt] 
(2,0.5) node[line width=0pt] (D) {}
(1,0.5) node[line width=0pt] (s) {}; 
\draw[redirected] (D) .. controls +(0,-1) and +(0,-1) .. (s);
\draw (1,0.32) -- (1,1.5)
node[above] {{{\small$X^\dagger$}}};
\draw (3,0.68) -- (3,-0.5)
node[below] {{{\small$Y^\dagger$}}};
\end{tikzpicture} 
%%%%%%%%%%%%%%%%%%%%%%
\, , \quad
%%%%%%%%%%%%%%%%%%%%%%
\begin{tikzpicture}[very thick,scale=0.75,color=blue!50!black, baseline=.4cm]
\draw[line width=0pt] 
(3,0.5) node[line width=0pt] (D) {}
(2,0.5) node[line width=0pt] (s) {}; 
\draw[redirected] (D) .. controls +(0,-1) and +(0,-1) .. (s);
\fill (3,0.5) circle (2.5pt) node[left] {{\small $\Phi$}};
\draw (3,0.32) -- (3,1)
node[above] {{{\small$Y$}}};
\draw (2,0.32) -- (2,1)
node[above] {{{\small$X^\dagger$}}};
\end{tikzpicture} 
%%%%%%%%%%%%%%%%%%%%%%
=
%%%%%%%%%%%%%%%%%%%%%%
\begin{tikzpicture}[very thick,scale=0.75,color=blue!50!black, baseline=.4cm]
\draw[line width=0pt] 
(3,0.5) node[line width=0pt] (D) {}
(2,0.5) node[line width=0pt] (s) {}; 
\draw[redirected] (D) .. controls +(0,-1) and +(0,-1) .. (s);
\fill (2,0.5) circle (2.5pt) node[right] {{\small $\Phi^\dagger$}};
\draw (3,0.32) -- (3,1)
node[above] {{{\small$Y$}}};
\draw (2,0.32) -- (2,1)
node[above] {{{\small$X^\dagger$}}};
\end{tikzpicture} 
%%%%%%%%%%%%%%%%%%%%%%
, \quad
%%%%%%%%%%%%%%%%%%%%%%
\begin{tikzpicture}[very thick,scale=0.75,color=blue!50!black, baseline=.4cm]
\draw[line width=0pt] 
(3,0.5) node[line width=0pt] (D) {}
(2,0.5) node[line width=0pt] (s) {}; 
\draw[redirected] (D) .. controls +(0,1) and +(0,1) .. (s);
\fill (2,0.5) circle (2.5pt) node[right] {{\small $\Phi$}};
\draw (3,0.68) -- (3,0)
node[below] {{{\small$Y^\dagger$}}};
\draw (2,0.68) -- (2,0)
node[below] {{{\small$X\vphantom{\dY}$}}};
\end{tikzpicture} 
%%%%%%%%%%%%%%%%%%%%%%
=
%%%%%%%%%%%%%%%%%%%%%%
\begin{tikzpicture}[very thick,scale=0.75,color=blue!50!black, baseline=.4cm]
\draw[line width=0pt] 
(3,0.5) node[line width=0pt] (D) {}
(2,0.5) node[line width=0pt] (s) {}; 
\draw[redirected] (D) .. controls +(0,1) and +(0,1) .. (s);
\fill (3,0.5) circle (2.5pt) node[left] {{\small $\Phi^\dagger$}};
\draw (3,0.68) -- (3,0)
node[below] {{{\small$Y^\dagger$}}};
\draw (2,0.68) -- (2,0)
node[below] {{{\small$X\vphantom{\dY}$}}};
\end{tikzpicture} 
%%%%%%%%%%%%%%%%%%%%%%
\ee
where the first identity is the definition of $\Phi^\dagger \in \Hom(Y^\dagger, X^\dagger)$, from which the other two identities follow by applying the Zorro moves. 

Similarly, $\B$ \textsl{has left adjoints} if for every $X\in\B(a,b)$ there is a $\dX\in\B(b,a)$ together with maps
$$
\ev_X = \!
%%%%%%%%%%%%%%%%%%%%%%
\begin{tikzpicture}[very thick,scale=1.0,color=blue!50!black, baseline=.4cm]
\draw[line width=0pt] 
(3,0) node[line width=0pt] (D) {{\small$X\vphantom{X^\dagger}$}}
(2,0) node[line width=0pt] (s) {{\small$\dX$}}; 
\draw[directed] (D) .. controls +(0,1) and +(0,1) .. (s);
\end{tikzpicture}
%%%%%%%%%%%%%%%%%%%%%%
\!\! : \dX \otimes X \lra I_a
\, , \quad
\coev_X = \!
%%%%%%%%%%%%%%%%%%%%%%
\begin{tikzpicture}[very thick,scale=1.0,color=blue!50!black, baseline=-.4cm,rotate=180]
\draw[line width=0pt] 
(3,0) node[line width=0pt] (D) {{\small$X\vphantom{X^\dagger}$}}
(2,0) node[line width=0pt] (s) {{\small$\dX$}}; 
\draw[redirected] (D) .. controls +(0,1) and +(0,1) .. (s);
\end{tikzpicture}
%%%%%%%%%%%%%%%%%%%%%%
\!\! : I_b \lra X \otimes \dX
$$
satisfying their versions of Zorro moves. (The details as well as the analogue of~\eqref{eq:rightPhi} we leave as a mystery to the reader.) Under the assumption that $\dX = X^\dagger$ and $\dPhi = \Phi^\dagger$ for all 1- and 2-morphisms we call the bicategory \textsl{pivotal} if 
\be\label{eq:pivotal}
%%%%%%%%%%%%%%%%%%%%%%
\begin{tikzpicture}[very thick,scale=0.65,color=blue!50!black, baseline=-0.2cm, rotate=180]
\draw[line width=1pt] 
(2,-1.5) node[line width=0pt] (Y) {{\small $Y^\dagger$}}
(3,-1.5) node[line width=0pt] (X) {{\small $X^\dagger$}}
(-1,2) node[line width=0pt] (XY) {{\small $(Y\otimes X)^\dagger$}}; 
\draw[redirected] (1,0) .. controls +(0,1) and +(0,1) .. (2,0);
\draw[redirected] (0,0) .. controls +(0,2) and +(0,2) .. (3,0);
\draw[redirected, ultra thick] (-1,0) .. controls +(0,-1) and +(0,-1) .. (0.5,0);
\draw (2,0) -- (Y);
\draw (3,0) -- (X);
\draw[dotted] (0,0) -- (1,0);
\draw[ultra thick] (-1,0) -- (XY);
\end{tikzpicture}
%%%%%%%%%%%%%%%%%%%%%%
\! =  \!
%%%%%%%%%%%%%%%%%%%%%%
\begin{tikzpicture}[very thick,scale=0.65,color=blue!50!black, baseline=-0.2cm, rotate=180]
\draw[line width=1pt] 
(-2,-1.5) node[line width=0pt] (X) {{\small $Y^\dagger$}}
(-1,-1.5) node[line width=0pt] (Y) {{\small $X^\dagger$}}
(2,2) node[line width=0pt] (XY) {{\small $(Y\otimes X)^\dagger$}}; 
\draw[redirected] (0,0) .. controls +(0,1) and +(0,1) .. (-1,0);
\draw[redirected] (1,0) .. controls +(0,2) and +(0,2) .. (-2,0);
\draw[redirected, ultra thick] (2,0) .. controls +(0,-1) and +(0,-1) .. (0.5,0);
\draw (-1,0) -- (Y);
\draw (-2,0) -- (X);
\draw[dotted] (0,0) -- (1,0);
\draw[ultra thick] (2,0) -- (XY);
\end{tikzpicture}
%%%%%%%%%%%%%%%%%%%%%% 
\ee
for all composable $X,Y$. In a pivotal bicategory the \textsl{left} and \textsl{right quantum dimensions} of a 1-morphism~$X$ are defined as 
\be\label{eq:qdims}
\dim_{\text{l}}(X) = 
%%%%%%%%%%%%%%%%%%%%%%
\begin{tikzpicture}[very thick,scale=0.45,color=blue!50!black, baseline]
\draw (0,0) circle (2);
\draw[<-, very thick] (0.100,-2) -- (-0.101,-2) node[above] {}; 
\draw[<-, very thick] (-0.100,2) -- (0.101,2) node[below] {}; 
\fill (0:2) circle (0pt) node[left] {{\small$X$}};
\end{tikzpicture} 
%%%%%%%%%%%%%%%%%%%%%% 
\, , \quad
\dim_{\text{r}}(X) = 
%%%%%%%%%%%%%%%%%%%%%%
\begin{tikzpicture}[very thick,scale=0.45,color=blue!50!black, baseline]
\draw (0,0) circle (2);
\draw[->, very thick] (0.100,-2) -- (-0.101,-2) node[above] {}; 
\draw[->, very thick] (0.100,2) -- (0.101,2) node[below] {}; 
\fill (180:2) circle (0pt) node[right] {{\small$X$}};
\end{tikzpicture} 
%%%%%%%%%%%%%%%%%%%%%% 
\, . 
\ee

\medskip

As shown in~\cite{cr1006.5609, bfk1105.3177, cm1208.1481} the bicategory $\LG$ of Landau-Ginzburg models does have left and right adjoints, see Section~\ref{subsec:adjointdefects} for the details. However, $\dX$ and~$X^\dagger$ are the same only up to a possible shift, so~\eqref{eq:pivotal} cannot be strictly true for $\LG$ in general. This issue was carefully addressed in~\cite[Sect.\,7]{cm1208.1481}, see also Remark~\ref{rem:twistedpivotal}. The upshot is that for our purposes we may nevertheless assume pivotality~\eqref{eq:pivotal} to hold generally for Landau-Ginzburg models as the signs introduced by shifts will cancel out in all our applications (such as~\eqref{eq:qdims} above). The same is true for other examples such as A- and B-models.

\subsubsection*{Algebras and modules}

As we have seen every defect in topological Landau-Ginzburg models has adjoints; we will now consider properties that only certain special defects have -- the kind that will play a central role in generalising the discussion of Section~\ref{sec:ordinaryLGorbs} based on the symmetry defect~$A_G$. This can be done while continuing to work within the general framework of a pivotal bicategory~$\B$. 

An \textsl{algebra} is a 1-morphism $A\in \B(a,a)$ together with maps 
$$
%%%%%%%%%%%%%%%%%%%%%%
\begin{tikzpicture}[very thick,scale=0.75,color=green!50!black, baseline=0.4cm]
\draw[-dot-] (3,0) .. controls +(0,1) and +(0,1) .. (2,0);
\draw (2.5,0.75) -- (2.5,1.5); 
\end{tikzpicture} 
%%%%%%%%%%%%%%%%%%%%%% 
: A\otimes A \lra A
\, , \quad 
%%%%%%%%%%%%%%%%%%%%%%
\begin{tikzpicture}[very thick,scale=0.75,color=green!50!black, baseline]
\draw (0,-0.5) node[Odot] (D) {}; 
\draw (D) -- (0,0.6); 
\end{tikzpicture} 
%%%%%%%%%%%%%%%%%%%%%% 
: I_a \lra A
$$
with 
\be\label{eq:associativeAlgebra}
%%%%%%%%%%%%%%%%%%%%%%
\begin{tikzpicture}[very thick,scale=0.75,color=green!50!black, baseline=0.6cm]
\draw[-dot-] (3,0) .. controls +(0,1) and +(0,1) .. (2,0);
\draw[-dot-] (2.5,0.75) .. controls +(0,1) and +(0,1) .. (3.5,0.75);
\draw (3.5,0.75) -- (3.5,0); 
\draw (3,1.5) -- (3,2.25); 
\end{tikzpicture} 
%%%%%%%%%%%%%%%%%%%%%% 
=
%%%%%%%%%%%%%%%%%%%%%%
\begin{tikzpicture}[very thick,scale=0.75,color=green!50!black, baseline=0.6cm]
\draw[-dot-] (3,0) .. controls +(0,1) and +(0,1) .. (2,0);
\draw[-dot-] (2.5,0.75) .. controls +(0,1) and +(0,1) .. (1.5,0.75);
\draw (1.5,0.75) -- (1.5,0); 
\draw (2,1.5) -- (2,2.25); 
\end{tikzpicture} 
%%%%%%%%%%%%%%%%%%%%%% 
\, , \quad
%%%%%%%%%%%%%%%%%%%%%%
\begin{tikzpicture}[very thick,scale=0.75,color=green!50!black, baseline=-0.2cm]
\draw (-0.5,-0.5) node[Odot] (unit) {}; 
\fill (0,0.6) circle (2.5pt) node (meet) {};
\draw (unit) .. controls +(0,0.5) and +(-0.5,-0.5) .. (0,0.6);
\draw (0,-1) -- (0,1); 
\end{tikzpicture} 
%%%%%%%%%%%%%%%%%%%%%% 
=
%%%%%%%%%%%%%%%%%%%%%%
\begin{tikzpicture}[very thick,scale=0.75,color=green!50!black, baseline=-0.2cm]
\draw (0,-1) -- (0,1); 
\end{tikzpicture} 
%%%%%%%%%%%%%%%%%%%%%% 
=
%%%%%%%%%%%%%%%%%%%%%%
\begin{tikzpicture}[very thick,scale=0.75,color=green!50!black, baseline=-0.2cm]
\draw (0.5,-0.5) node[Odot] (unit) {}; 
\fill (0,0.6) circle (2.5pt) node (meet) {};
\draw (unit) .. controls +(0,0.5) and +(0.5,-0.5) .. (0,0.6);
\draw (0,-1) -- (0,1); 
\end{tikzpicture} 
%%%%%%%%%%%%%%%%%%%%%% 
\, . 
\ee
This is the obvious generalisation of an ordinary associative unital algebra in the category of vector spaces. Similarly, we call~$A$ a \textsl{coalgebra} if it has maps 
$$
%%%%%%%%%%%%%%%%%%%%%%
\begin{tikzpicture}[very thick,scale=0.75,color=green!50!black, baseline=-0.6cm, rotate=180]
\draw[-dot-] (3,0) .. controls +(0,1) and +(0,1) .. (2,0);
\draw (2.5,0.75) -- (2.5,1.5); 
\end{tikzpicture} 
%%%%%%%%%%%%%%%%%%%%%% 
: A \lra A \otimes A
\, , \quad
%%%%%%%%%%%%%%%%%%%%%%
\begin{tikzpicture}[very thick,scale=0.75,color=green!50!black, baseline, rotate=180]
\draw (0,-0.5) node[Odot] (D) {}; 
\draw (D) -- (0,0.6); 
\end{tikzpicture} 
%%%%%%%%%%%%%%%%%%%%%%  
: A \lra I_a
$$
with
\be\label{eq:coassociativeCoalgebra}
%%%%%%%%%%%%%%%%%%%%%%
\begin{tikzpicture}[very thick,scale=0.75,color=green!50!black, baseline=-1.0cm, rotate=180]
\draw[-dot-] (3,0) .. controls +(0,1) and +(0,1) .. (2,0);
\draw[-dot-] (2.5,0.75) .. controls +(0,1) and +(0,1) .. (1.5,0.75);
\draw (1.5,0.75) -- (1.5,0); 
\draw (2,1.5) -- (2,2.25); 
\end{tikzpicture} 
%%%%%%%%%%%%%%%%%%%%%% 
=
%%%%%%%%%%%%%%%%%%%%%%
\begin{tikzpicture}[very thick,scale=0.75,color=green!50!black, baseline=-1.0cm, rotate=180]
\draw[-dot-] (3,0) .. controls +(0,1) and +(0,1) .. (2,0);
\draw[-dot-] (2.5,0.75) .. controls +(0,1) and +(0,1) .. (3.5,0.75);
\draw (3.5,0.75) -- (3.5,0); 
\draw (3,1.5) -- (3,2.25); 
\end{tikzpicture} 
%%%%%%%%%%%%%%%%%%%%%% 
\, , \quad
%%%%%%%%%%%%%%%%%%%%%%
\begin{tikzpicture}[very thick,scale=0.75,color=green!50!black, baseline=-0.1cm, rotate=180]
\draw (0.5,-0.5) node[Odot] (unit) {}; 
\fill (0,0.6) circle (2.5pt) node (meet) {};
\draw (unit) .. controls +(0,0.5) and +(0.5,-0.5) .. (0,0.6);
\draw (0,-1) -- (0,1); 
\end{tikzpicture} 
%%%%%%%%%%%%%%%%%%%%%% 
=
%%%%%%%%%%%%%%%%%%%%%%
\begin{tikzpicture}[very thick,scale=0.75,color=green!50!black, baseline=-0.1cm, rotate=180]
\draw (0,-1) -- (0,1); 
\end{tikzpicture} 
%%%%%%%%%%%%%%%%%%%%%% 
=
%%%%%%%%%%%%%%%%%%%%%%
\begin{tikzpicture}[very thick,scale=0.75,color=green!50!black, baseline=-0.1cm, rotate=180]
\draw (-0.5,-0.5) node[Odot] (unit) {}; 
\fill (0,0.6) circle (2.5pt) node (meet) {};
\draw (unit) .. controls +(0,0.5) and +(-0.5,-0.5) .. (0,0.6);
\draw (0,-1) -- (0,1); 
\end{tikzpicture} 
%%%%%%%%%%%%%%%%%%%%%% 
\, . 
\ee

Now let~$A$ be both an algebra and a coalgebra. It is \textsl{Frobenius} if the identities 
\be\label{eq:FrobeniusProperty}
%%%%%%%%%%%%%%%%%%%%%%
\begin{tikzpicture}[very thick,scale=0.75,color=green!50!black, baseline=0cm]
\draw[-dot-] (0,0) .. controls +(0,-1) and +(0,-1) .. (-1,0);
\draw[-dot-] (1,0) .. controls +(0,1) and +(0,1) .. (0,0);
\draw (-1,0) -- (-1,1.5); 
\draw (1,0) -- (1,-1.5); 
\draw (0.5,0.8) -- (0.5,1.5); 
\draw (-0.5,-0.8) -- (-0.5,-1.5); 
\end{tikzpicture}
%%%%%%%%%%%%%%%%%%%%%%
=
%%%%%%%%%%%%%%%%%%%%%%
\begin{tikzpicture}[very thick,scale=0.75,color=green!50!black, baseline=0cm]
\draw[-dot-] (0,1.5) .. controls +(0,-1) and +(0,-1) .. (1,1.5);
\draw[-dot-] (0,-1.5) .. controls +(0,1) and +(0,1) .. (1,-1.5);
\draw (0.5,-0.8) -- (0.5,0.8); 
\end{tikzpicture}
%%%%%%%%%%%%%%%%%%%%%%
=
%%%%%%%%%%%%%%%%%%%%%%
\begin{tikzpicture}[very thick,scale=0.75,color=green!50!black, baseline=0cm]
\draw[-dot-] (0,0) .. controls +(0,1) and +(0,1) .. (-1,0);
\draw[-dot-] (1,0) .. controls +(0,-1) and +(0,-1) .. (0,0);
\draw (-1,0) -- (-1,-1.5); 
\draw (1,0) -- (1,1.5); 
\draw (0.5,-0.8) -- (0.5,-1.5); 
\draw (-0.5,0.8) -- (-0.5,1.5); 
\end{tikzpicture}
%%%%%%%%%%%%%%%%%%%%%%
\ee
hold. We call~$A$ \textsl{separable} if 
\be\label{eq:separability}
%%%%%%%%%%%%%%%%%%%%%%
\begin{tikzpicture}[very thick,scale=0.75,color=green!50!black, baseline=0cm]
\draw[-dot-] (0,0) .. controls +(0,-1) and +(0,-1) .. (1,0);
\draw[-dot-] (0,0) .. controls +(0,1) and +(0,1) .. (1,0);
\draw (0.5,-0.8) -- (0.5,-1.2); 
\draw (0.5,0.8) -- (0.5,1.2); 
\end{tikzpicture}
%%%%%%%%%%%%%%%%%%%%%%
= 
%%%%%%%%%%%%%%%%%%%%%%
\begin{tikzpicture}[very thick,scale=0.75,color=green!50!black, baseline=0cm]
\draw (0.5,-1.2) -- (0.5,1.2); 
\end{tikzpicture}
%%%%%%%%%%%%%%%%%%%%%%
\ee
and \textsl{symmetric} if 
\be\label{eq:symmetric}
%%%%%%%%%%%%%%%%%%%%%%
\begin{tikzpicture}[very thick,scale=0.75,color=green!50!black, baseline=0cm]
\draw[-dot-] (0,0) .. controls +(0,1) and +(0,1) .. (-1,0);
\draw[directedgreen, color=green!50!black] (1,0) .. controls +(0,-1) and +(0,-1) .. (0,0);
\draw (-1,0) -- (-1,-1.5); 
\draw (1,0) -- (1,1.5); 
\draw (-0.5,1.2) node[Odot] (end) {}; 
\draw (-0.5,0.8) -- (end); 
\end{tikzpicture}
%%%%%%%%%%%%%%%%%%%%%%
= 
%%%%%%%%%%%%%%%%%%%%%%
\begin{tikzpicture}[very thick,scale=0.75,color=green!50!black, baseline=0cm]
\draw[redirectedgreen, color=green!50!black] (0,0) .. controls +(0,-1) and +(0,-1) .. (-1,0);
\draw[-dot-] (1,0) .. controls +(0,1) and +(0,1) .. (0,0);
\draw (-1,0) -- (-1,1.5); 
\draw (1,0) -- (1,-1.5); 
\draw (0.5,1.2) node[Odot] (end) {}; 
\draw (0.5,0.8) -- (end); 
\end{tikzpicture}
%%%%%%%%%%%%%%%%%%%%%%
\, . 
\ee

A map $\phi \in \End(A)$ is an \textsl{algebra morphism} if 
\be\label{eq:AlgebraMorphism}
%%%%%%%%%%%%%%%%%%%%%%
\begin{tikzpicture}[very thick,scale=0.75,color=green!50!black, baseline=0.4cm]
\draw[-dot-] (3,0) .. controls +(0,1) and +(0,1) .. (2,0);
\draw (2.5,0.75) -- (2.5,1.5); 
\fill (2,0) circle (2.5pt) node[left] (alpha) {{\small $\phi$}};
\fill (3,0) circle (2.5pt) node[right] (beta) {{\small $\phi$}};
\draw (2,0) -- (2,-0.5); 
\draw (3,0) -- (3,-0.5); 
\end{tikzpicture} 
%%%%%%%%%%%%%%%%%%%%%% 
=
%%%%%%%%%%%%%%%%%%%%%%
\begin{tikzpicture}[very thick,scale=0.75,color=green!50!black, baseline=0.4cm]
\draw[-dot-] (3,0) .. controls +(0,1) and +(0,1) .. (2,0);
\draw (2.5,0.75) -- (2.5,1.5); 
\fill (2.5,1.15) circle (2.5pt) node[left] (alpha) {{\small $\phi$}};
\draw (2,0) -- (2,-0.5); 
\draw (3,0) -- (3,-0.5); 
\end{tikzpicture} 
%%%%%%%%%%%%%%%%%%%%%% 
\, , \quad
%%%%%%%%%%%%%%%%%%%%%%
\begin{tikzpicture}[very thick,scale=0.75,color=green!50!black, baseline=0.4cm]
\draw (0,-0.4) node[Odot] (D) {}; 
\draw (D) -- (0,1.5); 
\fill (0,0.65) circle (2.5pt) node[left] (alpha) {{\small $\phi$}};
\end{tikzpicture} 
%%%%%%%%%%%%%%%%%%%%%% 
=
%%%%%%%%%%%%%%%%%%%%%%
\begin{tikzpicture}[very thick,scale=0.75,color=green!50!black, baseline=0.4cm]
\draw (0,-0.4) node[Odot] (D) {}; 
\draw (D) -- (0,1.5); 
\end{tikzpicture} 
%%%%%%%%%%%%%%%%%%%%%% 
\, . 
\ee
If~$A$ is Frobenius then an important example is the \textsl{Nakayama automorphism}
\be\label{eq:Nakayama}
\gamma_A = 
%%%%%%%%%%%%%%%%%%%%%%
\begin{tikzpicture}[very thick, scale=0.5,color=green!50!black, baseline=-0.35cm]
\draw (0,0.8) -- (0,2);
\draw[-dot-] (0,0.8) .. controls +(0,-0.5) and +(0,-0.5) .. (-0.75,0.8);
\draw[directedgreen, color=green!50!black] (-0.75,0.8) .. controls +(0,0.5) and +(0,0.5) .. (-1.5,0.8);
\draw[-dot-] (0,-1.8) .. controls +(0,0.5) and +(0,0.5) .. (-0.75,-1.8);
\draw[redirectedgreen, color=green!50!black] (-0.75,-1.8) .. controls +(0,-0.5) and +(0,-0.5) .. (-1.5,-1.8);
\draw (0,-1.8) -- (0,-3);
\draw (-1.5,0.8) -- (-1.5,-1.8);
\draw (-0.375,-0.2) node[Odot] (D) {}; 
\draw (-0.375,0.4) -- (D);
\draw (-0.375,-0.8) node[Odot] (E) {}; 
\draw (-0.375,-1.4) -- (E);
\end{tikzpicture}
%%%%%%%%%%%%%%%%%%%%%%
\, , \quad
\gamma_A^{-1} = 
%%%%%%%%%%%%%%%%%%%%%%
\begin{tikzpicture}[very thick, scale=0.5,color=green!50!black, baseline=-0.35cm]
\draw (0,0.8) -- (0,2);
\draw[-dot-] (0,0.8) .. controls +(0,-0.5) and +(0,-0.5) .. (0.75,0.8);
\draw[directedgreen, color=green!50!black] (0.75,0.8) .. controls +(0,0.5) and +(0,0.5) .. (1.5,0.8);
\draw[-dot-] (0,-1.8) .. controls +(0,0.5) and +(0,0.5) .. (0.75,-1.8);
\draw[redirectedgreen, color=green!50!black] (0.75,-1.8) .. controls +(0,-0.5) and +(0,-0.5) .. (1.5,-1.8);
\draw (0,-1.8) -- (0,-3);
\draw (1.5,0.8) -- (1.5,-1.8);
\draw (0.375,-0.2) node[Odot] (D) {}; 
\draw (0.375,0.4) -- (D);
\draw (0.375,-0.8) node[Odot] (E) {}; 
\draw (0.375,-1.4) -- (E);
\end{tikzpicture}
%%%%%%%%%%%%%%%%%%%%%%
\, . 
\ee
It is a measure for how far~$A$ is away from being symmetric in the sense that~$A$ is symmetric iff $\gamma_A = 1_A$, see e.\,g.~\cite{fs0901.4886}. 
We will see that the Nakayama automorphism plays a central role in our discussion, as does its variant
\be\label{eq:NakayamaPrime}
\gamma'_A = 
%%%%%%%%%%%%%%%%%%%%%%
\begin{tikzpicture}[very thick, scale=0.5,color=green!50!black, baseline=-0.35cm]
\draw (0,0.8) -- (0,2);
\draw[directedgreen] (0,0.8) .. controls +(0,-0.5) and +(0,-0.5) .. (0.75,0.8);
\draw[-dot-, color=green!50!black] (0.75,0.8) .. controls +(0,0.5) and +(0,0.5) .. (1.5,0.8);
\draw[redirectedgreen] (0,-1.8) .. controls +(0,0.5) and +(0,0.5) .. (0.75,-1.8);
\draw[-dot-, color=green!50!black] (0.75,-1.8) .. controls +(0,-0.5) and +(0,-0.5) .. (1.5,-1.8);
\draw (0,-1.8) -- (0,-3);
\draw (1.5,0.8) -- (1.5,-1.8);
\draw (1.125,1.78) node[Odot] (D) {}; 
\draw (1.125,1.25) -- (D);
\draw (1.125,-2.78) node[Odot] (E) {}; 
\draw (1.125,-2.1) -- (E);
\end{tikzpicture}
%%%%%%%%%%%%%%%%%%%%%%
\, , \quad
(\gamma'_A)^{-1} = 
%%%%%%%%%%%%%%%%%%%%%%
\begin{tikzpicture}[very thick, scale=0.5,color=green!50!black, baseline=-0.35cm]
\draw (0,0.8) -- (0,2);
\draw[directedgreen] (0,0.8) .. controls +(0,-0.5) and +(0,-0.5) .. (-0.75,0.8);
\draw[-dot-, color=green!50!black] (-0.75,0.8) .. controls +(0,0.5) and +(0,0.5) .. (-1.5,0.8);
\draw[redirectedgreen] (0,-1.8) .. controls +(0,0.5) and +(0,0.5) .. (-0.75,-1.8);
\draw[-dot-, color=green!50!black] (-0.75,-1.8) .. controls +(0,-0.5) and +(0,-0.5) .. (-1.5,-1.8);
\draw (0,-1.8) -- (0,-3);
\draw (-1.5,0.8) -- (-1.5,-1.8);
\draw (-1.125,1.78) node[Odot] (D) {}; 
\draw (-1.125,1.25) -- (D);
\draw (-1.125,-2.78) node[Odot] (E) {}; 
\draw (-1.125,-2.1) -- (E);
\end{tikzpicture}
%%%%%%%%%%%%%%%%%%%%%%
\, . 
\ee

\begin{example}\label{ex:A_G}
An important special case of a separable Frobenius algebra is the symmetry defect $A_G = \bigoplus_{g\in G} {}_g I$ of~\eqref{eq:AG} in the bicategory $\B = \LG$ of Landau-Ginzburg models (and similarly in other theories). Indeed, as shown in~\cite[Sect.\,7.1]{cr1210.6363}, $A_G$ together with multiplication $\sum_{g,h\in G} {}_h (\lambda_{{}_g I})$ and comultiplication $\frac{1}{|G|}\sum_{g,h\in G} {}_h (\lambda^{-1}_{{}_g I})$ is separable and Frobenius, and it is symmetric iff $\dim_{\text{r}}({}_g I) = 1$ for all $g\in G$. 

To compute the quantum dimension of ${}_g I \cong I_{g^{-1}}$ recall for example from~\cite{kr0405232,pv1002.2116} that $I=I_W$ has quantum dimensions one, so in particular
\be\label{eq:dimI}
\dim_{\text{r}}(I) = 
(-1)^{\binom{n+1}{2}}\Res \left[ \frac{ \str\big(  \partial_{x'_1} d_{I}\ldots \partial_{x'_n} d_{I} \,  \partial_{x_1} d_{I}\ldots \partial_{x_n} d_{I}\big) \operatorname{d}\! x'}{\partial_{x'_1} W, \ldots, \partial_{x'_n} W} \right] 
= 1 \, .
\ee
Using the explicit presentation~\eqref{eq:gIdefect} for~${}_g I$ together with the invariance of~$W$ under $g\in G$, we find that after a variable transformation $x'\mapsto g^{-1}(x')$ the only dependence of $\dim_{\text{r}}({}_g I)$ on~$g$ is via $\operatorname{d}\! x' \mapsto \det(g) \cdot \operatorname{d}\! x'$ in~\eqref{eq:dimI}. The left quantum dimension is computed analogously, and thus we have 
\be\label{eq:qdimIg}
\dim_{\text{r}}({}_g I) = \dim_{\text{r}}(I_{g^{-1}}) = \det(g)
 \, , \quad 
\dim_{\text{l}}({}_g I) = \dim_{\text{l}}(I_{g^{-1}}) = \det(g)^{-1} \, . 
\ee

Next we determine the Nakayama automorphism~$\gamma_{A_G}$ for~$A_G$: 
$$
\gamma_{A_G} = 
\sum_{g\in G} 
%%%%%%%%%%%%%%%%%%%%%%
\begin{tikzpicture}[very thick, scale=0.5,color=green!65!black, baseline=-0.35cm]
\draw (0,0.8) -- (0,1.5); 
\draw (0,-1.8) -- (0,-2.5); 
\draw[-dot-] (0,0.8) .. controls +(0,-0.5) and +(0,-0.5) .. (-0.75,0.8);
\draw[directedlightgreen, color=green!65!black] (-0.75,0.8) .. controls +(0,0.5) and +(0,0.5) .. (-1.5,0.8);
\draw[-dot-] (0,-1.8) .. controls +(0,0.5) and +(0,0.5) .. (-0.75,-1.8);
\draw[redirectedlightgreen, color=green!65!black] (-0.75,-1.8) .. controls +(0,-0.5) and +(0,-0.5) .. (-1.5,-1.8);
\draw (-1.5,0.8) -- (-1.5,-1.8);
\draw (-0.375,-0.2) node[Odot] (D) {}; 
\draw (-0.375,0.4) -- (D);
\draw (-0.375,-0.8) node[Odot] (E) {}; 
\draw (-0.375,-1.4) -- (E);
\fill[color=green!65!black] (-0.25,1.3) circle (0pt) node[right] (alpha) {{\tiny $g$}};
\fill[color=green!65!black] (-0.25,-2.3) circle (0pt) node[right] (alpha) {{\tiny $g$}};
\fill[color=green!65!black] (-1.7,-0.5) circle (0pt) node[right] (alpha) {{\tiny $g$}};
\end{tikzpicture}
%%%%%%%%%%%%%%%%%%%%%%
\eq^{(1)}
\sum_{g\in G} 
%%%%%%%%%%%%%%%%%%%%%%
\begin{tikzpicture}[very thick, scale=0.5,color=green!65!black, baseline=-0.35cm]
\draw (0,0.8) -- (0,1.5); 
\draw (0,-1.8) -- (0,-2.5); 
\draw[-dot-] (0,0.8) .. controls +(0,-0.5) and +(0,-0.5) .. (-0.75,0.8);
\draw[directedlightgreen, color=green!65!black] (-0.75,0.8) .. controls +(0,0.5) and +(0,0.5) .. (-1.5,0.8);
\draw[-dot-] (0,-1.8) .. controls +(0,0.5) and +(0,0.5) .. (-0.75,-1.8);
\draw[redirectedlightgreen, color=green!65!black] (-0.75,-1.8) .. controls +(0,-0.5) and +(0,-0.5) .. (-1.5,-1.8);
\draw (-1.5,0.8) -- (-1.5,-1.8);
\draw[dashed] (-0.375,0.4) -- (-0.375,-1.4);
\fill[color=green!65!black] (-0.25,1.3) circle (0pt) node[right] (alpha) {{\tiny $g$}};
\fill[color=green!65!black] (-0.25,-2.3) circle (0pt) node[right] (alpha) {{\tiny $g$}};
\fill[color=green!65!black] (-1.7,-0.5) circle (0pt) node[right] (alpha) {{\tiny $g$}};
\fill[color=green!65!black] (-0.55,0.35) circle (0pt) node[right] (alpha) {{\tiny ${}_{g^{-1}} (\lambda_{{}_g I}^{-1})$}};
\fill[color=green!65!black] (-0.55,-1.05) circle (0pt) node[right] (alpha) {{\tiny ${}_{g^{-1}} (\lambda_{{}_g I})$}};
\end{tikzpicture}
%%%%%%%%%%%%%%%%%%%%%%
\eq^{(2)}
\sum_{g\in G} 
%%%%%%%%%%%%%%%%%%%%%%
\begin{tikzpicture}[very thick, scale=0.5,color=green!65!black, baseline=-0.1cm]
\draw (2.5,2) -- (2.5,-2); 
\draw (0,0) circle (1.25);
\draw[<-] (0.100,-1.25) -- (-0.101,-1.25) node[above] {}; 
\draw[<-] (-0.100,1.25) -- (0.101,1.25) node[below] {}; 
\draw[dashed] (-45:1.25) .. controls +(0.3,0) and +(-0.2,-0.3) .. (2.5,0.75);
\fill (-45:1.25) circle (3.3pt) node[left] (beta) {};
\fill (2.5,0.75) circle (3.3pt) node[left] (beta) {};
\fill[color=green!65!black] (2.25,1.8) circle (0pt) node[right] (alpha) {{\tiny $g$}};
\fill[color=green!65!black] (2.25,-1.8) circle (0pt) node[right] (alpha) {{\tiny $g$}};
\fill[color=green!65!black] (-1.5,0) circle (0pt) node[right] (alpha) {{\tiny $g$}};
\fill[color=green!65!black] (2.4,0.65) circle (0pt) node[right] (alpha) {{\tiny $\lambda_{{}_g I}$}};
\fill[color=green!65!black] (-0.2,-1.3) circle (0pt) node[right] (alpha) {{\tiny ${}_{g^{-1}} (\lambda^{-1}_{I})$}};
\end{tikzpicture}
%%%%%%%%%%%%%%%%%%%%%%
\eq^{(3)}
\sum_{g\in G} 
%%%%%%%%%%%%%%%%%%%%%%
\begin{tikzpicture}[very thick, scale=0.5,color=green!65!black, baseline=-0.1cm]
\draw (2.5,2) -- (2.5,-2); 
\draw (0,0) circle (1.25);
\draw[<-] (0.100,-1.25) -- (-0.101,-1.25) node[above] {}; 
\draw[<-] (-0.100,1.25) -- (0.101,1.25) node[below] {}; 
\draw[dashed] (-45:1.25) .. controls +(0.3,0) and +(-0.2,-0.3) .. (2.5,0.75);
\fill (-45:1.25) circle (3.3pt) node[left] (beta) {};
\fill (2.5,0.75) circle (3.3pt) node[left] (beta) {};
\fill[color=green!65!black] (2.25,1.8) circle (0pt) node[right] (alpha) {{\tiny $g$}};
\fill[color=green!65!black] (2.25,-1.8) circle (0pt) node[right] (alpha) {{\tiny $g$}};
\fill[color=green!65!black] (-1.5,0) circle (0pt) node[right] (alpha) {{\tiny $g$}};
\fill[color=green!65!black] (2.4,0.65) circle (0pt) node[right] (alpha) {{\tiny $\lambda_{{}_g I}$}};
\fill[color=green!65!black] (-0,-1.5) circle (0pt) node[right] (alpha) {{\tiny $\rho_{{}_{g^{-1}}I}^{-1}$}};
\end{tikzpicture}
%%%%%%%%%%%%%%%%%%%%%%
$$
where we used (1) the explicit (co)multiplication maps for~$A_G$, (2) its Frobenius property, and (3) the identity ${}_{g^{-1}} (\lambda_I^{-1}) = \rho_{{}_{g^{-1}} I}^{-1}$ which can be read off of~\eqref{eq:rhoinverse}. 
Hence by~\eqref{eq:qdimIg} we find 
$$
\gamma_{A_G} = \sum_{g\in G} \det(g) \cdot 1_{{}_g I} \, . 
$$
\end{example}

\begin{example}\label{ex:XdaggerX}
Let $X\in \B(a,b)$ be any 1-morphism in a pivotal bicategory~$\B$. Then $A = X^\dagger \otimes X$ is a symmetric Frobenius algebra in $\B(a,a)$, and~$A$ is separable if $\dim_{\text{r}}(X) \in \End(I_a)$ is invertible as shown in~\cite[Thm.\,4.3]{cr1210.6363}. 
\end{example} 

\medskip

A discussion of algebras in the framework of bicategories should naturally be complemented by their appropriate representation theory. Given an algebra $A\in \B(a,a)$, we say that $X\in\B(a,b)$ is a \textsl{right $A$-module} if there is a map
\be\label{eq:rightAmod}
%%%%%%%%%%%%%%%%%%%%%%
\begin{tikzpicture}[very thick,scale=0.75,color=blue!50!black, baseline]

\draw (0,-1) node[left] (X) {{\small$X$}};
\draw (0,1) node[left] (Xu) {{\small$X$}};

\draw (0,-1) -- (0,1); 

\fill[color=green!50!black] (0,0.0) circle (2.5pt) node (meet) {};
\draw[color=green!50!black] (0.5,-1) .. controls +(0,0.5) and +(0.5,-0.5) .. (0,0.0);
\end{tikzpicture} 
%%%%%%%%%%%%%%%%%%%%%% 
: X \otimes A \lra X 
\quad\text{such that}\quad
%%%%%%%%%%%%%%%%%%%%%%
\begin{tikzpicture}[very thick,scale=0.75,color=blue!50!black, baseline]

\draw (0,-1) node[left] (X) {};
\draw (0,1) node[left] (Xu) {};

\draw (0,-1) -- (0,1); 

\fill[color=green!50!black] (0,-0.25) circle (2.5pt) node (meet) {};
\fill[color=green!50!black] (0,0.75) circle (2.5pt) node (meet2) {};
\draw[color=green!50!black] (0.5,-1) .. controls +(0,0.25) and +(0.25,-0.25) .. (0,-0.25);
\draw[color=green!50!black] (1,-1) .. controls +(0,0.5) and +(0.5,-0.5) .. (0,0.75);
\end{tikzpicture} 
%%%%%%%%%%%%%%%%%%%%%% 
=
%%%%%%%%%%%%%%%%%%%%%%
\begin{tikzpicture}[very thick,scale=0.75,color=blue!50!black, baseline]

\draw (0,-1) node[left] (X) {};
\draw (0,1) node[left] (Xu) {};

\draw (0,-1) -- (0,1); 

\fill[color=green!50!black] (0,0.75) circle (2.5pt) node (meet2) {};

\draw[-dot-, color=green!50!black] (0.5,-1) .. controls +(0,1) and +(0,1) .. (1,-1);

\draw[color=green!50!black] (0.75,-0.2) .. controls +(0,0.5) and +(0.5,-0.5) .. (0,0.75);
\end{tikzpicture} 
%%%%%%%%%%%%%%%%%%%%%% 
\, , \quad
%%%%%%%%%%%%%%%%%%%%%%
\begin{tikzpicture}[very thick,scale=0.75,color=blue!50!black, baseline]
\draw (0,-1) -- (0,1); 
\draw (0,-1) node[left] (X) {};
\draw (0,1) node[left] (Xu) {};
\end{tikzpicture} 
%%%%%%%%%%%%%%%%%%%%%% 
=
%%%%%%%%%%%%%%%%%%%%%%
\begin{tikzpicture}[very thick,scale=0.75,color=blue!50!black, baseline]
\draw (0,-1) -- (0,1); 
\draw (0,-1) node[left] (X) {};
\draw (0,1) node[left] (Xu) {};

\draw[color=green!50!black]  (0.5,-0.5) node[Odot] (unit) {}; 
\fill[color=green!50!black]  (0,0.6) circle (2.5pt) node (meet) {};
\draw[color=green!50!black]  (unit) .. controls +(0,0.5) and +(0.5,-0.5) .. (0,0.6);
\end{tikzpicture} 
%%%%%%%%%%%%%%%%%%%%%% 
\, . 
\ee

Since we reserve different colours for algebras and modules we may often refrain from labelling lines in string diagrams like the above. If $Y\in\B(a,b)$ is another right $A$-module a 2-morphism $\Phi: X\rightarrow Y$ is called a \textsl{module map} if it is compatible with the action of~$A$, i.\,e. 
\be\label{eq:ModuleMap}
%%%%%%%%%%%%%%%%%%%%%%
\begin{tikzpicture}[very thick,scale=0.75,color=blue!50!black, baseline]

\draw (0,-1) node[left] (X) {{\small$X$}};
\draw (0,1) node[left] (Xu) {{\small$Y$}};

\draw (0,-1) -- (0,1); 

\fill (0,-0.3) circle (2.5pt) node[left] (phi) {{\small $\Phi$}};

\fill[color=green!50!black] (0,0.3) circle (2.5pt) node (meet) {};
\draw[color=green!50!black] (0.5,-1) .. controls +(0,0.5) and +(0.5,-0.5) .. (0,0.3);
\end{tikzpicture} 
%%%%%%%%%%%%%%%%%%%%%% 
=
%%%%%%%%%%%%%%%%%%%%%%
\begin{tikzpicture}[very thick,scale=0.75,color=blue!50!black, baseline]

\draw (0,-1) node[left] (X) {{\small$X$}};
\draw (0,1) node[left] (Xu) {{\small$Y$}};

\draw (0,-1) -- (0,1); 

\fill (0,0.3) circle (2.5pt) node[left] (phi) {{\small $\Phi$}};

\fill[color=green!50!black] (0,-0.3) circle (2.5pt) node (meet) {};
\draw[color=green!50!black] (0.5,-1) .. controls +(0,0.25) and +(0.25,-0.25) .. (0,-0.3);
\end{tikzpicture} 
%%%%%%%%%%%%%%%%%%%%%% 
\, . 
\ee
We write $\Hom_A(X,Y)$ for the space of all module maps. If~$A$ is separable Frobenius one can show (see e.\,g.~\cite[Lem.\,2.2]{cr1210.6363}) that 
\be\label{eq:piA}
\Hom_A(X,Y) = \operatorname{im}(\pi_A) 
\quad\text{with}\quad
\pi_A : \Phi \lmt
%%%%%%%%%%%%%%%%%%%%%%
\begin{tikzpicture}[very thick,scale=0.75,color=blue!50!black, baseline=0cm]

\draw (1.5,-1.75) -- (1.5,1.75); 

\draw (1.5,-1.75) node[left] (X) {{\small$X$}};
\draw (1.5,1.75) node[left] (Xu) {{\small$Y$}};

\fill (1.5,0.7) circle (2.5pt) node[left] (phi) {{\small $\Phi$}};

\fill[color=green!50!black] (1.5,0.3) circle (2.5pt) node (meet) {};
\fill[color=green!50!black] (1.5,1.1) circle (2.5pt) node (meet) {};

\draw[color=green!50!black] (2,-0.25) .. controls +(0.0,0.25) and +(0.25,-0.25) .. (1.5,0.3);
\draw[color=green!50!black] (2.5,-0.25) .. controls +(0.0,0.5) and +(0.25,-0.25) .. (1.5,1.1);

\draw[-dot-, color=green!50!black] (2.5,-0.25) .. controls +(0,-0.5) and +(0,-0.5) .. (2,-0.25);

\draw[color=green!50!black] (2.25,-1.2) node[Odot] (unit) {}; 
\draw[color=green!50!black] (2.25,-0.6) -- (unit);

\end{tikzpicture}
%%%%%%%%%%%%%%%%%%%%%%
\; .
\ee

If~$A$ is a coalgebra then a \textsl{right $A$-comodule} is some $X\in \B(a,b)$ with~\eqref{eq:rightAmod} turned upside-down. In particular, for a Frobenius algebra $A\in \B(a,a)$ a right module $X\in\B(a,b)$ is also a right comodule with action 
\be\label{eq:comoduleAction}
%%%%%%%%%%%%%%%%%%%%%%
\begin{tikzpicture}[very thick,scale=0.75,color=blue!50!black, baseline]

\draw (0,-1.25) node[left] (X) {};
\draw (0,1.25) node[left] (Xu) {};

\draw (0,-1.25) -- (0,1.25); 

\fill[color=green!50!black] (0,0.0) circle (2.5pt) node (meet) {};
\draw[color=green!50!black] (0.75,1.25) .. controls +(0,-0.5) and +(0.5,0.5) .. (0,0.0);
\end{tikzpicture} 
%%%%%%%%%%%%%%%%%%%%%% 
=
%%%%%%%%%%%%%%%%%%%%%%
\begin{tikzpicture}[very thick,scale=0.75,color=blue!50!black, baseline=0cm]

\draw (1.5,-1.25) -- (1.5,1.25); 

\draw (1.5,-1.25) node[left] (X) {};
\draw (1.5,1.25) node[left] (Xu) {};

\fill[color=green!50!black] (1.5,0.3) circle (2.5pt) node (meet) {};

\draw[color=green!50!black] (2,-0.25) .. controls +(0.0,0.25) and +(0.25,-0.25) .. (1.5,0.3);
\draw[color=green!50!black] (2.5,-0.25) .. controls +(0.0,0.5) and +(0,-0.25) .. (2.5,1.25);

\draw[-dot-, color=green!50!black] (2.5,-0.25) .. controls +(0,-0.5) and +(0,-0.5) .. (2,-0.25);

\draw[color=green!50!black] (2.25,-1.2) node[Odot] (unit) {}; 
\draw[color=green!50!black] (2.25,-0.6) -- (unit);

\end{tikzpicture}
%%%%%%%%%%%%%%%%%%%%%%
\, . 
\ee

Of course there is also the notion of \textsl{left (co)modules}. For example, a right $A$-module~$X$ 
gives rise to a left $A$-module~$X^\dagger$ with left action
\be\label{eq:pivvotal}
%%%%%%%%%%%%%%%%%%%%%%
\begin{tikzpicture}[very thick,scale=0.85,color=blue!50!black, baseline=0cm]
\draw[line width=0] 
(-0.5,-1.25) node[line width=0pt, color=green!50!black] (Algebra) {{\small $A\vphantom{X^\dagger}$}}
(0,1.25) node[line width=0pt] (A) {{\small $X^\dagger$}}
(0,-1.25) node[line width=0pt] (A2) {{\small $X^\dagger$}};
\draw (A) -- (A2);
\draw[color=green!50!black] (0,0) .. controls +(-0.5,-0.25) and +(0,1) .. (Algebra);
\fill[color=green!50!black] (0,0) circle (2.5pt) node[left] {};
\end{tikzpicture}
%%%%%%%%%%%%%%%%%%%%%%
\!=\!
%%%%%%%%%%%%%%%%%%%%%%
\begin{tikzpicture}[very thick,scale=0.85,color=blue!50!black, baseline=0cm]
\draw[line width=0] 
(0.5,-1.25) node[line width=0pt, color=green!50!black] (Algebra) {{\small $A\vphantom{X^\dagger}$}}
(-1,1.25) node[line width=0pt] (A) {{\small $X^\dagger$}}
(1,-1.25) node[line width=0pt] (A2) {{\small $X^\dagger$}}; 
\draw[redirected] (0,0) .. controls +(0,-1) and +(0,-1) .. (-1,0);
\draw[redirected] (1,0) .. controls +(0,1) and +(0,1) .. (0,0);
\draw[color=green!50!black] (0,0) .. controls +(0.5,-0.25) and +(0,1) .. (Algebra);
\fill[color=green!50!black] (0,0) circle (2.5pt) node[left] {};
\draw (-1,0) -- (A); 
\draw (1,0) -- (A2); 
\end{tikzpicture}
%%%%%%%%%%%%%%%%%%%%%%
\! = \!
%%%%%%%%%%%%%%%%%%%%%%
\begin{tikzpicture}[very thick,scale=0.85,color=blue!50!black, baseline=0cm]
\draw[line width=0] 
(1,1.25) node[line width=0pt] (A) {{\small $X^\dagger$}}
(-1,-1.25) node[line width=0pt] (A2) {{\small $X^\dagger$}}
(-1.5,-1.25) node[line width=0pt, color=green!50!black] (Algebra) {{\small $A\vphantom{X^\dagger}$}}; 
\draw[directed] (0,0) .. controls +(0,1) and +(0,1) .. (-1,0);
\draw[directed] (1,0) .. controls +(0,-1) and +(0,-1) .. (0,0);
\fill[color=green!50!black] (0,0) circle (2.5pt) node[right] {};
\draw[color=green!50!black] (0,0) .. controls +(-0.15,0.15) and +(-0.2,0) .. (0.25,-0.25);
\draw[color=green!50!black] (0.25,-0.25) .. controls +(0.6,0) and +(0.85,0) .. (-0.5,1.25);
\draw[color=green!50!black] (-0.5,1.25) .. controls +(-0.7,0) and +(0,1) .. (-1.5,-0.25);
\draw[<-,color=green!50!black] (0.21,-0.25) -- (0.22,-0.25);
\draw[<-,color=green!50!black] (-0.41,1.25) -- (-0.42,1.25);
\draw[color=green!50!black] (-1.5,-0.25) -- (Algebra); 
\draw (-1,0) -- (A2); 
\draw (1,0) -- (A); 
\end{tikzpicture}
%%%%%%%%%%%%%%%%%%%%%%
,
\ee
where the second equality can be seen as follows: 
$$
%%%%%%%%%%%%%%%%%%%%%%
\begin{tikzpicture}[very thick,scale=0.85,color=blue!50!black, baseline=0cm]
\draw[line width=0] 
(0.5,-1.5) node[line width=0pt, color=green!50!black] (Algebra) {}
(-1,1.5) node[line width=0pt] (A) {}
(1,-1.5) node[line width=0pt] (A2) {}; 
\draw[redirected] (0,0) .. controls +(0,-1) and +(0,-1) .. (-1,0);
\draw[redirected] (1,0) .. controls +(0,1) and +(0,1) .. (0,0);
\draw[color=green!50!black] (0,0) .. controls +(0.5,-0.25) and +(0,1) .. (Algebra);
\fill[color=green!50!black] (0,0) circle (2.5pt) node[left] {};
\draw (-1,0) -- (A); 
\draw (1,0) -- (A2); 
\end{tikzpicture}
%%%%%%%%%%%%%%%%%%%%%%
\! 
\eq^{\text{Zorro}}
\!
%%%%%%%%%%%%%%%%%%%%%%
\begin{tikzpicture}[very thick,scale=0.85,color=blue!50!black, baseline=0cm]
\draw[line width=0] 
(-4,-1.5) node[line width=0pt, color=green!50!black] (Algebra) {}
(-3,1.5) node[line width=0pt] (A) {}
(1,-1.5) node[line width=0pt] (A2) {}; 
\draw[redirected, ultra thick] (0,0) .. controls +(0,-1) and +(0,-1) .. (-1,0);
\draw[redirected] (1,0) .. controls +(0,1) and +(0,1) .. (0,0);
\draw[redirected, ultra thick] (-1,0) .. controls +(0,1) and +(0,1) .. (-2,0);
\draw[redirected, dashed, color=green!50!black, ultra thick] (0,0) .. controls +(0,-1) and +(0,-1) .. (-1,0);
\draw[redirected, dashed, color=green!50!black, ultra thick] (-1,0) .. controls +(0,1) and +(0,1) .. (-2,0);
\fill[color=green!50!black] (0,0) circle (2.5pt) node[left] {};
\draw (1,0) -- (A2); 
\draw[dotted] (-2.25,0) -- (-1.75,0); 
\draw[redirected] (-2.25,0) .. controls +(0,-0.5) and +(0,-0.5) .. (-3,0);
\draw[redirectedgreen, color=green!50!black] (-1.75,0) .. controls +(0,-1) and +(0,-1) .. (-3.5,0);
\draw[redirectedgreen, color=green!50!black] (-3.5,0) .. controls +(0,0.5) and +(0,0.5) .. (-4,0);
\draw[color=green!50!black] (-4,0) -- (Algebra);
\draw (-3,0) -- (A);
\end{tikzpicture}
%%%%%%%%%%%%%%%%%%%%%%
\! 
\eq^{\eqref{eq:pivotal}}_{\text{Zorro}}
\!
%%%%%%%%%%%%%%%%%%%%%%
\begin{tikzpicture}[very thick,scale=0.85,color=blue!50!black, baseline=0cm]
\draw[line width=0] 
(-1.5,-1.5) node[line width=0pt, color=green!50!black] (Algebra) {}
(1.5,1.5) node[line width=0pt] (A) {}
(-1,-1.5) node[line width=0pt] (A2) {}; 
[
\draw[dotted] (-0.25,0) -- (0.25,0); 
\draw[redirectedgreen, color=green!50!black] (0.25,0) .. controls +(0,-0.5) and +(0,-0.5) .. (1,0);
\draw[redirected] (-0.25,0) .. controls +(0,-1) and +(0,-1) .. (1.5,0);
\draw (1.5,0) -- (A);
\draw[directed, ultra thick] (0,0) .. controls +(0,1) and +(0,1) .. (-1,0);
\draw[directed, ultra thick, dashed, color=green!50!black] (0,0) .. controls +(0,1) and +(0,1) .. (-1,0);
\draw (-1,0) -- (A2);
\fill[color=green!50!black] (-1,0) circle (2.5pt) node[left] {};
\draw[redirectedgreen, color=green!50!black] (1,0) .. controls +(0,1.75) and +(0,1.75) .. (-1.5,0);
\draw[color=green!50!black] (-1.5,0) -- (Algebra);
\end{tikzpicture}
%%%%%%%%%%%%%%%%%%%%%%
\!
\eq^{\eqref{eq:rightPhi}}
\!
%%%%%%%%%%%%%%%%%%%%%%
\begin{tikzpicture}[very thick,scale=0.85,color=blue!50!black, baseline=0cm]
\draw[line width=0] 
(1,1.5) node[line width=0pt] (A) {}
(-1,-1.5) node[line width=0pt] (A2) {}
(-1.5,-1.5) node[line width=0pt, color=green!50!black] (Algebra) {}; 
\draw[directed] (0,0) .. controls +(0,1) and +(0,1) .. (-1,0);
\draw[directed] (1,0) .. controls +(0,-1) and +(0,-1) .. (0,0);
\fill[color=green!50!black] (0,0) circle (2.5pt) node[right] {};
\draw[color=green!50!black] (0,0) .. controls +(-0.15,0.15) and +(-0.2,0) .. (0.25,-0.25);
\draw[color=green!50!black] (0.25,-0.25) .. controls +(0.6,0) and +(0.85,0) .. (-0.5,1.25);
\draw[color=green!50!black] (-0.5,1.25) .. controls +(-0.7,0) and +(0,1) .. (-1.5,-0.25);
\draw[<-,color=green!50!black] (0.21,-0.25) -- (0.22,-0.25);
\draw[<-,color=green!50!black] (-0.41,1.25) -- (-0.42,1.25);
\draw[color=green!50!black] (-1.5,-0.25) -- (Algebra); 
\draw (-1,0) -- (A2); 
\draw (1,0) -- (A); 
\end{tikzpicture}
%%%%%%%%%%%%%%%%%%%%%%
.
$$
We will often make use of this induced structure below. 
Everything that is true of right (co)modules holds as well for their left versions if the diagrams are horizontally reflected. 

For two algebras $A\in\B(a,a)$, $B\in\B(b,b)$ a \textsl{$B$-$A$-bimodule} is an $X\in\B(a,b)$ which is a left $B$-module and a right $A$-module, and the respective algebra actions are compatible in the sense that 
$$
%%%%%%%%%%%%%%%%%%%%%%
\begin{tikzpicture}[very thick,scale=0.75,color=blue!50!black, baseline]

\draw (0,-1) node[right] (X) {{\small$X$}};
\draw (0,1) node[right] (Xu) {{\small$X$}};

\draw[color=green!50!black] (1.0,-1) node[right] (A) {{\small$A$}};
\draw[color=green!50!black] (-0.5,-1) node[left] (B) {{\small$B$}};

\draw (0,-1) -- (0,1); 

\fill[color=green!50!black] (0,-0.3) circle (2.5pt) node (meet) {};
\fill[color=green!50!black] (0,0.3) circle (2.5pt) node (meet) {};

\draw[color=green!50!black] (-0.5,-1) .. controls +(0,0.25) and +(-0.5,-0.5) .. (0,-0.3);
\draw[color=green!50!black] (1.0,-1) .. controls +(0,0.5) and +(0.5,-0.5) .. (0,0.3);
\end{tikzpicture} 
%%%%%%%%%%%%%%%%%%%%%% 
=
%%%%%%%%%%%%%%%%%%%%%%
\begin{tikzpicture}[very thick,scale=0.75,color=blue!50!black, baseline]

\draw (0,-1) node[left] (X) {{\small$X$}};
\draw (0,1) node[left] (Xu) {{\small$X$}};

\draw[color=green!50!black] (0.5,-1) node[right] (A) {{\small$A$}};
\draw[color=green!50!black] (-1,-1) node[left] (B) {{\small$B$}};

\draw (0,-1) -- (0,1); 

\fill[color=green!50!black] (0,-0.3) circle (2.5pt) node (meet) {};
\fill[color=green!50!black] (0,0.3) circle (2.5pt) node (meet) {};

\draw[color=green!50!black] (-1,-1) .. controls +(0,0.5) and +(-0.5,-0.5) .. (0,0.3);
\draw[color=green!50!black] (0.5,-1) .. controls +(0,0.25) and +(0.5,-0.5) .. (0,-0.3);
\end{tikzpicture} 
%%%%%%%%%%%%%%%%%%%%%% 
. 
$$
We denote the space of \textsl{bimodule maps} $\Phi: X\rightarrow Y$ (i.\,e.~simultaneous left $B$- and right $A$-module maps) by $\Hom_{BA}(X,Y)$. If~$A$ and~$B$ are separable Frobenius we have
\be\label{eq:piBA}
\Hom_{BA}(X,Y) = \operatorname{im}(\pi_{BA}) 
\quad\text{with}\quad
\pi_{BA} : \Phi \lmt
%%%%%%%%%%%%%%%%%%%%%%
\begin{tikzpicture}[very thick,scale=0.75,color=blue!50!black, baseline=0cm]

\draw (1.5,-1.75) -- (1.5,1.95); 

\draw (1.5,-1.75) node[right] (X) {{\small$X$}};
\draw (1.5,1.95) node[right] (Xu) {{\small$Y$}};

\fill (1.5,0.7) circle (2.5pt) node[right] (phi) {{\small $\Phi$}};

\fill[color=green!50!black] (1.5,0.3) circle (2.5pt) node (meet) {};
\fill[color=green!50!black] (1.5,1.1) circle (2.5pt) node (meet) {};
\fill[color=green!50!black] (1.5,1.5) circle (2.5pt) node (meet) {};
\fill[color=green!50!black] (1.5,-0.1) circle (2.5pt) node (meet) {};

\draw[color=green!50!black] (1,-0.25) .. controls +(0.0,0.25) and +(-0.25,-0.25) .. (1.5,0.3);
\draw[color=green!50!black] (0.5,-0.25) .. controls +(0.0,0.5) and +(-0.25,-0.25) .. (1.5,1.1);
\draw[-dot-, color=green!50!black] (0.5,-0.25) .. controls +(0,-0.5) and +(0,-0.5) .. (1,-0.25);

\draw[color=green!50!black] (0.75,-1.2) node[Odot] (unit) {}; 
\draw[color=green!50!black] (0.75,-0.6) -- (unit);

\draw[color=green!50!black] (2.25,-0.45) .. controls +(0.0,0.25) and +(0.25,-0.25) .. (1.5,-0.1);
\draw[color=green!50!black] (2.75,-0.45) .. controls +(0.0,1) and +(0.25,-0.25) .. (1.5,1.5);
\draw[-dot-, color=green!50!black] (2.25,-0.45) .. controls +(0,-0.5) and +(0,-0.5) .. (2.75,-0.45);

\draw[color=green!50!black] (2.5,-1.4) node[Odot] (unit) {}; 
\draw[color=green!50!black] (2.5,-0.8) -- (unit);

\end{tikzpicture}
%%%%%%%%%%%%%%%%%%%%%%
\, . 
\ee

Given two algebra maps $\alpha\in\End(A)$ and $\beta\in\End(B)$ we obtain a functor ${}_\beta(-)_\alpha$ on the category of $B$-$A$-bimodules by twisting the module actions. It sends a $B$-$A$-bimodule $X\in\B(a,b)$ to the same underlying 1-morphism in~$\B$, but with left and right action twisted by~$\beta$ and~$\alpha$, respectively, 
\be\label{eq:TwistFunctor}
%%%%%%%%%%%%%%%%%%%%%%
\begin{tikzpicture}[very thick,scale=0.75,color=blue!50!black, baseline]

\draw (0,-1) node[below] (X) {{\small${}_\beta X_\alpha$}};
\draw (0,1) node[above] (Xu) {{\small${}_\beta X_\alpha$}};

\draw[color=green!50!black] (1.0,-1) node[below] (A) {{\small$A$}};
\draw[color=green!50!black] (-0.75,-1) node[below] (B) {{\small$B$}};

\draw (0,-1) -- (0,1); 

\fill[color=green!50!black] (0,-0.3) circle (2.5pt) node (meet) {};
\fill[color=green!50!black] (0,0.3) circle (2.5pt) node (meet) {};

\draw[color=green!50!black] (-0.75,-1) .. controls +(0,0.25) and +(-0.5,-0.25) .. (0,-0.3);
\draw[color=green!50!black] (1.0,-1) .. controls +(0,0.5) and +(0.5,-0.5) .. (0,0.3);
\end{tikzpicture} 
%%%%%%%%%%%%%%%%%%%%%% 
=
%%%%%%%%%%%%%%%%%%%%%%
\begin{tikzpicture}[very thick,scale=0.75,color=blue!50!black, baseline]

\draw (0,-1) node[below] (X) {{\small$X$}};
\draw (0,1) node[above] (Xu) {{\small$X$}};

\draw[color=green!50!black] (1.0,-1) node[below] (A) {{\small$A$}};
\draw[color=green!50!black] (-0.75,-1) node[below] (B) {{\small$B$}};

\draw (0,-1) -- (0,1); 

\fill[color=green!50!black] (0,-0.3) circle (2.5pt) node (meet) {};
\fill[color=green!50!black] (0,0.3) circle (2.5pt) node (meet) {};

\fill[color=green!50!black] (-0.5,-0.6) circle (2.5pt) node[left] (meet) {{\small$\beta$}};
\fill[color=green!50!black] (0.7,-0.35) circle (2.5pt) node[right] (meet) {{\small$\alpha$}};

\draw[color=green!50!black] (-0.75,-1) .. controls +(0,0.25) and +(-0.5,-0.25) .. (0,-0.3);\draw[color=green!50!black] (1.0,-1) .. controls +(0,0.5) and +(0.5,-0.5) .. (0,0.3);
\end{tikzpicture} 
%%%%%%%%%%%%%%%%%%%%%% 
\, , 
\ee
while on morphisms ${}_\beta(-)_\alpha$ acts as the identity.

We shall also need the notion of the tensor product over an algebra~$A$. For a right module~$X$ and a left module~$Y$, the \textsl{tensor product} $X \otimes_A Y$ is defined by a universal property, see e.\,g.~\cite[Sect.\,2.3]{cr1210.6363}. In the spirit of the above constructions~\eqref{eq:piA}, \eqref{eq:piBA} of module maps in terms of projectors, and under the assumption that~$A$ is separable Frobenius and idempotent 2-morphisms split in~$\B$ (which is e.\,g.~true for Landau-Ginzburg models and sigma models), we construct $X \otimes_A Y$ as the image of the projector
$$
\pi_A^{X,Y} = 
%%%%%%%%%%%%%%%%%%%%%%
\begin{tikzpicture}[very thick,scale=0.75,color=blue!50!black, baseline]

\draw (-1,-1) node[left] (X) {{\small$X$}};
\draw (1,-1) node[right] (Xu) {{\small$Y$}};

\draw (-1,-1) -- (-1,1); 
\draw (1,-1) -- (1,1); 

\fill[color=green!50!black] (-1,0.6) circle (2.5pt) node (meet) {};
\fill[color=green!50!black] (1,0.6) circle (2.5pt) node (meet) {};

\draw[-dot-, color=green!50!black] (0.35,-0.0) .. controls +(0,-0.5) and +(0,-0.5) .. (-0.35,-0.0);

\draw[color=green!50!black] (0.35,-0.0) .. controls +(0,0.25) and +(-0.25,-0.25) .. (1,0.6);
\draw[color=green!50!black] (-0.35,-0.0) .. controls +(0,0.25) and +(0.25,-0.25) .. (-1,0.6);

\draw[color=green!50!black] (0,-0.75) node[Odot] (down) {}; 
\draw[color=green!50!black] (down) -- (0,-0.35); 

\end{tikzpicture} 
%%%%%%%%%%%%%%%%%%%%%%  
. 
$$
The fact that this idempotent splits means that there are maps $\xi: X\otimes_A Y \rightarrow X\otimes Y$ and $\vartheta: X\otimes Y \rightarrow X\otimes_A Y$ such that $\vartheta\xi = 1_{X\otimes_A Y}$ and $\xi\vartheta = \pi_A^{X,Y}$. This implicitly defines $X \otimes_A Y$, and for maps $\Phi\in\End(X)$, $\Psi\in\End(Y)$ one sets $\Phi\otimes_A \Psi = \vartheta(\Phi\otimes\Psi)\xi$. 

\begin{remark}\label{rem:AGmodules}
We are now in a position to make good on a promise made at the end of Section~\ref{subsec:conventionalorbi}, namely explain how to describe boundary conditions and defects in the $G$-orbifold of a Landau-Ginzburg model with potential~$W$. Indeed, as shown in~\cite[Sect.\,7.1]{cr1210.6363} they are precisely given by modules and bimodules over~$A_G$. This in particular means that~$A_G$ is the invisible defect and the tensor product $\otimes_{A_G}$ describes fusion in the orbifold theory. 
\end{remark}

\subsection{Generalised twisted sectors}\label{subsec:gentwist}

We continue to work with a pivotal bicategory~$\B$ which for convenience we also assume to be $\C$-linear. Furthermore, for the remainder of this section we fix a separable Frobenius algebra $A\in\B(a,a)$, and we write $I=I_a$. $A$ could be of the form $A=A_G$ as in~\eqref{eq:AG}, but the point is that we want to develop a generalised orbifold theory that does not necessarily depend on a (classical) symmetry group~$G$. Rather, the defect~$A$ itself is to be viewed as encoding a (quantum) symmetry, and we will associate a generalised orbifold theory to the pair $(a,A)$. To achieve this we shall formalise and abstract from the discussion in Section~\ref{sec:ordinaryLGorbs}. 

As the generalisation of the vector space of twisted sectors \textsl{before} orbifold projection we propose 
\be\label{eq:HomIA}
\HIA \, . 
\ee
We now discuss two orbifold projections on this space, intended to generalise the projections to (c,c) fields and RR ground states. Indeed, in the case $A=A_G$ we recover the results of Section~\ref{sec:ordinaryLGorbs}.

\subsubsection*{`Chiral primaries'}

Our first projector acting on $\HIA$ is 
\be\label{eq:ccProj}
\picc : 
%%%%%%%%%%%%%%%%%%%%%%
\begin{tikzpicture}[very thick,scale=0.75,color=green!50!black, baseline]
\fill (0,-0.5) circle (2.5pt) node[left] (D) {{\small $\alpha$}};
\draw (0,-0.5) -- (0,0.6); 
\end{tikzpicture} 
%%%%%%%%%%%%%%%%%%%%%% 
\lmt 
%%%%%%%%%%%%%%%%%%%%%%
\begin{tikzpicture}[very thick,scale=0.75,color=green!50!black, baseline]
\draw (0,0) -- (0,1.25);
\fill (0,0) circle (2.5pt) node[left] {{\small $\alpha$}};
\draw (0,0.8) .. controls +(-0.9,-0.3) and +(-0.9,0) .. (0,-0.8);
\draw (0,-0.8) .. controls +(0.9,0) and +(0.7,-0.1) .. (0,0.4);
\fill (0,-0.8) circle (2.5pt) node {};
\fill (0,0.4) circle (2.5pt) node {};
\fill (0,0.8) circle (2.5pt) node {};
\draw (0,-1.2) node[Odot] (unit) {};
\draw (0,-0.8) -- (unit);
\end{tikzpicture}
%%%%%%%%%%%%%%%%%%%%%%
\!
=
\!
%%%%%%%%%%%%%%%%%%%%%%
\begin{tikzpicture}[very thick,scale=0.75,color=green!50!black, baseline]
\draw (0,0) -- (0,1.25);
\fill (0,0) circle (2.5pt) node[right] {{\small $\alpha$}};
\draw (0,0.8) .. controls +(0.9,-0.3) and +(0.9,0) .. (0,-0.8);
\draw (0,-0.8) .. controls +(-0.9,0) and +(-0.7,-0.1) .. (0,0.4);
\fill (0,-0.8) circle (2.5pt) node {};
\fill (0,0.4) circle (2.5pt) node {};
\fill (0,0.8) circle (2.5pt) node {};
\draw (0,-1.2) node[Odot] (unit) {};
\draw (0,-0.8) -- (unit);
\end{tikzpicture}
%%%%%%%%%%%%%%%%%%%%%%
\ee
where the identity holds since~$A$ is associative. Using associativity again together with the properties of a separable Frobenius algebra we can convince ourselves that $\picc$ is indeed a projector:
\be\label{eq:piccIsProjector}
\picc(\picc(\alpha))
=
%%%%%%%%%%%%%%%%%%%%%%
\begin{tikzpicture}[very thick,scale=0.75,color=green!50!black, baseline]
\draw (0,0) -- (0,2.0);
\fill (0,0) circle (2.5pt) node[left] {{\small $\alpha$}};
\draw (0,0.8) .. controls +(-0.9,-0.3) and +(-0.9,0) .. (0,-0.8);
\draw (0,-0.8) .. controls +(0.9,0) and +(0.7,-0.1) .. (0,0.4);
\fill (0,-0.8) circle (2.5pt) node {};
\fill (0,0.4) circle (2.5pt) node {};
\fill (0,0.8) circle (2.5pt) node {};
\draw (0,-1.2) node[Odot] (unit) {};
\draw (0,-0.8) -- (unit);
\fill (0,1.2) circle (2.5pt) node {};
\fill (0,1.6) circle (2.5pt) node {};
\fill (0,-1.6) circle (2.5pt) node {};
\draw (0,1.6) .. controls +(-1.4,-0.3) and +(-1.4,0) .. (0,-1.6);
\draw (0,-1.6) .. controls +(1.4,0) and +(1.1,-0.1) .. (0,1.2);
\draw (0,-2.0) node[Odot] (unit) {};
\draw (0,-1.6) -- (unit);
\end{tikzpicture}
%%%%%%%%%%%%%%%%%%%%%%
\!\!
\eq^{\eqref{eq:associativeAlgebra}}
\!\!
%%%%%%%%%%%%%%%%%%%%%%
\begin{tikzpicture}[very thick,scale=0.75,color=green!50!black, baseline]
\draw (0,0) -- (0,2.0);
\fill (0,0) circle (2.5pt) node[left] {{\small $\alpha$}};
\draw (0,0.8) .. controls +(-0.9,-0.3) and +(-0.9,0) .. (0,-0.8);
\draw (0,-0.8) .. controls +(0.9,0) and +(0.7,-0.1) .. (0,0.4);
\fill (0,-0.8) circle (2.5pt) node {};
\fill (0,0.4) circle (2.5pt) node {};
\fill (0,0.8) circle (2.5pt) node {};
\draw (0,-1.2) node[Odot] (unit) {};
\draw (0,-0.8) -- (unit);
\fill (0.6,-0.2) circle (2.5pt) node {};
\fill (-0.65,-0.2) circle (2.5pt) node {};
\fill (0,-1.6) circle (2.5pt) node {};
\draw (-0.65,-0.2) .. controls +(-0.7,-0.3) and +(-0.9,0) .. (0,-1.6);
\draw (0,-1.6) .. controls +(0.9,0) and +(0.7,-0.1) .. (0.6,-0.2);
\draw (0,-2.0) node[Odot] (unit) {};
\draw (0,-1.6) -- (unit);
\end{tikzpicture}
%%%%%%%%%%%%%%%%%%%%%%
\!
\eq^{\eqref{eq:FrobeniusProperty}}
%%%%%%%%%%%%%%%%%%%%%%
\begin{tikzpicture}[very thick,scale=0.75,color=green!50!black, baseline]
\draw (0,0) -- (0,2.0);
\fill (0,0) circle (2.5pt) node[left] {{\small $\alpha$}};
\draw (0,0.8) .. controls +(-0.9,-0.3) and +(-0.9,0) .. (0,-0.8);
\draw (0,-0.8) .. controls +(0.9,0) and +(0.7,-0.1) .. (0,0.4);
\fill (0,-0.8) circle (2.5pt) node {};
\fill (0,0.4) circle (2.5pt) node {};
\fill (0,0.8) circle (2.5pt) node {};
\draw (0,-0.8) -- (0,-1.05);
\fill (0,-1.05) circle (2.5pt) node {};
\fill (0,-1.6) circle (2.5pt) node {};
\draw (0,-1.05) .. controls +(-0.4,-0.1) and +(-0.4,0.1) .. (0,-1.6);
\draw (0,-1.6) .. controls +(0.4,0.1) and +(0.4,-0.1) .. (0,-1.05);
\draw (0,-2.0) node[Odot] (unit) {};
\draw (0,-1.6) -- (unit);
\end{tikzpicture}
%%%%%%%%%%%%%%%%%%%%%%
\!
\eq^{\eqref{eq:separability}}
%%%%%%%%%%%%%%%%%%%%%%
\begin{tikzpicture}[very thick,scale=0.75,color=green!50!black, baseline]
\draw (0,0) -- (0,2.0);
\fill (0,0) circle (2.5pt) node[left] {{\small $\alpha$}};
\draw (0,0.8) .. controls +(-0.9,-0.3) and +(-0.9,0) .. (0,-0.8);
\draw (0,-0.8) .. controls +(0.9,0) and +(0.7,-0.1) .. (0,0.4);
\fill (0,-0.8) circle (2.5pt) node {};
\fill (0,0.4) circle (2.5pt) node {};
\fill (0,0.8) circle (2.5pt) node {};
\draw (0,-1.2) node[Odot] (unit) {};
\draw (0,-0.8) -- (unit);
\end{tikzpicture}
%%%%%%%%%%%%%%%%%%%%%%
=
\picc(\alpha)
\, . 
\ee

We will write 
$$
\Hcc = \im(\picc)
$$ 
for the image of this projector to \textsl{(generalised) (c,c) fields}, which has several useful presentations. 
For this we recall that the \textsl{relative centre of $\HIA$ with respect to~$A$} is 
\be\label{eq:ZAdefinition}
\ZA = 
\Big\{ 
\alpha \in \HIA \; \Big| \;\, 
%%%%%%%%%%%%%%%%%%%%%%
\begin{tikzpicture}[very thick,scale=0.45,color=green!50!black, baseline=0.1cm]
\draw[-dot-] (3,0) .. controls +(0,1) and +(0,1) .. (2,0);
\draw (2.5,0.75) -- (2.5,1.5); 
\fill (3,0) circle (4.0pt) node[right] (alpha) {{\small $\alpha$}};
\draw (2,0) -- (2,-0.5); 
\end{tikzpicture} 
%%%%%%%%%%%%%%%%%%%%%% 
=
%%%%%%%%%%%%%%%%%%%%%%
\begin{tikzpicture}[very thick,scale=0.45,color=green!50!black, baseline=0.1cm]
\draw[-dot-] (3,0) .. controls +(0,1) and +(0,1) .. (2,0);
\draw (2.5,0.75) -- (2.5,1.5); 
\fill (2,0) circle (4.0pt) node[left] (alpha) {{\small $\alpha$}};
\draw (3,0) -- (3,-0.5); 
\end{tikzpicture} 
%%%%%%%%%%%%%%%%%%%%%% 
\;\Big\}
\ee
(which also appears in rational CFT, see~\cite[Sect.\,3.4]{tft1}), and that \textsl{Hochschild cohomology} $\operatorname{HH}^\bullet$ is the endomorphism space of the identity functor (see e.\,g.~\cite[Sect.\,4.1]{cw1007.2679}). The characterisation of $\im(\picc)$ in terms of endomorphisms of~$A$ as a bimodule over itself connects our discussion to the generalised orbifolds of~\cite{cr1210.6363} where all generalised twisted sectors are given by bimodule maps. 

\begin{proposition}
We have 
\be\label{eq:Hccequivalents}
\Hcc = \ZA  
\cong 
\End_{AA}(A)
= 
\operatorname{HH}^\bullet(A) 
\, . 
\ee
\end{proposition}

\begin{proof}
The inclusion $\Hcc \subset \ZA$ follows from
\be\label{eq:RelativeCentre1}
%%%%%%%%%%%%%%%%%%%%%%
\begin{tikzpicture}[very thick,scale=0.75,color=green!50!black, baseline=0.1cm]
\draw[-dot-] (3,0) .. controls +(0,1) and +(0,1) .. (2,0);
\draw (2.5,0.75) -- (2.5,1.5); 
\fill (3,0) circle (2.5pt) node[right] (alpha) {};
\draw (2,0) -- (2,-1.0); 
\draw (2.96,0.3) .. controls +(-0.6,-0.3) and +(-0.6,0) .. (3,-0.5);
\draw (3,-0.5) .. controls +(0.9,0) and +(0.7,-0.1) .. (2.82,0.6);
\fill (3,-0.5) circle (2.5pt) node {};
\fill (2.96,0.3) circle (2.5pt) node {};
\fill (2.82,0.6) circle (2.5pt) node {};
\draw (3,-0.9) node[Odot] (unit) {};
\draw (3,-0.5) -- (unit);
\end{tikzpicture} 
%%%%%%%%%%%%%%%%%%%%%% 
\!\!\!
\eq^{\eqref{eq:associativeAlgebra}}
\,
%%%%%%%%%%%%%%%%%%%%%%
\begin{tikzpicture}[very thick,scale=0.75,color=green!50!black, baseline=0.1cm]
\draw[-dot-] (3,0) .. controls +(0,1) and +(0,1) .. (2,0);
\draw (2.5,0.75) -- (2.5,1.5); 
\draw (2,0) -- (2,-1.0); 
\draw[-dot-, color=green!50!black] (3,0) .. controls +(0,-0.5) and +(0,-0.5) .. (2.25,0);
\draw[color=green!50!black] (2.625,-0.9) node[Odot] (unit) {}; 
\draw[color=green!50!black] (2.625,-0.4) -- (unit);
\draw[color=green!50!black] (2.25,0) .. controls +(0,0.1) and +(0.05,0.05) .. (2.02,0.2);
\fill (2.02,0.2) circle (2.5pt) node[right] (alpha) {};
\draw[color=green!50!black] (2.5,0.15) .. controls +(0,0.25) and +(0.05,0.05) .. (2.14,0.55);
\fill (2.14,0.55) circle (2.5pt) node[right] (alpha) {};
\fill (2.5,0.15) circle (2.5pt) node[right] (alpha) {};
\end{tikzpicture} 
%%%%%%%%%%%%%%%%%%%%%% 
\,
\eq^{\eqref{eq:FrobeniusProperty}}
\,
%%%%%%%%%%%%%%%%%%%%%%
\begin{tikzpicture}[very thick,scale=0.75,color=green!50!black, baseline=0.1cm]
\draw[-dot-] (3,0) .. controls +(0,1) and +(0,1) .. (2,0);
\draw (2.5,0.75) -- (2.5,1.5); 
\draw[-dot-] (3,0) .. controls +(0,-1) and +(0,-1) .. (2,0);
\draw (2.5,-0.75) -- (2.5,-1.0); 
\draw[color=green!50!black] (2.5,0.15) .. controls +(0,0.25) and +(0.05,0.05) .. (2.14,0.55);
\fill (2.14,0.55) circle (2.5pt) node[right] (alpha) {};
\fill (2.5,0.15) circle (2.5pt) node[right] (alpha) {};
\end{tikzpicture} 
%%%%%%%%%%%%%%%%%%%%%% 
\,
\eq^{\eqref{eq:FrobeniusProperty}}
\!\!
%%%%%%%%%%%%%%%%%%%%%%
\begin{tikzpicture}[very thick,scale=0.75,color=green!50!black, baseline=0.1cm]
\draw[-dot-] (3,0) .. controls +(0,1) and +(0,1) .. (2,0);
\draw (2.5,0.75) -- (2.5,1.5); 
\fill (2,0) circle (2.5pt) node[left] (alpha) {};
\draw (3,0) -- (3,-1.0); 
\draw (2.04,0.3) .. controls +(-0.6,-0.3) and +(-0.6,0) .. (2,-0.5);
\draw (2,-0.5) .. controls +(0.7,0) and +(-0.4,-0.1) .. (2.82,0.6);
\fill (2,-0.5) circle (2.5pt) node {};
\fill (2.04,0.3) circle (2.5pt) node {};
\fill (2.82,0.6) circle (2.5pt) node {};
\draw (2,-0.9) node[Odot] (unit) {};
\draw (2,-0.5) -- (unit);
\end{tikzpicture} 
%%%%%%%%%%%%%%%%%%%%%% 
\,
\eq^{\eqref{eq:associativeAlgebra}}
\!\!
%%%%%%%%%%%%%%%%%%%%%%
\begin{tikzpicture}[very thick,scale=0.75,color=green!50!black, baseline=0.1cm]
\draw[-dot-] (3,0) .. controls +(0,1) and +(0,1) .. (2,0);
\draw (2.5,0.75) -- (2.5,1.5); 
\fill (2,0) circle (2.5pt) node[left] (alpha) {};
\draw (3,0) -- (3,-1.0); 
\draw (2.04,0.3) .. controls +(-0.6,-0.3) and +(-0.6,0) .. (2,-0.5);
\draw (2,-0.5) .. controls +(0.7,0) and +(0.4,-0.1) .. (2.18,0.6);
\fill (2,-0.5) circle (2.5pt) node {};
\fill (2.04,0.3) circle (2.5pt) node {};
\fill (2.18,0.6) circle (2.5pt) node {};
\draw (2,-0.9) node[Odot] (unit) {};
\draw (2,-0.5) -- (unit);
\end{tikzpicture} 
%%%%%%%%%%%%%%%%%%%%%% 
\, , 
\ee
while for the opposite inclusion we note
\be\label{eq:RelativeCentre2}
%%%%%%%%%%%%%%%%%%%%%%
\begin{tikzpicture}[very thick,scale=0.75,color=green!50!black, baseline]
\draw (0,0) -- (0,1.25);
\fill (0,0) circle (2.5pt) node[left] {{\small $\alpha$}};
\draw (0,0.8) .. controls +(-0.9,-0.3) and +(-0.9,0) .. (0,-0.8);
\draw (0,-0.8) .. controls +(0.9,0) and +(0.7,-0.1) .. (0,0.4);
\fill (0,-0.8) circle (2.5pt) node {};
\fill (0,0.4) circle (2.5pt) node {};
\fill (0,0.8) circle (2.5pt) node {};
\draw (0,-1.2) node[Odot] (unit) {};
\draw (0,-0.8) -- (unit);
\end{tikzpicture}
%%%%%%%%%%%%%%%%%%%%%%
\!\!\!\!
\eq^{\eqref{eq:ZAdefinition}}
\!\!\!\!\!
%%%%%%%%%%%%%%%%%%%%%%
\begin{tikzpicture}[very thick,scale=0.75,color=green!50!black, baseline]
\draw (0,0) -- (0,1.25);
\fill (0,0) circle (2.5pt) node[right] {{\small $\alpha$}};
\draw (0,0.8) .. controls +(-0.9,-0.3) and +(-0.9,0) .. (-0.6,-0.8);
\draw (-0.6,-0.8) .. controls +(0.9,0) and +(-0.7,-0.1) .. (0,0.4);
\fill (-0.6,-0.8) circle (2.5pt) node {};
\fill (0,0.4) circle (2.5pt) node {};
\fill (0,0.8) circle (2.5pt) node {};
\draw (-0.6,-1.2) node[Odot] (unit) {};
\draw (-0.6,-0.8) -- (unit);
\end{tikzpicture}
%%%%%%%%%%%%%%%%%%%%%%
\eq^{\eqref{eq:associativeAlgebra}}_{\eqref{eq:separability}}
\;
%%%%%%%%%%%%%%%%%%%%%%
\begin{tikzpicture}[very thick,scale=0.75,color=green!50!black, baseline]
\draw (0,0) -- (0,1.25);
\fill (0,0) circle (2.5pt) node[right] {{\small $\alpha$}};
\end{tikzpicture}
%%%%%%%%%%%%%%%%%%%%%%
.
\ee

To prove the isomorphism in~\eqref{eq:Hccequivalents} we consider the maps 
\begin{align}
\label{eq:EndAAAtoZA}& \End_{AA}(A) \lra \ZA \, , \quad 
\Phi \lmt
%%%%%%%%%%%%%%%%%%%%%%
\begin{tikzpicture}[very thick,scale=0.75,color=green!50!black, baseline=-0.1cm]
\draw (0,-0.5) node[Odot] (D) {}; 
\draw (D) -- (0,0.6); 
\fill (0,0) circle (2.5pt) node[left] (alpha) {{\small $\Phi$}};
\end{tikzpicture} 
%%%%%%%%%%%%%%%%%%%%%% 
\, , \\
\label{eq:ZAtoEndAAA}& \ZA \lra \End_{AA}(A) \, , \quad
\alpha \lmt 
%%%%%%%%%%%%%%%%%%%%%%
\begin{tikzpicture}[very thick,scale=0.45,color=green!50!black, baseline=0.1cm]
\draw[-dot-] (3,0) .. controls +(0,1) and +(0,1) .. (2,0);
\draw (2.5,0.75) -- (2.5,1.5); 
\fill (3,0) circle (4.0pt) node[right] (alpha) {{\small $\alpha$}};
\draw (2,0) -- (2,-0.5); 
\end{tikzpicture} 
%%%%%%%%%%%%%%%%%%%%%% 
. 
\end{align}
It is easy to check that these are well-defined, and they are indeed mutually inverse, e.\,g. 
$$
\ZA \ni \alpha 
\stackrel{\eqref{eq:ZAtoEndAAA}}{\lmt}
%%%%%%%%%%%%%%%%%%%%%%
\begin{tikzpicture}[very thick,scale=0.45,color=green!50!black, baseline=0.1cm]
\draw[-dot-] (3,0) .. controls +(0,1) and +(0,1) .. (2,0);
\draw (2.5,0.75) -- (2.5,1.5); 
\fill (3,0) circle (4.0pt) node[right] (alpha) {{\small $\alpha$}};
\draw (2,0) -- (2,-0.5); 
\end{tikzpicture} 
%%%%%%%%%%%%%%%%%%%%%% 
\stackrel{\eqref{eq:EndAAAtoZA}}{\lmt}
%%%%%%%%%%%%%%%%%%%%%%
\begin{tikzpicture}[very thick,scale=0.45,color=green!50!black, baseline=0.1cm]
\draw[-dot-] (3,0) .. controls +(0,1) and +(0,1) .. (2,0);
\draw (2.5,0.75) -- (2.5,1.5); 
\fill (3,0) circle (4.0pt) node[right] (alpha) {{\small $\alpha$}};
\draw (2,0) -- (2,-0.45); 
\draw (2,-0.5) node[Odot] (D) {}; 
\end{tikzpicture} 
%%%%%%%%%%%%%%%%%%%%%% 
= 
\alpha \, . 
$$

The description in terms of Hochschild cohomology is now tautological since in the orbifold theory $(a,A)$ the identity functor is given by~$A$ itself, cf.\,\cite[Def.\,4.1]{cr1210.6363}. 
\end{proof}

\medskip

To describe the boundary sector of chiral primaries for the generalised orbifold theory $(a,A)$ let us assume that the bicategory~$\B$ has a trivial object~$0$;\footnote{If~$\B$ is monoidal (like~$\LG$), then~$0$ is the unit object.} for example, in the case of Landau-Ginzburg models the trivial theory is the one with vanishing potential, $W=0$. Then boundary conditions are left $A$-modules $Q,P,\ldots\in \modu(A)$ in $\B(0,a)$, and \textsl{chiral primaries on the boundary} are module maps, i.\,e.~elements of
$$
\Hom_A(Q,P) \, . 
$$
These are precisely those elements $\Phi \in \Hom(Q,P)$ which are invariant under the projector
\be\label{eq:boundaryccproj}
\Phi \lmt 
%%%%%%%%%%%%%%%%%%%%%%
\begin{tikzpicture}[very thick,scale=0.75,color=blue!50!black, baseline]
\draw (0,-1) node[left] (X) {};
\draw (0,1) node[left] (Xu) {};
\draw (0,-1) -- (0,1); 
\fill (0,0) circle (2.5pt) node[right] (phi) {{\small $\Phi$}};
\fill[color=green!50!black] (0,0.6) circle (2.5pt) node (up) {};
\fill[color=green!50!black] (0,-0.6) circle (2.5pt) node (down) {};
\draw[color=green!50!black] (0,0.6) .. controls +(-0.75,0) and +(-0.75,0) .. (0,-0.6);
\end{tikzpicture} 
%%%%%%%%%%%%%%%%%%%%%%
= 
%%%%%%%%%%%%%%%%%%%%%%
\begin{tikzpicture}[very thick,scale=0.75,color=blue!50!black, baseline]
\draw (0,-1) node[left] (X) {};
\draw (0,1) node[left] (Xu) {};
\draw (0,-1) -- (0,1); 
\fill (0,0.2) circle (2.5pt) node[right] (phi) {{\small $\Phi$}};
\fill[color=green!50!black] (0,0.7) circle (2.5pt) node (up) {};
\fill[color=green!50!black] (0,-0.3) circle (2.5pt) node (down) {};
\draw[-dot-, color=green!50!black] (-0.3,-0.6) .. controls +(0,-0.5) and +(0,-0.5) .. (-0.8,-0.6);
\draw[color=green!50!black] (-0.3,-0.6) .. controls +(0,0.2) and +(-0.1,-0.1) .. (0,-0.3);
\draw[color=green!50!black] (-0.8,-0.6) .. controls +(0,0.6) and +(-0.1,-0.1) .. (0,0.7);
\draw[color=green!50!black] (-0.55,-1.3) node[Odot] (unit) {}; 
\draw[color=green!50!black] (-0.55,-1.0) -- (unit);
\end{tikzpicture} 
%%%%%%%%%%%%%%%%%%%%%%
\ee
where the first expression on the right-hand side is simply short-hand for the second expression, cf.\,\eqref{eq:piA} and~\eqref{eq:comoduleAction}.

\subsubsection*{`RR ground states'}

The second projector that we consider on $\HIA$ is
\be\label{eq:piRR}
\pirr : 
%%%%%%%%%%%%%%%%%%%%%%
\begin{tikzpicture}[very thick,scale=0.75,color=green!50!black, baseline]
\fill (0,-0.5) circle (2.5pt) node[left] (D) {{\small $\alpha$}};
\draw (0,-0.5) -- (0,0.6); 
\end{tikzpicture} 
%%%%%%%%%%%%%%%%%%%%%% 
\lmt 
%%%%%%%%%%%%%%%%%%%%%%
\begin{tikzpicture}[very thick,scale=0.75,color=green!50!black, baseline=0.2cm]
\draw (0,0) -- (0,1.3);
\fill (0,0) circle (2.5pt) node[left] {{\small$\alpha$}};
\draw (0,0.8) .. controls +(-0.9,-0.3) and +(-0.9,0) .. (0,-0.8);
\draw[
	decoration={markings, mark=at position 0.83 with {\arrow{<}}}, postaction={decorate}
	]
	 (0,-0.8) .. controls +(0.9,0) and +(0.9,0.8) .. (0,0.4);
\draw[->] (0.01,-0.8) -- (-0.01,-0.8);
\fill (0,0.4) circle (2.5pt) node {};
\fill (0,0.8) circle (2.5pt) node {};
\end{tikzpicture}
%%%%%%%%%%%%%%%%%%%%%%
=
%%%%%%%%%%%%%%%%%%%%%%
\begin{tikzpicture}[very thick,scale=0.75,color=green!50!black, baseline=0.2cm]
\draw (0,0) -- (0,1.3);
\fill (0,0) circle (2.5pt) node[right] {{\small$\alpha$}};
\draw (0,0.8) .. controls +(0.9,-0.3) and +(0.9,0) .. (0,-0.8);
\draw[
	decoration={markings, mark=at position 0.83 with {\arrow{<}}}, postaction={decorate}
	]
	 (0,-0.8) .. controls +(-0.9,0) and +(-0.9,0.8) .. (0,0.4);
\draw[<-] (0.01,-0.8) -- (-0.01,-0.8);
\fill (0,0.4) circle (2.5pt) node {};
\fill (0,0.8) circle (2.5pt) node {};
\end{tikzpicture}
%%%%%%%%%%%%%%%%%%%%%%
.
\ee
%where the identity holds because of the discussion around~\eqref{eq:RRfromccgamma} below.
The identity is checked with the help of the Nakayama automorphism $\gamma_A$: 
\be\label{eq:pirrGamma}
%%%%%%%%%%%%%%%%%%%%%%
\begin{tikzpicture}[very thick,scale=0.75,color=green!50!black, baseline=0.2cm]
\draw (0,0) -- (0,1.3);
\fill (0,0) circle (2.5pt) node[left] {{\small$\alpha$}};
\draw (0,0.8) .. controls +(-0.9,-0.3) and +(-0.9,0) .. (0,-0.8);
\draw[
	decoration={markings, mark=at position 0.83 with {\arrow{<}}}, postaction={decorate}
	]
	 (0,-0.8) .. controls +(0.9,0) and +(0.9,0.8) .. (0,0.4);
\draw[->] (0.01,-0.8) -- (-0.01,-0.8);
\fill (0,0.4) circle (2.5pt) node {};
\fill (0,0.8) circle (2.5pt) node {};
\end{tikzpicture}
%%%%%%%%%%%%%%%%%%%%%%
\eq^{\eqref{eq:FrobeniusProperty}}
%%%%%%%%%%%%%%%%%%%%%%
\begin{tikzpicture}[very thick,scale=0.75,color=green!50!black, baseline=0.2cm]
\draw (0,0) -- (0,1.3);
\fill (0,0) circle (2.5pt) node[left] {{\small$\alpha$}};
\draw (0,0.8) .. controls +(-0.9,-0.3) and +(-0.9,0) .. (0,-0.8);
\draw (0.35,0.12) .. controls +(-0.25,0.1) and +(0.1,-0.1) .. (0,0.4);
\fill (0.35,0.12) circle (2.5pt) node {};
\draw (0.35,-0.2) node[Odot] (unit) {};
\draw (0.35,0.12) -- (unit);
\draw (0.35,0.12) .. controls +(0.32,0.1) and +(-0.3,0) .. (0.8,0.45);
\draw[->] (0.85,0.45) -- (0.86,0.45);
\draw (0.8,0.45) .. controls +(0.5,0) and +(0.9,0) .. (0,-0.8);
\draw[->] (0.01,-0.8) -- (-0.01,-0.8);
\fill (0,0.4) circle (2.5pt) node {};
\fill (0,0.8) circle (2.5pt) node {};
\end{tikzpicture}
%%%%%%%%%%%%%%%%%%%%%%
\eq^{\eqref{eq:FrobeniusProperty}}
%%%%%%%%%%%%%%%%%%%%%%
\begin{tikzpicture}[very thick,scale=0.75,color=green!50!black, baseline=0.2cm]
\draw (0,0) -- (0,1.3);
\fill (0,0) circle (2.5pt) node[left] {{\small$\alpha$}};
\draw (0,0.8) .. controls +(-0.9,-0.3) and +(-0.9,0) .. (0,-1.15);
\draw (0.35,0.12) .. controls +(-0.25,0.1) and +(0.1,-0.1) .. (0,0.4);
\fill (0.35,0.12) circle (2.5pt) node {};
\draw (0.35,-0.2) node[Odot] (unit1) {};
\draw (0.35,0.12) -- (unit1);
\draw (0.35,0.12) .. controls +(0.35,0.1) and +(-0.2,0) .. (0.7,0.45);
\draw[->] (0.77,0.45) -- (0.78,0.45);
\draw (0.7,0.45) .. controls +(0.4,0) and +(0.4,0) .. (0.7,-1.15);
\draw[->] (0.62,-1.15) -- (0.61,-1.15);
\draw (0.7,-1.15) .. controls +(-0.2,0) and +(0.35,-0.1) .. (0.35,-0.82);
\fill (0.35,-0.82) circle (2.5pt) node {};
\draw (0.35,-0.5) node[Odot] (unit2) {};
\draw (0.35,-0.82) -- (unit2);
\draw (0.35,-0.82) .. controls +(-0.25,-0.1) and +(0.25,0) .. (0,-1.15);
\fill (0,-1.15) circle (2.5pt) node {};
\draw (0,-1.47) node[Odot] (unit3) {};
\draw (0,-1.15) -- (unit3);
\fill (0,0.4) circle (2.5pt) node {};
\fill (0,0.8) circle (2.5pt) node {};
\end{tikzpicture}
%%%%%%%%%%%%%%%%%%%%%%
\eq^{\eqref{eq:Nakayama}}                                             
%%%%%%%%%%%%%%%%%%%%%%
\begin{tikzpicture}[very thick,scale=0.75,color=green!50!black, baseline=0.2cm]
\draw (0,0) -- (0,1.3);
\fill (0,0) circle (2.5pt) node[left] {{\small $\alpha$}};
\fill (0.53,0) circle (2.5pt) node[right] {{\small $\gamma_A^{-1}$}};
\draw (0,0.8) .. controls +(-0.9,-0.3) and +(-0.9,0) .. (0,-0.8);
\draw (0,-0.8) .. controls +(0.9,0) and +(0.7,-0.1) .. (0,0.4);
\fill (0,-0.8) circle (2.5pt) node {};
\fill (0,0.4) circle (2.5pt) node {};
\fill (0,0.8) circle (2.5pt) node {};
\draw (0,-1.2) node[Odot] (unit) {};
\draw (0,-0.8) -- (unit);
\end{tikzpicture}
%%%%%%%%%%%%%%%%%%%%%%
\!
\eq^{\eqref{eq:AlgebraMorphism}}
\!
%%%%%%%%%%%%%%%%%%%%%%
\begin{tikzpicture}[very thick,scale=0.75,color=green!50!black, baseline=0.2cm]
\draw (0,0) -- (0,1.3);
\fill (0,0) circle (2.5pt) node[left] {{\small $\alpha$}};
\fill (-0.67,0) circle (2.5pt) node[left] {{\small $\gamma_A$}};
\draw (0,0.8) .. controls +(-0.9,-0.3) and +(-0.9,0) .. (0,-0.8);
\draw (0,-0.8) .. controls +(0.9,0) and +(0.7,-0.1) .. (0,0.4);
\fill (0,-0.8) circle (2.5pt) node {};
\fill (0,0.4) circle (2.5pt) node {};
\fill (0,0.8) circle (2.5pt) node {};
\draw (0,-1.2) node[Odot] (unit) {};
\draw (0,-0.8) -- (unit);
\end{tikzpicture}
%%%%%%%%%%%%%%%%%%%%%%
\!
=
\!
%%%%%%%%%%%%%%%%%%%%%%
\begin{tikzpicture}[very thick,scale=0.75,color=green!50!black, baseline=0.2cm]
\draw (0,0) -- (0,1.3);
\fill (0,0) circle (2.5pt) node[right] {{\small$\alpha$}};
\draw (0,0.8) .. controls +(0.9,-0.3) and +(0.9,0) .. (0,-0.8);
\draw[
	decoration={markings, mark=at position 0.83 with {\arrow{<}}}, postaction={decorate}
	]
	 (0,-0.8) .. controls +(-0.9,0) and +(-0.9,0.8) .. (0,0.4);
\draw[<-] (0.01,-0.8) -- (-0.01,-0.8);
\fill (0,0.4) circle (2.5pt) node {};
\fill (0,0.8) circle (2.5pt) node {};
\end{tikzpicture}
%%%%%%%%%%%%%%%%%%%%%%
\ee
where the last equality follows analogously to the first three played backwards. One can again easily verify that $\pirr$ squares to itself, using the same steps as in~\eqref{eq:piccIsProjector} together with~\eqref{eq:pirrGamma} and the fact that $\gamma_A$ is an algebra morphism:
$$
\pirr(\pirr(\alpha))
=
\!\!
%%%%%%%%%%%%%%%%%%%%%%
\begin{tikzpicture}[very thick,scale=0.75,color=green!50!black, baseline]
\draw (0,0) -- (0,2.0);
\fill (0,0) circle (2.5pt) node[left] {{\small $\alpha$}};
\draw (0,0.8) .. controls +(-0.9,-0.3) and +(-0.9,0) .. (0,-0.8);
\draw (0,-0.8) .. controls +(0.9,0) and +(0.7,-0.1) .. (0,0.4);
\fill (0,-0.8) circle (2.5pt) node {};
\fill (0,0.4) circle (2.5pt) node {};
\fill (0,0.8) circle (2.5pt) node {};
\draw (0,-1.2) node[Odot] (unit) {};
\draw (0,-0.8) -- (unit);
\fill (0,1.2) circle (2.5pt) node {};
\fill (0,1.6) circle (2.5pt) node {};
\fill (0,-1.6) circle (2.5pt) node {};
\draw (0,1.6) .. controls +(-1.8,-0.3) and +(-1.4,0) .. (0,-1.6);
\draw (0,-1.6) .. controls +(1.4,0) and +(1.1,-0.1) .. (0,1.2);
\draw (0,-2.0) node[Odot] (unit) {};
\draw (0,-1.6) -- (unit);
\fill (-0.77,1.2) circle (2.5pt) node {};
\fill (-0.5,0.4) circle (2.5pt) node {};
\fill (-0.57,1.1) circle (0pt) node[above left] {{$\gamma_A$}};
\fill (-0.19,0.32) circle (0pt) node[above left] {{$\gamma_A$}};
\end{tikzpicture}
%%%%%%%%%%%%%%%%%%%%%%
\!\!
=
\!\!
%%%%%%%%%%%%%%%%%%%%%%
\begin{tikzpicture}[very thick,scale=0.75,color=green!50!black, baseline]
\draw (0,0) -- (0,2.0);
\fill (0,0) circle (2.5pt) node[left] {{\small $\alpha$}};
\draw (0,0.8) .. controls +(-0.9,-0.3) and +(-0.9,0) .. (0,-0.8);
\draw (0,-0.8) .. controls +(0.9,0) and +(0.7,-0.1) .. (0,0.4);
\fill (0,-0.8) circle (2.5pt) node {};
\fill (0,0.4) circle (2.5pt) node {};
\fill (0,0.8) circle (2.5pt) node {};
\draw (0,-1.2) node[Odot] (unit) {};
\draw (0,-0.8) -- (unit);
\fill (0.6,-0.2) circle (2.5pt) node {};
\fill (-0.65,-0.2) circle (2.5pt) node {};
\fill (0,-1.6) circle (2.5pt) node {};
\draw (-0.65,-0.2) .. controls +(-0.7,-0.3) and +(-0.9,0) .. (0,-1.6);
\draw (0,-1.6) .. controls +(0.9,0) and +(0.7,-0.1) .. (0.6,-0.2);
\draw (0,-2.0) node[Odot] (unit) {};
\draw (0,-1.6) -- (unit);
\fill (-0.5,0.4) circle (2.5pt) node {};
\fill (-0.3,0.3) circle (0pt) node[above left] {{$\gamma_A$}};
\end{tikzpicture}
%%%%%%%%%%%%%%%%%%%%%%
\!\!
=
\!\!
%%%%%%%%%%%%%%%%%%%%%%
\begin{tikzpicture}[very thick,scale=0.75,color=green!50!black, baseline]
\draw (0,0) -- (0,2.0);
\fill (0,0) circle (2.5pt) node[left] {{\small $\alpha$}};
\draw (0,0.8) .. controls +(-0.9,-0.3) and +(-0.9,0) .. (0,-0.8);
\draw (0,-0.8) .. controls +(0.9,0) and +(0.7,-0.1) .. (0,0.4);
\fill (0,-0.8) circle (2.5pt) node {};
\fill (0,0.4) circle (2.5pt) node {};
\fill (0,0.8) circle (2.5pt) node {};
\draw (0,-1.2) node[Odot] (unit) {};
\draw (0,-0.8) -- (unit);
\fill (-0.5,0.4) circle (2.5pt) node {};
\fill (-0.3,0.3) circle (0pt) node[above left] {{$\gamma_A$}};
\end{tikzpicture}
=
\pirr(\alpha) 
\, .
$$

By definition the bulk space of \textsl{(generalised) RR ground states} is 
\be\label{eq:HARR}
\Hrr = \im(\pirr) \, . 
\ee
The Ramond boundary sector is given by left $A$-modules $Q,P,\ldots\in \modu(A)$ in $\B(0,a)$ and maps $\Phi \in \Hom(Q,P)$ in the image of the projector 
\be\label{eq:RRboundaryprojector}
%%%%%%%%%%%%%%%%%%%%%%
\begin{tikzpicture}[very thick,scale=0.75,color=blue!50!black, baseline]
\draw (0,-1) node[left] (X) {};
\draw (0,1) node[left] (Xu) {};
\draw (0,-1) -- (0,1); 
\fill (0,0) circle (2.5pt) node[right] (phi) {{\small $\Phi$}};
\end{tikzpicture} 
%%%%%%%%%%%%%%%%%%%%%%
\lmt
%%%%%%%%%%%%%%%%%%%%%%
\begin{tikzpicture}[very thick,scale=0.75,color=blue!50!black, baseline]
\draw (0,-1) node[left] (X) {};
\draw (0,1) node[left] (Xu) {};
\draw (0,-1) -- (0,1); 
\fill (0,0) circle (2.5pt) node[right] (phi) {{\small $\Phi$}};
\draw[
	color=green!50!black, 
	decoration={markings, mark=at position 0.84 with {\arrow{>}}, mark=at position 0.16 with {\arrow{>}}}, postaction={decorate}
	]
	 (0,0.6) .. controls +(-1.15,1.15) and +(-1.15,-1.15) .. (0,-0.6);
\fill[color=green!50!black] (0,0.6) circle (2.5pt) node (up) {};
\fill[color=green!50!black] (0,-0.6) circle (2.5pt) node (down) {};
\end{tikzpicture} 
%%%%%%%%%%%%%%%%%%%%%%
=
%%%%%%%%%%%%%%%%%%%%%%
\begin{tikzpicture}[very thick,scale=0.75,color=blue!50!black, baseline]
\draw (0,-1) node[left] (X) {};
\draw (0,1) node[left] (Xu) {};
\draw (0,-1) -- (0,1); 
\fill (0,0) circle (2.5pt) node[right] (phi) {{\small $\Phi$}};
\fill[color=green!50!black] (0,0.6) circle (2.5pt) node (up) {};
\fill[color=green!50!black] (0,-0.6) circle (2.5pt) node (down) {};
\draw[color=green!50!black] (0,0.6) .. controls +(-0.75,0) and +(-0.75,0) .. (0,-0.6);
\fill[color=green!50!black] (-0.56,0) circle (2.5pt) node[left] {{\small $\gamma_A$}};
\end{tikzpicture} 
%%%%%%%%%%%%%%%%%%%%%%
\, ,
\ee
where the identity can be proved similarly as \eqref{eq:pirrGamma}. Hence the \textsl{Ramond boundary space} is 
$$
\Hom_A(Q,{}_{\gamma_A}P) \, , 
$$
where ${}_{\gamma_A}(-)$ denotes the functor that twists the left action by $\gamma_A$, cf.~\eqref{eq:TwistFunctor}.

\begin{example}
In the special case where $\B = \LG$, $W$ has a finite symmetry~$G$ and $A=A_G$ the projector~\eqref{eq:RRboundaryprojector} is explicitly given by
$$
\Phi \lmt 
\sum_{g\in G} 
\;
%%%%%%%%%%%%%%%%%%%%%%
\begin{tikzpicture}[very thick, scale=0.5,color=green!65!black, baseline=-0.35cm]
\draw[color=blue!50!black] (0.5,-4) -- (0.5,3);
\draw[-dot-] (0,0.8) .. controls +(0,-0.5) and +(0,-0.5) .. (-0.75,0.8);
\draw[directedlightgreen, color=green!65!black] (-0.75,0.8) .. controls +(0,0.5) and +(0,0.5) .. (-1.5,0.8);
\draw[-dot-] (0,-1.8) .. controls +(0,0.5) and +(0,0.5) .. (-0.75,-1.8);
\draw[redirectedlightgreen, color=green!65!black] (-0.75,-1.8) .. controls +(0,-0.5) and +(0,-0.5) .. (-1.5,-1.8);
\draw (-1.5,0.8) -- (-1.5,-1.8);
\draw (-0.375,-0.2) node[Odot] (D) {}; 
\draw (-0.375,0.4) -- (D);
\draw (-0.375,-0.8) node[Odot] (E) {}; 
\draw (-0.375,-1.4) -- (E);
\draw (0,0.8) .. controls +(0,0.5) and +(-0.25,-0.25) .. (0.5,1.5);
\draw (0,-1.8) .. controls +(0,-0.5) and +(-0.25,0.25) .. (0.5,-2.5);
\fill[color=blue!50!black] (0.5,2.5) circle (3.3pt) node[right] (alpha) {{\small $(\gamma_g^{(P)})$}};
\fill (0.5,1.5) circle (3.3pt) node[right] (alpha) {{\small ${}_g(\lambda_P)$}};
\fill[color=blue!50!black] (0.5,-0.5) circle (3.3pt) node[right] (alpha) {{\small $\Phi$}};
\fill (0.5,-2.5) circle (3.3pt) node[right] (alpha) {{\small ${}_g(\lambda_Q)^{-1}$}};
\fill[color=blue!50!black] (0.5,-3.5) circle (3.3pt) node[right] (alpha) {{\small $(\gamma_g^{(Q)})^{-1}$}};
\fill[color=green!65!black] (-0.25,0.6) circle (0pt) node[right] (alpha) {{\tiny $g$}};
\fill[color=green!65!black] (-0.25,-1.6) circle (0pt) node[right] (alpha) {{\tiny $g$}};
\fill[color=green!65!black] (-1.7,-0.5) circle (0pt) node[right] (alpha) {{\tiny $g$}};
\end{tikzpicture}
%%%%%%%%%%%%%%%%%%%%%%
\!\!
= 
\frac{1}{|G|} \sum_{g\in G} \det(g) \, \gamma_g^{(P)} \circ {}_g \Phi \circ (\gamma_g^{(Q)})^{-1} \, ,
$$
where we first spelled out the Nakayama automorphism and then wrote it in terms of the determinant, cf.~Example~\ref{ex:A_G}. On the other hand, the boundary projector to chiral primaries~\eqref{eq:boundaryccproj} is
$$
\Phi \lmt 
\frac{1}{|G|} \sum_{g\in G} \gamma_g^{(P)} \circ {}_g \Phi \circ (\gamma_g^{(Q)})^{-1} \, ,
$$
and hence invariance precisely recovers the standard equivariance condition~\eqref{eq:equivmorphs}. 
\end{example}

The space of generalised RR ground states $\Hrr$ also admits several characterisations. To state the result we define the \textsl{twisted relative centre} to be
\be\label{eq:gZAdefinition}
\gZA = 
\Big\{ 
\alpha \in \HIA \; \Big| \,
%%%%%%%%%%%%%%%%%%%%%%
\begin{tikzpicture}[very thick,scale=0.45,color=green!50!black, baseline=0.1cm]
\draw[-dot-] (3,0) .. controls +(0,1) and +(0,1) .. (2,0);
\draw (2.5,0.75) -- (2.5,1.5); 
\fill (3,0) circle (4.0pt) node[right] (alpha) {{\small $\alpha$}};
\fill (2,0) circle (4.0pt) node[left] {{$\gamma_A$}};
\draw (2,0) -- (2,-0.5); 
\end{tikzpicture} 
%%%%%%%%%%%%%%%%%%%%%% 
=
%%%%%%%%%%%%%%%%%%%%%%
\begin{tikzpicture}[very thick,scale=0.45,color=green!50!black, baseline=0.1cm]
\draw[-dot-] (3,0) .. controls +(0,1) and +(0,1) .. (2,0);
\draw (2.5,0.75) -- (2.5,1.5); 
\fill (2,0) circle (4.0pt) node[left] (alpha) {{\small $\alpha$}};
\draw (3,0) -- (3,-0.5); 
\end{tikzpicture} 
%%%%%%%%%%%%%%%%%%%%%% 
\;\Big\}
\, , 
\ee
and we recall from~\cite[Sect.\,4.2]{cw1007.2679} and~\cite{s0702590} that \textsl{Hochschild homology} $\operatorname{HH}_\bullet$ may be defined as the space of maps from the inverse Serre kernel to the identity. As follows from~\cite[Sect.\,4]{cr1210.6363} and will be discussed in more detail in Section~\ref{subsec:ocTFT} below, in our setting the Serre kernel for the theory $(a,A)$ is given by 
\be\label{eq:SerreFunctor}
\Sigma_A = \gA 
\ee
and accordingly $\Sigma_A^{-1} = {}_{\gamma_A^{-1}} A$. 

\begin{proposition}\label{prop:HRRHochschildHomology}
We have 
$$
\Hrr
=
\gZA
\cong 
\Hom_{AA}(A,\gA)
\cong 
\Hom_{AA}({}_{\gamma_A^{-1}} A,A)
=
\operatorname{HH}_\bullet(A) \, . 
$$
\end{proposition}

\begin{proof}
Using~\eqref{eq:pirrGamma} together with the automorphism property of $\gamma_A$ and mimicking the steps in \eqref{eq:RelativeCentre1}, \eqref{eq:RelativeCentre2} we obtain the first equality. The other descriptions of $\Hrr$ follow from the definition of Hochschild homology and the isomorphisms
\begin{align}
&\gZA \lra \Hom_{AA}(A,\gA) \, , \quad
\alpha \lmt 
%%%%%%%%%%%%%%%%%%%%%%
\begin{tikzpicture}[very thick,scale=0.45,color=green!50!black, baseline=0.1cm]
\draw[-dot-] (3,0) .. controls +(0,1) and +(0,1) .. (2,0);
\draw (2.5,0.75) -- (2.5,1.5); 
\fill (2,0) circle (4pt) node[left] (alpha) {{\small $\alpha$}};
\draw (3,0) -- (3,-0.5);  
\end{tikzpicture} 
%%%%%%%%%%%%%%%%%%%%%%
\label{eq:pirrHomAAAgA} \, , \\
& \gZA \lra \Hom_{AA}(A,\Aginv) \, , \quad
\alpha \lmt \quad
%%%%%%%%%%%%%%%%%%%%%%
\begin{tikzpicture}[very thick,scale=0.45,color=green!50!black, baseline=0.1cm]
\draw[-dot-] (3,0) .. controls +(0,1) and +(0,1) .. (2,0);
\draw (2.5,0.75) -- (2.5,1.5); 
\fill (3,0) circle (4pt) node[right] (alpha) {{\small $\alpha$}};
\draw (2,0) -- (2,-0.5); 
\end{tikzpicture} 
%%%%%%%%%%%%%%%%%%%%%%
\label{eq:pirrHomAAAAginv}
\end{align}
with inverses given by 
$
\Phi \mapsto
%%%%%%%%%%%%%%%%%%%%%%
\begin{tikzpicture}[very thick,scale=0.4,color=green!50!black, baseline=-0.15cm]
\draw (0,-0.5) node[Odot] (D) {}; 
\draw (D) -- (0,0.6); 
\fill (0,0) circle (4pt) node[left] (alpha) {{\small $\Phi$}};
\end{tikzpicture} 
%%%%%%%%%%%%%%%%%%%%%% 
\, .
$
\end{proof}

\subsubsection*{Relation between `chiral primaries' and `RR ground states'}

The projectors~\eqref{eq:ccProj}, \eqref{eq:boundaryccproj} and~\eqref{eq:piRR}, \eqref{eq:RRboundaryprojector} to chiral and Ramond sectors are related by the Nakayama automorphisms $\gamma_A, \gamma'_A$ of~\eqref{eq:Nakayama}, \eqref{eq:NakayamaPrime}:  

\begin{lemma}
For $\alpha \in \Hom(I,A)$ the bulk relations are
\begin{align}
\label{eq:RRfromccgamma}
\pirr(\alpha) = & 
%%%%%%%%%%%%%%%%%%%%%%
\begin{tikzpicture}[very thick,scale=0.75,color=green!50!black, baseline=0.2cm]
\draw (0,0) -- (0,1.3);
\fill (0,0) circle (2.5pt) node[left] {{\small$\alpha$}};
\draw (0,0.8) .. controls +(-0.9,-0.3) and +(-0.9,0) .. (0,-0.8);
\draw[
	decoration={markings, mark=at position 0.83 with {\arrow{<}}}, postaction={decorate}
	]
	 (0,-0.8) .. controls +(0.9,0) and +(0.9,0.8) .. (0,0.4);
\draw[->] (0.01,-0.8) -- (-0.01,-0.8);
\fill (0,0.4) circle (2.5pt) node {};
\fill (0,0.8) circle (2.5pt) node {};
\end{tikzpicture}
%%%%%%%%%%%%%%%%%%%%%%
= 
%%%%%%%%%%%%%%%%%%%%%%
\begin{tikzpicture}[very thick,scale=0.75,color=green!50!black, baseline=0.2cm]
\draw (0,0) -- (0,1.3);
\fill (0,0) circle (2.5pt) node[left] {{\small $\alpha$}};
\fill (0.53,0) circle (2.5pt) node[right] {{\small $\gamma_A^{-1}$}};
\draw (0,0.8) .. controls +(-0.9,-0.3) and +(-0.9,0) .. (0,-0.8);
\draw (0,-0.8) .. controls +(0.9,0) and +(0.7,-0.1) .. (0,0.4);
\fill (0,-0.8) circle (2.5pt) node {};
\fill (0,0.4) circle (2.5pt) node {};
\fill (0,0.8) circle (2.5pt) node {};
\draw (0,-1.2) node[Odot] (unit) {};
\draw (0,-0.8) -- (unit);
\end{tikzpicture}
%%%%%%%%%%%%%%%%%%%%%%
\!
=
\!
%%%%%%%%%%%%%%%%%%%%%%
\begin{tikzpicture}[very thick,scale=0.75,color=green!50!black, baseline=0.2cm]
\draw (0,0) -- (0,1.3);
\fill (0,0) circle (2.5pt) node[left] {{\small $\alpha$}};
\fill (-0.67,0) circle (2.5pt) node[left] {{\small $\gamma_A$}};
\draw (0,0.8) .. controls +(-0.9,-0.3) and +(-0.9,0) .. (0,-0.8);
\draw (0,-0.8) .. controls +(0.9,0) and +(0.7,-0.1) .. (0,0.4);
\fill (0,-0.8) circle (2.5pt) node {};
\fill (0,0.4) circle (2.5pt) node {};
\fill (0,0.8) circle (2.5pt) node {};
\draw (0,-1.2) node[Odot] (unit) {};
\draw (0,-0.8) -- (unit);
\end{tikzpicture}
%%%%%%%%%%%%%%%%%%%%%%
\!
=
\!
%%%%%%%%%%%%%%%%%%%%%%
\begin{tikzpicture}[very thick,scale=0.75,color=green!50!black, baseline=0.2cm]
\draw (0,0) -- (0,1.3);
\fill (0,0) circle (2.5pt) node[right] {{\small$\alpha$}};
\draw (0,0.8) .. controls +(0.9,-0.3) and +(0.9,0) .. (0,-0.8);
\draw[
	decoration={markings, mark=at position 0.83 with {\arrow{<}}}, postaction={decorate}
	]
	 (0,-0.8) .. controls +(-0.9,0) and +(-0.9,0.8) .. (0,0.4);
\draw[<-] (0.01,-0.8) -- (-0.01,-0.8);
\fill (0,0.4) circle (2.5pt) node {};
\fill (0,0.8) circle (2.5pt) node {};
\end{tikzpicture}
%%%%%%%%%%%%%%%%%%%%%%
, \\
\label{eq:ccfromRRgamma}
\picc(\alpha) = & 
%%%%%%%%%%%%%%%%%%%%%%
\begin{tikzpicture}[very thick,scale=0.75,color=green!50!black, baseline=0.2cm]
\draw (0,0) -- (0,1.3);
\fill (0,0) circle (2.5pt) node[left] {{\small $\alpha$}};
\draw (0,0.8) .. controls +(-0.9,-0.3) and +(-0.9,0) .. (0,-0.8);
\draw (0,-0.8) .. controls +(0.9,0) and +(0.7,-0.1) .. (0,0.4);
\fill (0,-0.8) circle (2.5pt) node {};
\fill (0,0.4) circle (2.5pt) node {};
\fill (0,0.8) circle (2.5pt) node {};
\draw (0,-1.2) node[Odot] (unit) {};
\draw (0,-0.8) -- (unit);
\end{tikzpicture}
%%%%%%%%%%%%%%%%%%%%%%
=
%%%%%%%%%%%%%%%%%%%%%%
\begin{tikzpicture}[very thick,scale=0.75,color=green!50!black, baseline=0.2cm]
\draw (0,0) -- (0,1.3);
\fill (0,0) circle (2.5pt) node[left] {{\small$\alpha$}};
\fill (0.67,0) circle (2.5pt) node[right] {{\small $\gamma'_A$}};
\draw (0,0.8) .. controls +(-0.9,-0.3) and +(-0.9,0) .. (0,-0.8);
\draw[
	decoration={markings, mark=at position 0.83 with {\arrow{<}}}, postaction={decorate}
	]
	 (0,-0.8) .. controls +(0.9,0) and +(0.9,0.8) .. (0,0.4);
\draw[->] (0.01,-0.8) -- (-0.01,-0.8);
\fill (0,0.4) circle (2.5pt) node {};
\fill (0,0.8) circle (2.5pt) node {};
\end{tikzpicture}
%%%%%%%%%%%%%%%%%%%%%%
=
%%%%%%%%%%%%%%%%%%%%%%
\begin{tikzpicture}[very thick,scale=0.75,color=green!50!black, baseline=0.2cm]
\draw (0,0) -- (0,1.3);
\fill (0,0) circle (2.5pt) node[right] {{\small$\alpha$}};
\fill (-0.67,0) circle (2.5pt) node[left] {{\small $(\gamma'_A)^{-1}$}};
\draw (0,0.8) .. controls +(0.9,-0.3) and +(0.9,0) .. (0,-0.8);
\draw[
	decoration={markings, mark=at position 0.83 with {\arrow{<}}}, postaction={decorate}
	]
	 (0,-0.8) .. controls +(-0.9,0) and +(-0.9,0.8) .. (0,0.4);
\draw[<-] (0.01,-0.8) -- (-0.01,-0.8);
\fill (0,0.4) circle (2.5pt) node {};
\fill (0,0.8) circle (2.5pt) node {};
\end{tikzpicture}
%%%%%%%%%%%%%%%%%%%%%%
,
\end{align}
while in the boundary sector for $\Phi \in \Hom(Q,P)$ we have 
$$
%%%%%%%%%%%%%%%%%%%%%%
\begin{tikzpicture}[very thick,scale=0.75,color=blue!50!black, baseline]
\draw (0,-1) node[left] (X) {};
\draw (0,1) node[left] (Xu) {};
\draw (0,-1) -- (0,1); 
\fill (0,0) circle (2.5pt) node[right] (phi) {{\small $\Phi$}};
\fill[color=green!50!black] (0,0.6) circle (2.5pt) node (up) {};
\fill[color=green!50!black] (0,-0.6) circle (2.5pt) node (down) {};
\draw[color=green!50!black] (0,0.6) .. controls +(-0.75,0) and +(-0.75,0) .. (0,-0.6);
\end{tikzpicture} 
%%%%%%%%%%%%%%%%%%%%%%
= 
%%%%%%%%%%%%%%%%%%%%%%
\begin{tikzpicture}[very thick,scale=0.75,color=blue!50!black, baseline]
\draw (0,-1) node[left] (X) {};
\draw (0,1) node[left] (Xu) {};
\draw (0,-1) -- (0,1); 
\fill (0,0) circle (2.5pt) node[right] (phi) {{\small $\Phi$}};
\draw[
	color=green!50!black, 
	decoration={markings, mark=at position 0.84 with {\arrow{>}}, mark=at position 0.16 with {\arrow{>}}}, postaction={decorate}
	]
	 (0,0.6) .. controls +(-1.15,1.15) and +(-1.15,-1.15) .. (0,-0.6);
\fill[color=green!50!black] (0,0.6) circle (2.5pt) node (up) {};
\fill[color=green!50!black] (0,-0.6) circle (2.5pt) node (down) {};
\fill[color=green!50!black] (-0.86,0) circle (2.5pt) node[left] {{\small $(\gamma'_A)^{-1}$}};
\end{tikzpicture} 
%%%%%%%%%%%%%%%%%%%%%%
\, , \quad 
%%%%%%%%%%%%%%%%%%%%%%
\begin{tikzpicture}[very thick,scale=0.75,color=blue!50!black, baseline]
\draw (0,-1) node[left] (X) {};
\draw (0,1) node[left] (Xu) {};
\draw (0,-1) -- (0,1); 
\fill (0,0) circle (2.5pt) node[right] (phi) {{\small $\Phi$}};
\draw[
	color=green!50!black, 
	decoration={markings, mark=at position 0.84 with {\arrow{>}}, mark=at position 0.16 with {\arrow{>}}}, postaction={decorate}
	]
	 (0,0.6) .. controls +(-1.15,1.15) and +(-1.15,-1.15) .. (0,-0.6);
\fill[color=green!50!black] (0,0.6) circle (2.5pt) node (up) {};
\fill[color=green!50!black] (0,-0.6) circle (2.5pt) node (down) {};
\end{tikzpicture} 
%%%%%%%%%%%%%%%%%%%%%%
= 
%%%%%%%%%%%%%%%%%%%%%%
\begin{tikzpicture}[very thick,scale=0.75,color=blue!50!black, baseline]
\draw (0,-1) node[left] (X) {};
\draw (0,1) node[left] (Xu) {};
\draw (0,-1) -- (0,1); 
\fill (0,0) circle (2.5pt) node[right] (phi) {{\small $\Phi$}};
\fill[color=green!50!black] (0,0.6) circle (2.5pt) node (up) {};
\fill[color=green!50!black] (0,-0.6) circle (2.5pt) node (down) {};
\draw[color=green!50!black] (0,0.6) .. controls +(-0.75,0) and +(-0.75,0) .. (0,-0.6);
\fill[color=green!50!black] (-0.56,0) circle (2.5pt) node[left] {{\small $\gamma_A$}};
\end{tikzpicture} 
%%%%%%%%%%%%%%%%%%%%%%
\, . 
$$
\end{lemma}

\begin{proof}
The first identity~\eqref{eq:RRfromccgamma} was already shown in~\eqref{eq:pirrGamma}, the other ones follow similarly from the definitions of $\gamma_A, \gamma'_A$ and the fact that~$A$ is Frobenius, for example 
$$
%%%%%%%%%%%%%%%%%%%%%%
\begin{tikzpicture}[very thick,scale=0.75,color=green!50!black, baseline=0.2cm]
\draw (0,0) -- (0,1.3);
\fill (0.65,0) circle (2.5pt) node[right] {{$\gamma_A'$}};
\fill (0,0) circle (2.5pt) node[left] {};
\draw (0,0.8) .. controls +(-0.9,-0.3) and +(-0.9,0) .. (0,-0.8);
\draw[
	decoration={markings, mark=at position 0.83 with {\arrow{<}}}, postaction={decorate}
	]
	 (0,-0.8) .. controls +(0.9,0) and +(0.9,0.8) .. (0,0.4);
\draw[->] (0.01,-0.8) -- (-0.01,-0.8);
\fill (0,0.4) circle (2.5pt) node {};
\fill (0,0.8) circle (2.5pt) node {};
\end{tikzpicture}
%%%%%%%%%%%%%%%%%%%%%%
\eq^{\eqref{eq:FrobeniusProperty}}
%%%%%%%%%%%%%%%%%%%%%%
\begin{tikzpicture}[very thick,scale=0.75,color=green!50!black, baseline=0.2cm]
\draw (0,0) -- (0,1.3);
\fill (0,0) circle (2.5pt) node[left] {};
\fill (0,0.4) circle (2.5pt) node {};
\fill (0,0.8) circle (2.5pt) node {};
\draw (0,0.4) .. controls +(0.4,-0.2) and +(0,0.2) .. (0.4,0);
\draw[-dot-] (0.4,0) .. controls +(0,-0.5) and +(0,-0.5) .. (1.0,0);
\draw (0.7,-0.75) node[Odot] (unit) {};
\draw (0.7,-0.3) -- (unit);
\draw[directedgreen] (1.0,0) .. controls +(0,0.5) and +(0,0.5) .. (1.6,0);
\draw[directedgreen] (1.6,0) .. controls +(0,-1.5) and +(0,-1.5) .. (-0.5,0);
\draw (-0.5,0) .. controls +(0,0.4) and +(-0.2,-0.2) .. (0,0.8);
\fill (1.6,0) circle (2.5pt) node[right] {{\small $\gamma'_A$}};
\end{tikzpicture}
%%%%%%%%%%%%%%%%%%%%%%
\eq^{\eqref{eq:NakayamaPrime}}
%%%%%%%%%%%%%%%%%%%%%%
\begin{tikzpicture}[very thick,scale=0.75,color=green!50!black, baseline=0.2cm]
\draw (0,0) -- (0,1.3);
\fill (0,0) circle (2.5pt) node[left] {};
\fill (0,0.4) circle (2.5pt) node {};
\fill (0,0.8) circle (2.5pt) node {};
\draw (0,0.4) .. controls +(0.4,-0.2) and +(0,0.2) .. (0.4,0);
\draw[-dot-] (0.4,0) .. controls +(0,-0.5) and +(0,-0.5) .. (1.0,0);
\draw (0.7,-0.75) node[Odot] (unit) {};
\draw (0.7,-0.3) -- (unit);
\draw[directedgreen] (1.0,0) .. controls +(0,0.5) and +(0,0.5) .. (1.6,0);
\draw[directedgreen] (1.6,0) .. controls +(0,-0.5) and +(0,-0.5) .. (2.2,0);
\draw[-dot-] (2.2,0) .. controls +(0,0.5) and +(0,0.5) .. (2.8,0);
\draw (2.5,0.75) node[Odot] (unit) {};
\draw (2.5,0.4) -- (unit);
\draw (2.8,0) -- (2.8,-1.0);
\draw[-dot-] (2.8,-1.0) .. controls +(0,-0.5) and +(0,-0.5) .. (2.2,-1.0);
\draw (2.5,-1.75) node[Odot] (unit) {};
\draw (2.5,-1.3) -- (unit);
\draw[directedgreen] (2.2,-1.0) .. controls +(0,0.5) and +(0,0.5) .. (1.6,-1.0);
\draw[directedgreen] (1.6,-1) .. controls +(0,-0.95) and +(0,-0.95) .. (-0.5,-1);
\draw (-0.5,-1) -- (-0.5,0);
\draw (-0.5,0) .. controls +(0,0.4) and +(-0.2,-0.2) .. (0,0.8);
\end{tikzpicture}
%%%%%%%%%%%%%%%%%%%%%%
\, 
\eq^{\eqref{eq:FrobeniusProperty}}_{\text{Zorro}}
%%%%%%%%%%%%%%%%%%%%%%
\begin{tikzpicture}[very thick,scale=0.75,color=green!50!black, baseline=0.2cm]
\draw (0,0) -- (0,1.3);
\fill (0,0) circle (2.5pt) node[left] {};
\draw (0,0.8) .. controls +(-0.9,-0.3) and +(-0.9,0) .. (0,-0.8);
\draw (0,-0.8) .. controls +(0.9,0) and +(0.7,-0.1) .. (0,0.4);
\fill (0,-0.8) circle (2.5pt) node {};
\fill (0,0.4) circle (2.5pt) node {};
\fill (0,0.8) circle (2.5pt) node {};
\draw (0,-1.2) node[Odot] (unit) {};
\draw (0,-0.8) -- (unit);
\end{tikzpicture}
%%%%%%%%%%%%%%%%%%%%%%
\!\! . 
$$
\end{proof}

If~$A$ is symmetric then $\Hrr = \Hcc$. Recall from Example~\ref{ex:A_G} that if $A=A_G$ symmetry is equivalent to $\det(g)=1$ for all $g\in G$, and in Section~\ref{subsec:conventionalorbi} we saw that in this case spectral flow indeed provides an isomorphism between (c,c) fields and RR ground states. The identities 
\be\label{eq:HRRHccSymA}
\Hrr = \Hcc 
\quad\text{and}\quad 
\operatorname{HH}_\bullet(A)
=
\operatorname{HH}^\bullet(A)
\quad \text{for symmetric } A
\ee 
generalise this result. To see that~\eqref{eq:HRRHccSymA} is true we simply recall that for symmetric algebras the Nakayama automorphism is the identity, and thus \eqref{eq:pirrGamma} implies $\pirr=\picc$ in this case.
That Hochschild cohomology and homology are isomorphic is familiar in the case of Calabi-Yau manifolds, and in conjunction with the discussion in Section~\ref{subsec:conventionalorbi} we find that the orbifolding defect~$A$ being symmetric is analogous to the Calabi-Yau condition in our setting. 

Furthermore we have that
\be\label{eq:AsymSF}
A\text{ is symmetric} 
\quad\Longleftrightarrow\quad 
\pirr
\Big(\;
%%%%%%%%%%%%%%%%%%%%%%
\begin{tikzpicture}[very thick,scale=0.75,color=green!50!black, baseline=-0.08cm]
\draw (0,-0.5) node[Odot] (D) {}; 
\draw (D) -- (0,0.6); 
\end{tikzpicture} 
%%%%%%%%%%%%%%%%%%%%%% 
\;\Big)
=
%%%%%%%%%%%%%%%%%%%%%%
\begin{tikzpicture}[very thick,scale=0.75,color=green!50!black, baseline=-0.08cm]
\draw (0,-0.5) node[Odot] (D) {}; 
\draw (D) -- (0,0.6); 
\end{tikzpicture} 
%%%%%%%%%%%%%%%%%%%%%% 
\, , 
\ee
interpreted as the condition that the state 
$
%%%%%%%%%%%%%%%%%%%%%%
\begin{tikzpicture}[very thick,scale=0.4,color=green!50!black, baseline=-0.12cm]
\draw (0,-0.5) node[Odot] (D) {}; 
\draw (D) -- (0,0.6); 
\end{tikzpicture} 
%%%%%%%%%%%%%%%%%%%%%% 
$ 
associated to the (generalised) spectral flow operator survives the orbifold projection in the RR sector. This immediately follows from the characterisation in terms of the twisted relative centre, since 
$
%%%%%%%%%%%%%%%%%%%%%%
\begin{tikzpicture}[very thick,scale=0.35,color=green!50!black, baseline=-0.12cm]
\draw (0,-0.5) node[Odot] (D) {}; 
\draw (D) -- (0,0.6); 
\end{tikzpicture} 
%%%%%%%%%%%%%%%%%%%%%% 
\in \gZA 
$ 
iff 
$
\gamma_A=1_A
$ 
which in turn is equivalent to~$A$ being symmetric.

\subsection{Open/closed topological field theory}\label{subsec:ocTFT}

We will now explore to what extent a generalised orbifold gives rise to the structure of an open/closed topological field theory. For this we start by recalling the algebraic description of a 2d open/closed TFT \cite{l0010269,ms0609042}, which consists of the following data:
 \begin{itemize}
    \item a commutative Frobenius algebra $C$ with trace map $\langle - \rangle_C: C \rightarrow \C$,
    \item a Calabi-Yau category $\mathcal{O}$, i.\,e.~a $\C$-linear category with trace maps $\langle - \rangle_Q: \End_{\mathcal{O}}(Q) \rightarrow \C$ for every $Q\in \mathcal{O}$, such that the induced pairings $\langle -,- \rangle_{Q,P}: \Hom_{\mathcal{O}}(Q,P) \times \Hom_{\mathcal{O}}(P,Q) \rightarrow \C$ given by $\langle \Psi_1,\Psi_2 \rangle_{Q,P}=\langle \Psi_1 \Psi_2 \rangle_Q$ are symmetric and nondegenerate,
    \item for every $Q\in\mathcal{O}$ a \textsl{bulk-boundary map} $\beta_Q:C\rightarrow\End_{\mathcal{O}}(Q)$ and a \textsl{boundary-bulk map} $\beta^Q: \End_{\mathcal{O}}(Q) \rightarrow C$.
 \end{itemize}
These data have to satisfy the axioms:
 \begin{itemize}
  \item the bulk-boundary maps $\beta_Q$ are morphisms of algebras (which we always take to be unital and associative) whose image is in the centre of $\End_{\mathcal{O}}(Q)$,
  \item $\beta_Q$ and $\beta^Q$ are adjoint with respect to the nondegenerate pairings on~$C$ and $\End_{\mathcal{O}}(Q)$ in the sense that 
$
\left\langle \phi \, \beta^Q(\Psi) \right\rangle_C = \left\langle \beta_Q(\phi) \, \Psi \right\rangle_Q
$ 
for all $\phi \in C$ and $\Psi \in \End_{\mathcal{O}}(Q)$,
  \item the \textsl{Cardy condition}
$$
\Big\langle \beta^Q(\Phi) \, \beta^P(\Psi) \Big\rangle_C = \tr({}_\Psi m_\Phi)
$$ 
is satisfied for every $\Phi \in \End_{\mathcal{O}}(Q)$ and $\Psi \in \End_{\mathcal{O}}(P)$, where ${}_\Psi m_\Phi(\alpha)= \Psi \circ \alpha \circ \Phi$ and the trace is taken over $\Hom_{\mathcal{O}}(Q,P)$.
 \end{itemize}

Examples of 2d open/closed TFTs include the A- and B-twisted sigma models~\cite{cw1007.2679, a1001.4593, g1304.7312}, as well as unorbifolded affine B-twisted Landau-Ginzburg models~\cite{pv1002.2116, m0912.1629, cm1208.1481}. Furthermore, it was shown in~\cite{cr1210.6363} that every generalised Landau-Ginzburg orbifold $(W,A)$ with~$A$ a \textsl{symmetric} separable Frobenius algebra in the bicategory of Landau-Ginzburg models gives rise to an open/closed TFT. In what follows, we will determine how much of this structure remains when the symmetry condition on~$A$ is dropped.

\subsubsection*{Products in bulk and boundary sectors}

We start with the observation that the vector space $\HIA$ can be made into an algebra (over the complex numbers) by endowing it with the multiplication
\be\label{eq:alphatimesbeta}
\alpha \cdot \beta = 
%%%%%%%%%%%%%%%%%%%%%%
\begin{tikzpicture}[very thick,scale=0.75,color=green!50!black, baseline=0.4cm]
\draw[-dot-] (3,0) .. controls +(0,1) and +(0,1) .. (2,0);
\draw (2.5,0.75) -- (2.5,1.5); 
\fill (2,0) circle (2.5pt) node[left] (alpha) {{\small $\alpha$}};
\fill (3,0) circle (2.5pt) node[right] (beta) {{\small $\beta$}};
\end{tikzpicture} 
%%%%%%%%%%%%%%%%%%%%%% 
. 
\ee
Associativity and unitality are immediate consequences of the fact that~$A$ is associative and unital. The product on $\HIA$ descends to the subspaces $\Hcc=\ZA$ and $\Hrr=\gZA$, so in particular we have the bilinear maps 
\begin{align*}
 \Hcc \times \Hcc \lra \Hcc 
 \, , \quad
 \Hrr \times \Hcc \lra \Hrr \, .
\end{align*}
That these products indeed map to the subspaces of $\Hom(I,A)$ indicated can be seen from
\be\label{eq:productsinHIA}
%%%%%%%%%%%%%%%%%%%%%%
\begin{tikzpicture}[very thick,scale=0.75,color=green!50!black, baseline=0.4cm]
\draw[-dot-] (3,0) .. controls +(0,1) and +(0,1) .. (2,0);
\fill (2,0) circle (2.5pt) node[below] (alpha) {{\small $\alpha$}};
\fill (3,0) circle (2.5pt) node[below] (beta) {{\small $\beta$}};
\draw[-dot-] (2.5,0.75) .. controls +(0,1) and +(0,1) .. (3.5,0.75);
\draw (3,1.5) -- (3,2.25);
\draw (3.5,-0.75) -- (3.5,0.75);
\end{tikzpicture}
%%%%%%%%%%%%%%%%%%%%%%
\, 
\eq^{\eqref{eq:associativeAlgebra}}
%%%%%%%%%%%%%%%%%%%%%%
\begin{tikzpicture}[very thick,scale=0.75,color=green!50!black, baseline=0.4cm]
\fill (2.5,0) circle (2.5pt) node[below] (alpha) {{\small $\alpha$}};
\draw (2.5,0) -- (2.5,0.75);
\draw[-dot-] (2.5,0.75) .. controls +(0,1) and +(0,1) .. (3.5,0.75);
\fill (3,0) circle (2.5pt) node[below] (beta) {{\small $\beta$}};
\draw[-dot-] (3,0) .. controls +(0,1) and +(0,1) .. (4,0);
\draw (3,1.5) -- (3,2.25);
\draw (4,-0.75) -- (4,0);
\end{tikzpicture}
%%%%%%%%%%%%%%%%%%%%%%
\,
=
%%%%%%%%%%%%%%%%%%%%%%
\begin{tikzpicture}[very thick,scale=0.75,color=green!50!black, baseline=0.4cm]
\fill (2.5,0) circle (2.5pt) node[below] (alpha) {{\small $\alpha$}};
\draw (2.5,0) -- (2.5,0.75);
\draw[-dot-] (2.5,0.75) .. controls +(0,1) and +(0,1) .. (3.5,0.75);
\fill (4,0) circle (2.5pt) node[below] (beta) {{\small $\beta$}};
\draw[-dot-] (3,0) .. controls +(0,1) and +(0,1) .. (4,0);
\draw (3,1.5) -- (3,2.25);
\draw (3,-0.75) -- (3,0);
\end{tikzpicture}
%%%%%%%%%%%%%%%%%%%%%%
\,
\eq^{\eqref{eq:associativeAlgebra}}
%%%%%%%%%%%%%%%%%%%%%%
\begin{tikzpicture}[very thick,scale=0.75,color=green!50!black, baseline=0.4cm]
\draw[-dot-] (3,0) .. controls +(0,1) and +(0,1) .. (2,0);
\fill (3.5,0) circle (2.5pt) node[below] (alpha) {{\small $\beta$}};
\fill (2,0) circle (2.5pt) node[below] {{\small $\alpha$}};
\draw[-dot-] (2.5,0.75) .. controls +(0,1) and +(0,1) .. (3.5,0.75);
\draw (3,1.5) -- (3,2.25);
\draw (3.5,0) -- (3.5,0.75);
\draw (3,-0.75) -- (3,0);
\end{tikzpicture}
%%%%%%%%%%%%%%%%%%%%%%
=
%%%%%%%%%%%%%%%%%%%%%%
\begin{tikzpicture}[very thick,scale=0.75,color=green!50!black, baseline=0.4cm]
\draw[-dot-] (3,0) .. controls +(0,1) and +(0,1) .. (2,0);
\fill (3.5,0) circle (2.5pt) node[below] (alpha) {{\small $\beta$}};
\fill (3,0) circle (2.5pt) node[below] (beta) {{\small $\alpha$}};
\fill (2,0) circle (2.5pt) node[right] {{\small $\gamma$}};
\draw[-dot-] (2.5,0.75) .. controls +(0,1) and +(0,1) .. (3.5,0.75);
\draw (3,1.5) -- (3,2.25);
\draw (3.5,0) -- (3.5,0.75);
\draw (2,-0.75) -- (2,0);
\end{tikzpicture}
%%%%%%%%%%%%%%%%%%%%%%
\eq^{\eqref{eq:associativeAlgebra}}
%%%%%%%%%%%%%%%%%%%%%%
\begin{tikzpicture}[very thick,scale=0.75,color=green!50!black, baseline=0.4cm]
\fill (4,0) circle (2.5pt) node[below] (alpha) {{\small $\beta$}};
\draw (2.5,0) -- (2.5,0.75);
\draw[-dot-] (2.5,0.75) .. controls +(0,1) and +(0,1) .. (3.5,0.75);
\fill (3,0) circle (2.5pt) node[below] (beta) {{\small $\alpha$}};
\fill (2.5,0) circle (2.5pt) node[left] {{\small $\gamma$}};
\draw[-dot-] (3,0) .. controls +(0,1) and +(0,1) .. (4,0);
\draw (3,1.5) -- (3,2.25);
\draw (2.5,-0.75) -- (2.5,0);
\end{tikzpicture}
%%%%%%%%%%%%%%%%%%%%%%
,
\ee
where $\beta \in \Hcc$, and $\gamma=1$ if $\alpha \in \Hcc$, while for $\alpha \in\Hrr$ we take $\gamma=\gamma_A$. In the unlabelled steps we used the characterisations in terms of the (twisted) relative centre \eqref{eq:ZAdefinition} and \eqref{eq:gZAdefinition}. Thus we have: 

\begin{lemma}\label{lem:H0H1algebramodule}
The space $\Hcc$ is a commutative algebra, while $\Hrr$ is a module over $\Hcc$. 
\end{lemma}

This is precisely the algebraic structure present on the spaces of (ordinary) (c,c) fields and RR ground states in any superconformal field theory. Note also that the induced actions
\begin{align*}
\operatorname{HH}^\bullet (A)\times \operatorname{HH}^\bullet (A)
\lra
\operatorname{HH}^\bullet (A)
\, , \quad
\operatorname{HH}_\bullet (A)\times \operatorname{HH}^\bullet (A)
\lra
\operatorname{HH}_\bullet (A)
\end{align*}
generalise the usual cup and cap products on Hochschild cohomology and homology of ordinary algebras.

\begin{remark}\label{rem:extendedTFT}
For any integer~$n$ let us define the space
$$
\mathcal H_n^A = \im 
\Bigg(
%%%%%%%%%%%%%%%%%%%%%%
\begin{tikzpicture}[very thick,scale=0.65,color=green!50!black, baseline]
\fill (0,-0.5) circle (2.5pt) node[left] (D) {{\small $\alpha$}};
\draw (0,-0.5) -- (0,0.6); 
\end{tikzpicture} 
%%%%%%%%%%%%%%%%%%%%%% 
\lmt 
%%%%%%%%%%%%%%%%%%%%%%
\begin{tikzpicture}[very thick,scale=0.65,color=green!50!black, baseline=0cm]
\draw (0,0) -- (0,1.3);
\fill (0,0) circle (2.5pt) node[below] {{\small $\alpha$}};
\fill (-0.67,0) circle (2.5pt) node[left] {{\small $\gamma_A^n$}};
\draw (0,0.8) .. controls +(-0.9,-0.3) and +(-0.9,0) .. (0,-0.8);
\draw (0,-0.8) .. controls +(0.9,0) and +(0.7,-0.1) .. (0,0.4);
\fill (0,-0.8) circle (2.5pt) node {};
\fill (0,0.4) circle (2.5pt) node {};
\fill (0,0.8) circle (2.5pt) node {};
\draw (0,-1.2) node[Odot] (unit) {};
\draw (0,-0.8) -- (unit);
\end{tikzpicture}
%%%%%%%%%%%%%%%%%%%%%%
\Bigg)
=
\Big\{ 
\alpha \in \HIA \; \Big| \,
%%%%%%%%%%%%%%%%%%%%%%
\begin{tikzpicture}[very thick,scale=0.45,color=green!50!black, baseline=0.1cm]
\draw[-dot-] (3,0) .. controls +(0,1) and +(0,1) .. (2,0);
\draw (2.5,0.75) -- (2.5,1.5); 
\fill (3,0) circle (4.0pt) node[right] (alpha) {{\small $\alpha$}};
\fill (2,0) circle (4.0pt) node[left] {{$\gamma_A^n$}};
\draw (2,0) -- (2,-0.5); 
\end{tikzpicture} 
%%%%%%%%%%%%%%%%%%%%%% 
=
%%%%%%%%%%%%%%%%%%%%%%
\begin{tikzpicture}[very thick,scale=0.45,color=green!50!black, baseline=0.1cm]
\draw[-dot-] (3,0) .. controls +(0,1) and +(0,1) .. (2,0);
\draw (2.5,0.75) -- (2.5,1.5); 
\fill (2,0) circle (4.0pt) node[left] (alpha) {{\small $\alpha$}};
\draw (3,0) -- (3,-0.5); 
\end{tikzpicture} 
%%%%%%%%%%%%%%%%%%%%%% 
\;\Big\}
$$
which we can think of as performing an $n$-fold twist of the defect~$A$ before wrapping it around a bulk field. 
As in~\eqref{eq:productsinHIA} one finds that the product~\eqref{eq:alphatimesbeta} in $\HIA$ restricts to a map
\be\label{eq:HmHn}
\mathcal H^A_m \times \mathcal H^A_n \lra \mathcal H^A_{m+n} \, , 
\ee
and we recover Lemma~\ref{lem:H0H1algebramodule} in the special cases $\mathcal H^A_0 = \Hcc$ and $\mathcal H^A_1 = \Hrr$. 

A product of the form~\eqref{eq:HmHn} is familiar from \textsl{extended, framed} TFT~\cite{l0905.0465}. In the two-dimensional case there are as many inequivalent 2-framed circles~$S_n^1$ as there are integers $n\in \Z \cong \pi_1(\operatorname{SO}(2))$. The TFT maps these circles to spaces~$\mathcal H_n$, and the pair of pants bordism gives rise to products $\mathcal H_m \times \mathcal H_n \rightarrow \mathcal H_{m+n}$. Furthermore, the special spaces~$\mathcal H_0$ and~$\mathcal H_1$ are always given by (suitably defined) Hochschild cohomology and homology, respectively, and in general the framing of~$S_n^1$ is encoded in the space~$\mathcal H_n$ via the $n$-th power of the Serre automorphism, see~\cite[Prop.\,4.2.3 \& Rem.\,4.2.5]{l0905.0465}. 

Since in our setting of generalised orbifolds $(a,A)$ the Serre functor is given by twisting with the Nakayama automorphism~$\gamma_A$, we produced exactly the structure described in the previous paragraph. It thus seems natural to expect that generalised orbifolds give rise (via the cobordism hypothesis of \cite{bdCobordismHypothesis, l0905.0465}) to fully extended framed TFTs with values in the equivariant completion~$\Beq$ of~\cite{cr1210.6363}.\footnote{For this to make sense one has to assume at least that the original bicategory~$\B$ is symmetric monoidal -- a condition that is e.\,g.~satisfied for topological sigma models and Landau-Ginzburg models.} 
The details of this will be discussed elsewhere. For now we only note that from this perspective $\mathcal H^A_0 = \Hcc$ and $\mathcal H^A_1 = \Hrr$ continue to have natural interpretations even if the underlying theory encoded in~$\B$ is not of supersymmetric origin (in which case the notion of chiral primaries and RR ground states is spurious). 
\end{remark}

\begin{remark}
Instead of allowing framings for bordisms, one can also consider other additional geometric structures, leading to other flavours of topological field theories. 
As explained in~\cite{NovakRunkel} every separable Frobenius algebra whose Nakayama automorphism squares to the identity gives rise to a two-dimensional \textsl{spin} TFT. In the supersymmetric setting this should be regarded as an orbifold by the group $G = \langle (-1)^F \rangle$ between the NS and~R sector. 
\end{remark}

\medskip

On a boundary given by a left $A$-module $Q\in \modu(A)$ the multiplication before orbifold projection is given simply by the composition of maps in $\End(Q)$. The induced product on the images of the boundary orbifold projectors then endows the boundary chiral sector $\End_A(Q)$ with the structure of a (not necessarily commutative) algebra. One also easily verifies that the boundary Ramond sector $\Hom_A(Q,{}_{\gamma_A}Q)$ forms a right module over this algebra. 

Finally, we define the bulk-boundary and boundary-bulk maps in the generalised orbifold theory by
\be
\beta_Q(\phi) = 
%%%%%%%%%%%%%%%%%%%%%%
\begin{tikzpicture}[very thick,scale=0.75,color=blue!50!black, baseline]
\draw (0,-1) -- (0,1); 
\draw (0,-1) node[right] (X) {{\small$Q$}};
\draw (0,1) node[right] (Xu) {{\small$Q$}};
\fill[color=green!50!black] (-0.5,-0.5) circle (2.5pt) node[left] {{\small$\phi$}};
\fill[color=green!50!black]  (0,0.6) circle (2.5pt) node (meet) {};
\draw[color=green!50!black]  (-0.5,-0.5) .. controls +(0,0.5) and +(-0.5,-0.5) .. (0,0.6);
\end{tikzpicture} 
%%%%%%%%%%%%%%%%%%%%%% 
, \qquad \
\beta^Q(\Psi) = 
%%%%%%%%%%%%%%%%%%%%%%
\begin{tikzpicture}[very thick,scale=0.6,color=blue!50!black, baseline]
\draw (0,0) circle (1.5);
\draw[->, very thick] (0.100,-1.5) -- (-0.101,-1.5) node[above] {};
\draw[->, very thick] (-0.100,1.5) -- (0.101,1.5) node[below] {}; 
\fill (202.5:1.5) circle (3.3pt) node[right] {{\small$\Psi$}};
\fill (270:1.5) circle (0pt) node[above] {{\small$Q$}};
\fill[color=green!50!black] (157.5:1.5) circle (3.3pt) node[right] {};
\draw[-dot-, color=green!50!black] (-3,-0.5) .. controls +(0,-1) and +(0,-1) .. (-2,-0.5);
\draw[color=green!50!black] (-2,-0.5) .. controls +(0,0.5) and +(-0.5,-0.5) .. (157.5:1.5);
\draw[color=green!50!black] (-2.5,-2) node[Odot] (unit) {};
\draw[color=green!50!black] (-2.5,-1.3) -- (unit);
\draw[color=green!50!black] (-3,-0.5) -- (-3,2);
\end{tikzpicture} 
%%%%%%%%%%%%%%%%%%%%%%
\, . 
\ee
These operators are a priori defined on the unprojected bulk and boundary spaces; by appropriate abuse of notation we will use the same symbols to denote their restrictions to the orbifold projections. 

\begin{proposition}
\begin{enumerate}
\item 
The bulk-boundary map gives an algebra morphism 
$$
\beta_Q: 
\Hcc \lra \End_A(Q)
$$
with image in the centre of $\End_A(Q)$, while the restriction to $\Hrr$ gives
$$
\beta_Q: 
\Hrr \lra \Hom_A(Q,{}_{\gamma_A}Q)
$$
satisfying
$\beta_Q(\alpha \cdot \alpha') = \beta_Q(\alpha) \beta_Q(\alpha')
$
for $\alpha \in \Hrr$ and $\alpha' \in \Hcc$. 
\item The boundary-bulk map $\beta^Q$ restricted to $\End_A(Q)$ gives a map
$$
\beta^Q: 
\End_A(Q) \lra \Hrr \, . 
$$
\end{enumerate}
\end{proposition}

\begin{proof}
Part~(i) is an easy direct check, and part~(ii) follows because the image of $\beta^Q$ is invariant under the projector~\eqref{eq:piRR} to $\Hrr$: 
\begin{align}
%%%%%%%%%%%%%%%%%%%%%%
\begin{tikzpicture}[very thick,scale=0.6,color=blue!50!black, baseline]
\draw (0,0) circle (1.5);
\draw[->, very thick] (0.100,-1.5) -- (-0.101,-1.5) node[above] {};
\draw[->, very thick] (-0.100,1.5) -- (0.101,1.5) node[below] {}; 
\fill (202.5:1.5) circle (3.3pt) node[right] {{\small$\Psi$}};
\fill (270:1.5) circle (0pt) node[above] {{\small$Q$}};
\fill[color=green!50!black] (157.5:1.5) circle (3.3pt) node[right] {};
\draw[color=green!50!black] (-2.15,2.4) .. controls +(0,-0.5) and +(-0.5,0.5) .. (157.5:1.5);
\fill[color=green!50!black] (152:1.85) circle (3.3pt) node {};
\fill[color=green!50!black] (141.5:2.57) circle (3.3pt) node {};
\draw[domain=35.04:41.51,variable=\t,smooth,samples=75,color=green!50!black] 
plot ({-\t r}:{0.0015*\t * \t});
\draw[->, very thick,color=green!50!black] (-0.2,-2.3) -- (-0.201,-2.3) node[above] {};
\draw[->, very thick,color=green!50!black] (0.2,1.95) -- (0.3,1.95) node[below] {};
\end{tikzpicture} 
%%%%%%%%%%%%%%%%%%%%%%
&\eq^{(1)}
%%%%%%%%%%%%%%%%%%%%%%
\begin{tikzpicture}[very thick,scale=0.6,color=blue!50!black, baseline]
\draw (0,0) circle (1.5);
\draw[->, very thick] (0.100,-1.5) -- (-0.101,-1.5) node[above] {};
\draw[->, very thick] (-0.100,1.5) -- (0.101,1.5) node[below] {}; 
\fill (202.5:1.5) circle (3.3pt) node[right] {{\small$\Psi$}};
\fill (270:1.5) circle (0pt) node[above] {{\small$Q$}};
\fill[color=green!50!black] (157.5:1.5) circle (3.3pt) node[right] {};
\draw[color=green!50!black] (-2.15,2.4) .. controls +(0,-0.5) and +(-0.5,0.5) .. (157.5:1.5);
\fill[color=green!50!black] (126:1.5) circle (3.3pt) node {};
\fill[color=green!50!black] (141.5:2.57) circle (3.3pt) node {};
\draw[domain=36.2:41.51,variable=\t,smooth,samples=75,color=green!50!black] 
plot ({-\t r}:{0.0015*\t * \t});
\draw[color=green!50!black] (126:1.5) .. controls +(0.2,0.3) and +(-0.75,0) .. (0.2,1.95);
\draw[->, very thick,color=green!50!black] (-0.2,-2.3) -- (-0.201,-2.3) node[above] {};
\draw[->, very thick,color=green!50!black] (0.2,1.95) -- (0.3,1.95) node[below] {};
\end{tikzpicture} 
%%%%%%%%%%%%%%%%%%%%%%
\eq^{\text{Zorro}}
%%%%%%%%%%%%%%%%%%%%%%
\begin{tikzpicture}[very thick,scale=0.6,color=blue!50!black, baseline]
\draw[directed] (1.5,0) .. controls +(0,-2) and +(0,-2) .. (-1.5,0);
\draw (-1.5,0) .. controls +(0,2.15) and +(0,1.15) .. (-0.5,0.5);
\draw[directed] (-0.5,0.5) .. controls +(0,-1) and +(0,-1) .. (0.5,0.5);
\draw (0.5,0.5) .. controls +(0,1.15) and +(0,2.15) .. (1.5,0);
\draw[->, very thick] (-1.01,1.5) -- (-0.95,1.5) node[below] {};
\draw[->, very thick] (1,1.5) -- (1.05,1.5) node[above] {};
\fill (202.5:1.5) circle (3.3pt) node[right] {{\small$\Psi$}};
\fill (270:1.5) circle (0pt) node[above] {{\small$Q$}};
\fill[color=green!50!black] (157.5:1.58) circle (3.3pt) node[right] {};
\draw[color=green!50!black] (-2.15,2.4) .. controls +(0,-0.5) and +(-0.5,0.5) .. (157.5:1.5);
\fill[color=green!50!black] (0.65,1.18) circle (3.3pt) node {};
\fill[color=green!50!black] (141.5:2.57) circle (3.3pt) node {};
\draw[domain=36.2:41.51,variable=\t,smooth,samples=75,color=green!50!black] 
plot ({-\t r}:{0.0015*\t * \t});
\draw[color=green!50!black] (0.65,1.18) .. controls +(-0.5,0.3) and +(-0.8,0) .. (0.2,1.95);
\draw[->, very thick,color=green!50!black] (-0.2,-2.3) -- (-0.201,-2.3) node[above] {};
\draw[->, very thick,color=green!50!black] (0.2,1.95) -- (0.3,1.95) node[below] {};
\end{tikzpicture} 
%%%%%%%%%%%%%%%%%%%%%%
\eq^{(2)}
%%%%%%%%%%%%%%%%%%%%%%
\begin{tikzpicture}[very thick,scale=0.6,color=blue!50!black, baseline]
\draw[directed] (-1.5,0) .. controls +(0,2) and +(0,2) .. (1.5,0);
\draw (-1.5,0) .. controls +(0,-2.15) and +(0,-1.15) .. (-0.5,-0.5);
\draw[directed] (0.5,-0.5) .. controls +(0,1) and +(0,1) .. (-0.5,-0.5);
\draw (0.5,-0.5) .. controls +(0,-1.15) and +(0,-2.15) .. (1.5,0);
\draw[<-, very thick] (-1.06,-1.5) -- (-1,-1.5) node[below] {};
\draw[<-, very thick] (0.94,-1.5) -- (0.96,-1.5) node[above] {};
\fill (-1.48,-0.574) circle (3.3pt) node[right] {{\small$\Psi$}};
\fill (90:0.25) circle (0pt) node[above] {{\small$Q$}};
\fill[color=green!50!black] (157.5:1.5) circle (3.3pt) node[right] {};
\draw[color=green!50!black] (-2.15,2.4) .. controls +(0,-0.5) and +(-0.5,0.5) .. (157.5:1.5);
\fill[color=green!50!black] (141.5:2.57) circle (3.3pt) node {};
\fill[color=green!50!black] (0.65,-1.18) circle (3.3pt) node {};
\draw[directedgreen] (0.65,-1.18) .. controls +(-0.2,0.45) and +(0,0.45) .. (-0.15,-1.18);
\draw[directedgreen] (-0.15,-1.18) .. controls +(0,-1.35) and +(0,-1.35) .. (1.85,-1.18);
\draw[directedgreen] (1.85,-1.18) .. controls +(0,0.45) and +(0,0.45) .. (2.65,-1.18);
\draw[<-, very thick,color=green!50!black] (0.8,-3.25) -- (0.901,-3.25) node[below] {};
\draw[color=green!50!black] (2.65,-1.18) .. controls +(0,-2.7) and +(1.7,-2.8) .. (-2.125,-1.18);
\draw[domain=40.3:41.51,variable=\t,smooth,samples=75,color=green!50!black] 
plot ({-\t r}:{0.0015*\t * \t});
\end{tikzpicture} 
%%%%%%%%%%%%%%%%%%%%%%
\!\!\nonumber
\\&\eq^{\text{Zorro}}
\;\;
%%%%%%%%%%%%%%%%%%%%%%
\!\!\!\!\!
\begin{tikzpicture}[very thick,scale=0.6,color=blue!50!black, baseline]
\draw (0,0) circle (1.5);
\draw[->, very thick] (0.100,-1.5) -- (-0.101,-1.5) node[above] {};
\draw[->, very thick] (-0.100,1.5) -- (0.101,1.5) node[below] {}; 
\fill (202.5:1.5) circle (3.3pt) node[right] {{\small$\Psi$}};
\fill (270:1.5) circle (0pt) node[above] {{\small$Q$}};
\fill[color=green!50!black] (157.5:1.5) circle (3.3pt) node[right] {};
\draw[color=green!50!black] (-2.15,2.4) .. controls +(0,-0.5) and +(-0.5,0.5) .. (157.5:1.5);
\fill[color=green!50!black] (234:1.5) circle (3.3pt) node {};
\fill[color=green!50!black] (141.5:2.57) circle (3.3pt) node {};
\draw[color=green!50!black] (234:1.5) .. controls +(-2,0) and +(-0.6,-0.6) .. (141.5:2.57);
\end{tikzpicture} 
%%%%%%%%%%%%%%%%%%%%%%
\eq^{\eqref{eq:ModuleMap}}_{(3)}
\!\!\!\!
%%%%%%%%%%%%%%%%%%%%%%
\begin{tikzpicture}[very thick,scale=0.6,color=blue!50!black, baseline]
\draw (0,0) circle (1.5);
\draw[->, very thick] (0.100,-1.5) -- (-0.101,-1.5) node[above] {};
\draw[->, very thick] (-0.100,1.5) -- (0.101,1.5) node[below] {}; 
\fill (202.5:1.5) circle (3.3pt) node[right] {{\small$\Psi$}};
\fill (270:1.5) circle (0pt) node[above] {{\small$Q$}};
\fill[color=green!50!black] (157.5:1.5) circle (3.3pt) node[right] {};
\draw[color=green!50!black] (-2.15,2.4) .. controls +(0,-0.5) and +(-0.5,0.5) .. (157.5:1.5);
\fill[color=green!50!black] (152:1.85) circle (3.3pt) node {};
\fill[color=green!50!black] (141.5:2.57) circle (3.3pt) node {};
\draw[color=green!50!black] (152:1.85) .. controls +(-0.3,0) and +(-0.5,-0.5) .. (141.5:2.57);
\end{tikzpicture} 
%%%%%%%%%%%%%%%%%%%%%%
\eq^{\eqref{eq:separability}}
%%%%%%%%%%%%%%%%%%%%%%
\begin{tikzpicture}[very thick,scale=0.6,color=blue!50!black, baseline]
\draw (0,0) circle (1.5);
\draw[->, very thick] (0.100,-1.5) -- (-0.101,-1.5) node[above] {};
\draw[->, very thick] (-0.100,1.5) -- (0.101,1.5) node[below] {}; 
\fill (202.5:1.5) circle (3.3pt) node[right] {{\small$\Psi$}};
\fill (270:1.5) circle (0pt) node[above] {{\small$Q$}};
\fill[color=green!50!black] (157.5:1.5) circle (3.3pt) node[right] {};
\draw[color=green!50!black] (-2.15,2.4) .. controls +(0,-0.5) and +(-0.5,0.5) .. (157.5:1.5);
\end{tikzpicture} 
%%%%%%%%%%%%%%%%%%%%%%
\; . \label{eq:loopinvariant}
\end{align}
Here we first used~\eqref{eq:comoduleAction} to express~$\beta^Q$ in terms of the comodule action, whose defining property is then used in steps (1) and (3), and in step (2) we used the comodule version of~\eqref{eq:pivvotal}.
\end{proof}

\subsubsection*{Nondegenerate pairings and Serre functors}

To study nondegenerate pairings in the bulk and boundary sectors, we need to make further assumptions on the bicategory $\B$. Namely, in addition to pivotality and $\C$-linearity of $\B$, we will assume that for every object $a\in \B$ there is a linear map
$$
\langle- \rangle_{a}: \End(I_a) \lra \C \, .
$$
This can be interpreted as the bulk trace in the unorbifolded theory (under the identification of $\End(I_a)$ with the space of bulk fields). Using this, we define \textsl{defect pairings}
\be\label{eq:DefectPairing}
\langle -,- \rangle_{X,Y} : \Hom(Y,X) \times \Hom(X,Y) \lra \C 
\, , \quad 
\big\langle \Psi_1, \Psi_2 \big\rangle_{X,Y} = 
\left\langle \, 
%%%%%%%%%%%%%%%%%%%%%%
\begin{tikzpicture}[very thick,scale=0.6,color=blue!50!black, baseline]
\draw (0,0) circle (1.5);
\draw[<-, very thick] (0.100,-1.5) -- (-0.101,-1.5) node[above] {}; 
\draw[<-, very thick] (-0.100,1.5) -- (0.101,1.5) node[below] {}; 
\fill (-22.5:1.5) circle (3.3pt) node[left] {{\small$\Psi_2$}};
\fill (22.5:1.5) circle (3.3pt) node[left] {{\small$\Psi_1$}};
\fill (270:1.5) circle (0pt) node[above] {{\small$X$}};
\end{tikzpicture} 
%%%%%%%%%%%%%%%%%%%%%% 
\,  \right\rangle_{\raisemath{10pt}{\!\!\!a}}
\ee
for any $X,Y\in \B(a,b)$. Note that we have
\be\label{eq:DefectParingSymmetry}
\big\langle \Psi_1\Psi_2, \Psi_3 \big\rangle_{X,Y} 
= \big\langle \Psi_1,\Psi_2 \Psi_3 \big\rangle_{X,Y} 
= \big\langle \Psi_3 \Psi_1, \Psi_2 \big\rangle_{X,Y} \, , 
\ee
where the second identity follows from \eqref{eq:rightPhi} and pivotality. We will assume that these pairings are nondegenerate (which is guaranteed if our bicategory~$\B$, i.\,e.~the unorbifolded theory, indeed describes a TFT), and furthermore we will require the following property for every $\Psi\in\End(X)$:
\be\label{eq:oriflip}
\left\langle \, 
%%%%%%%%%%%%%%%%%%%%%%
\begin{tikzpicture}[very thick,scale=0.6,color=blue!50!black, baseline]
\draw (0,0) circle (1.5);
\draw[<-, very thick] (0.100,-1.5) -- (-0.101,-1.5) node[above] {}; 
\draw[<-, very thick] (-0.100,1.5) -- (0.101,1.5) node[below] {}; 
\fill (0:1.5) circle (3.3pt) node[left] {{\small$\Psi$}};
\fill (270:1.5) circle (0pt) node[above] {{\small$X$}};
\end{tikzpicture} 
%%%%%%%%%%%%%%%%%%%%%% 
\, \right\rangle_{\raisemath{10pt}{\!\!\!a}}
=
\left\langle \, 
%%%%%%%%%%%%%%%%%%%%%%
\begin{tikzpicture}[very thick,scale=0.6,color=blue!50!black, baseline]
\draw (0,0) circle (1.5);
\draw[->, very thick] (0.100,-1.5) -- (-0.101,-1.5) node[above] {}; 
\draw[->, very thick] (-0.100,1.5) -- (0.101,1.5) node[below] {}; 
\fill (180:1.5) circle (3.3pt) node[right] {{\small$\Psi$}};
\fill (270:1.5) circle (0pt) node[above] {{\small$X$}};
\end{tikzpicture} 
%%%%%%%%%%%%%%%%%%%%%% 
\, \right\rangle_{\raisemath{10pt}{\!\!\!b}} \, . 
\ee

We now want to understand to what extent the nondegenerate pairings \eqref{eq:DefectPairing} give rise to nondegenerate pairings on the images of the orbifold projectors. For this consider two separable Frobenius algebras $A\in \B(a,a)$ and $B \in \B(b,b)$, and $B$-$A$-bimodules $X,Y\in \B(a,b)$. Our analysis relies on the results of~\cite[Sect.\,4.4]{cr1210.6363}, where it was shown that \eqref{eq:DefectPairing} satisfying the above conditions descends to a nondegenerate pairing
\be\label{eq:nondegeneratePairing}
\langle -,- \rangle_{X,Y} : 
\Hom_{BA}(Y,X)\times \Hom_{BA}(X,{}_{\gamma_B}Y_{\gamma_A^{-1}}) \lra \C
\ee
involving twists by the Nakayama automorphisms $\gamma_A, \gamma_B$. 
Since this map is natural in both~$X$ and~$Y$ by construction, nondegeneracy means: 

\begin{proposition}\label{prop:SerreDefect}
$
S_{BA} = {}_{\gamma_B}(-)_{\gamma_A^{-1}}
$ 
is the Serre functor on the category of $B$-$A$-bimodules, with inverse given by $S_{BA}^{-1}={}_{\gamma_B^{-1}}(-)_{\gamma_A}$.\footnote{This is a generalisation of the following standard result. Let~$k$ be a field, $\mathcal A$ a finite-dimensional $k$-algebra of finite global dimension, and $\modu(\mathcal A)$ the category of finitely generated $\mathcal A$-modules. Then the \textsl{Nakayama functor} $\mathcal N_{\mathcal A} = (-)\otimes_{\mathcal A} \mathcal A^*$ is a bijection between projectives and injectives, and if~$\mathcal A$ is Frobenius we have $\mathcal N_{\mathcal A} \cong {}_{\gamma_{\mathcal A}}(-)$, see e.\,g.~\cite[Prop.\,IV.3.13]{syFrobeniusAlgebrasI}. Furthermore, $\mathcal N_{\mathcal A}$ induces the Serre functor on $\mathds{D}^{\text{b}}(\modu(\mathcal A))$~\cite[Sect.\,I.3]{rvdb9911242}.} 
\end{proposition}

Another way to phrase this is that the Serre kernel 
\be
\Sigma_A = \gA \; \congscript^{\gamma_A^{-1}} \; A_{\gamma_A^{-1}},
\ee
induces the action of $S_{BA}$ in the sense that $S_{BA}(X)\cong \Sigma_B \otimes_B X \otimes_A \Sigma_A$. 

\medskip

Let us now apply Proposition~\ref{prop:SerreDefect} to bulk and boundary pairings in the orbifold theory. 
On the boundary, given by left $A$-modules $Q, P \in \B(0,a)$, the natural choice for the pairing is \eqref{eq:DefectPairing} with $X=Q$ and $Y=P$: 
$$
\langle -,- \rangle_{Q,P} : \Hom(P,Q) \times \Hom(Q,P) \lra \C 
\, , \quad 
\big\langle \Psi_1, \Psi_2 \big\rangle_{Q,P} = 
%%%%%%%%%%%%%%%%%%%%%%
\begin{tikzpicture}[very thick,scale=0.6,color=blue!50!black, baseline]
\draw (0,0) circle (1.5);
\draw[<-, very thick] (0.100,-1.5) -- (-0.101,-1.5) node[above] {}; 
\draw[<-, very thick] (-0.100,1.5) -- (0.101,1.5) node[below] {}; 
\fill (-22.5:1.5) circle (3.3pt) node[left] {{\small$\Psi_2$}};
\fill (22.5:1.5) circle (3.3pt) node[left] {{\small$\Psi_1$}};
\fill (270:1.5) circle (0pt) node[above] {{\small$Q$}};
\end{tikzpicture} 
%%%%%%%%%%%%%%%%%%%%%% 
$$
where here and in the following we can safely omit the bracket for the trivial theory~$0$; furthermore, we will write $\langle -,- \rangle_{Q} = \langle -,- \rangle_{Q,Q}$. 

It immediately follows from Proposition~\ref{prop:SerreDefect} that on the boundary the chiral sector is perfectly paired with the Ramond sector: 

\begin{corollary}
The pairing $\langle -,- \rangle_{Q,P}$ is nondegenerate when restricted to 
$
\Hom_A(P,Q) \times \Hom_A(Q,{}_{\gamma_A}P) 
$. 
\end{corollary}

\medskip

In the bulk there is a natural choice for a pairing $\langle-,-\rangle_{(a,A)}:\HIA \times \HIA \rightarrow \C$ given by
\be\label{eq:genorbbulkpairing}
\big\langle\alpha_1,\alpha_2 \big\rangle_{(a,A)} 
= \left\langle
%%%%%%%%%%%%%%%%%%%%%%
\begin{tikzpicture}[very thick,scale=0.75,color=green!50!black, baseline=0.4cm]
\draw[-dot-] (3,0) .. controls +(0,1) and +(0,1) .. (2,0);
\fill (2,0) circle (2.5pt) node[left] (alpha) {{\small $\alpha_1$}};
\fill (3,0) circle (2.5pt) node[right] (beta) {{\small $\alpha_2$}};
\draw (2.5,1.3) node[Odot] (unit) {};
\draw (2.5,0.75) -- (unit); 
\end{tikzpicture} 
%%%%%%%%%%%%%%%%%%%%%% 
\right\rangle_{\raisemath{4pt}{\!\!\!a}} \, .
\ee
Note that $\langle\alpha_1 \cdot \alpha_2, \alpha_3\rangle_{(a,A)} = \langle\alpha_1, \alpha_2 \cdot \alpha_3\rangle_{(a,A)}$ by associativity of $A$, and that the pairing is symmetric only up to a twist by $\gamma_A$, i.\,e.
\be\label{eq:bulkPairingTwistedSymmetry}
\big\langle\alpha_1,\alpha_2 \big\rangle_{(a,A)} = \big\langle\gamma_A \circ \alpha_2,\alpha_1 \big\rangle_{(a,A)} 
\ee
for $\alpha_1,\alpha_2 \in \Hrr$, as follows from \eqref{eq:gZAdefinition}. 

\begin{proposition}
The pairing~\eqref{eq:genorbbulkpairing} is nondegenerate when restricted to $\Hrr$. 
\end{proposition}

\begin{proof}
Let us assume that $0\neq \alpha_1 \in \Hrr$. We want to find $\alpha_2\in\Hrr$ such that $\langle\alpha_1 , \alpha_2\rangle_{(a,A)} \neq 0$. To achieve this, first consider the image of $\alpha_1$ under the map \eqref{eq:pirrHomAAAgA}, which we denote by $\bar \alpha_1 \in \Hom_{AA}(A,\gA)$. Using the nondegeneracy of~\eqref{eq:nondegeneratePairing} with $X=\gA$ and $Y=A$ we can then find $\bar \alpha_2 \in \Hom_{AA}(A,\Aginv)$ such that $\langle \bar \alpha_1, \bar \alpha_2 \rangle_{A,\gA} \neq 0$. Taking $\alpha_2$ to be the preimage of $\bar \alpha_2$ under \eqref{eq:pirrHomAAAAginv}, we conclude the proof by observing
$$
\left\langle \, 
%%%%%%%%%%%%%%%%%%%%%%
\begin{tikzpicture}[very thick,scale=0.6,color=green!50!black, baseline]
\draw (0,0) circle (1.5);
\draw[<-, very thick] (0.100,-1.5) -- (-0.101,-1.5) node[above] {}; 
\draw[<-, very thick] (-0.100,1.5) -- (0.101,1.5) node[below] {}; 
\fill (-22.5:1.5) circle (3.3pt) node[left] {{\small$\bar \alpha_2$}};
\fill (22.5:1.5) circle (3.3pt) node[left] {{\small$\bar \alpha_1$}};
\end{tikzpicture} 
%%%%%%%%%%%%%%%%%%%%%% 
\,  \right\rangle_{\raisemath{10pt}{\!\!\!a}}
=
\left\langle \, 
%%%%%%%%%%%%%%%%%%%%%%
\begin{tikzpicture}[very thick,scale=0.6,color=green!50!black, baseline]
\draw (0,0) circle (1.5);
\draw[<-, very thick] (0.100,-1.5) -- (-0.101,-1.5) node[above] {}; 
\draw[<-, very thick] (-0.100,1.5) -- (0.101,1.5) node[below] {}; 
\fill (-22.5:1.5) circle (3.3pt) node[left] {};
\fill (22.5:1.5) circle (3.3pt) node[left] {};
\draw (-22.5:1.5) .. controls +(0.3,0) and +(0,0.3) .. (1.9,-1.25);
\draw (22.5:1.5) .. controls +(-0.3,0) and +(0,0.3) .. (0.9,-0.1);
\fill (0.9,-0.1) circle (3.3pt) node[left] (beta) {{\small $\alpha_1$}};
\fill (1.9,-1.25) circle (3.3pt) node[right] (beta) {{\small $\alpha_2$}};
\end{tikzpicture} 
%%%%%%%%%%%%%%%%%%%%%% 
\!\!
\right\rangle_{\raisemath{10pt}{\!\!\!a}}
\eq^{\eqref{eq:FrobeniusProperty}}
\left\langle
%%%%%%%%%%%%%%%%%%%%%%
\begin{tikzpicture}[very thick,scale=0.75,color=green!50!black, baseline=0.2cm]
\fill (0,0) circle (2.5pt) node[left] {{\small$\alpha_1$\!\!}};
\draw (0.19,0.6) .. controls +(0.9,-0.3) and +(0.9,0) .. (0,-0.8);
\draw[
	decoration={markings, mark=at position 0.83 with {\arrow{<}}}, postaction={decorate}
	]
	 (0,-0.8) .. controls +(-0.9,0) and +(-0.9,0.8) .. (0.05,0.35);
\draw[<-] (0.01,-0.8) -- (-0.01,-0.8);
\fill (0.05,0.35) circle (2.5pt) node {};
\fill (0.19,0.6) circle (2.5pt) node {};
\draw[-dot-] (0,0) .. controls +(0,1) and +(0,1) .. (1,0);
\fill (1,0) circle (2.5pt) node[right] (beta) {{\small $\alpha_2$}};
\draw (0.5,1.3) node[Odot] (unit) {};
\draw (0.5,0.75) -- (unit); 
\end{tikzpicture}
%%%%%%%%%%%%%%%%%%%%%%
\right\rangle_{\raisemath{7pt}{\!\!\!a}}
=
\left\langle
\!
%%%%%%%%%%%%%%%%%%%%%%
\begin{tikzpicture}[very thick,scale=0.75,color=green!50!black, baseline=0.4cm]
\draw[-dot-] (3,0) .. controls +(0,1) and +(0,1) .. (2,0);
\fill (2,0) circle (2.5pt) node[left] (alpha) {{\small $\alpha_1$}};
\fill (3,0) circle (2.5pt) node[right] (beta) {{\small $\alpha_2$}};
\draw (2.5,1.3) node[Odot] (unit) {};
\draw (2.5,0.75) -- (unit); 
\end{tikzpicture} 
%%%%%%%%%%%%%%%%%%%%%%
\!
\right\rangle_{\raisemath{4pt}{\!\!\!a}} \, .
$$
\end{proof}

We remark that the existence of a nondegenerate pairing on $\Hrr$ is in agreement with the identification of this space with Hochschild homology, see Proposition~\ref{prop:HRRHochschildHomology}. Indeed, a related result of~\cite{s0710.1937} states that Hochschild homology of any smooth and proper dg-category has a canonical nondegenerate pairing. On the other hand, it is only Hochschild cohomology that comes with a natural algebra structure, as discussed after~\eqref{eq:alphatimesbeta}. In our setting the pairing~\eqref{eq:genorbbulkpairing} on $\Hom(I,A)$ also induces a pairing on Hochschild cohomology $\Hcc \subset \Hom(I,A)$, but this pairing has no reason to be nondegenerate in general.

\subsubsection*{Cardy condition and adjunction}

The two remaining properties of open/closed TFT that we need to inspect are the adjunction between $\beta_Q$ and $\beta^Q$ as well as the Cardy condition. 
We first turn to the former, which is slightly complicated by the fact that the bulk pairing $\langle-,-\rangle_{(a,A)}$ is not symmetric for $\gamma_A\neq 1_A$. Indeed, bulk-boundary and boundary-bulk maps are ``twisted adjoint'' in the following sense. 

\begin{proposition}
For every $\alpha \in \Hrr$ and $\Psi \in \End_A(Q)$ we have 
\begin{align}
\Big\langle \alpha, \; \beta^Q(\Psi) \Big\rangle_{(a,A)}&= \Big \langle \beta_Q(\alpha),\; \Psi \Big\rangle_Q \, ,  \label{eq:adjunctionBobu1} \\ 
\Big\langle \beta^Q(\Psi),\; \alpha \Big\rangle_{(a,A)}&= \Big \langle \Psi,\; \beta_Q(\gamma_A \circ \alpha) \Big\rangle_Q \label{eq:adjunctionBobu2} \, .
\end{align}
\end{proposition}

\begin{proof}
For the first identity we compute
\begin{align*}
\Big\langle \alpha, \beta^Q (\Psi) \Big\rangle_{(a,A)} 
& 
=
\left\langle \, 
%%%%%%%%%%%%%%%%%%%%%%
\begin{tikzpicture}[very thick,scale=0.6,color=green!50!black, baseline]
\draw[color=blue!50!black] (0,0) circle (1.5);
\draw[->, color=blue!50!black, very thick] (0.100,-1.5) -- (-0.101,-1.5) node[above] {}; 
\draw[->, color=blue!50!black, very thick] (-0.100,1.5) -- (0.101,1.5) node[below] {}; 
\fill[color=blue!50!black] (202.5:1.5) circle (3.3pt) node[right] {{\small$\Psi$}};
\fill (155:1.5) circle (3.3pt) node[left] {};
\draw[color=green!50!black] (-2,0) .. controls +(0,0.4) and +(-0.4,-0.2) .. (155:1.5);
\draw[-dot-] (-2,0) .. controls +(0,-1) and +(0,-1) .. (-3,0);
\draw[-dot-] (-3,0) .. controls +(0,1) and +(0,1) .. (-4,0);
\draw (-4,0) -- (-4,-1);
\fill (-4,-1) circle (3.3pt) node[right] {{\small$\alpha$}};
\draw (-2.5,-1.2) node[Odot] (D) {}; 
\draw (-2.5,-0.8) -- (D);
\draw (-3.5,1.2) node[Odot] (D) {}; 
\draw (-3.5,0.8) -- (D);
\end{tikzpicture} 
%%%%%%%%%%%%%%%%%%%%%% 
\,  \right\rangle_{\raisemath{10pt}{\!\!\!a}}
\eq^{\eqref{eq:FrobeniusProperty}}
\left\langle 
%%%%%%%%%%%%%%%%%%%%%%
\begin{tikzpicture}[very thick,scale=0.6,color=green!50!black, baseline]
\draw[color=blue!50!black] (0,0) circle (1.5);
\draw[->, color=blue!50!black, very thick] (0.100,-1.5) -- (-0.101,-1.5) node[above] {}; 
\draw[->, color=blue!50!black, very thick] (-0.100,1.5) -- (0.101,1.5) node[below] {}; 
\fill[color=blue!50!black] (202.5:1.5) circle (3.3pt) node[right] {{\small$\Psi$}};
\fill (155:1.5) circle (3.3pt) node[left] {};
\draw[color=green!50!black] (-2,0) .. controls +(0,0.4) and +(-0.4,-0.2) .. (155:1.5);
\fill (-2,0) circle (3.3pt) node[left] {{\small$\alpha$}};
\end{tikzpicture} 
%%%%%%%%%%%%%%%%%%%%%% 
\,  \right\rangle_{\raisemath{10pt}{\!\!\!a}}
\\
& 
\eq^{\eqref{eq:oriflip}}
%%%%%%%%%%%%%%%%%%%%%%
\begin{tikzpicture}[very thick,scale=0.6,color=green!50!black, baseline]
\draw[color=blue!50!black] (0,0) circle (1.5);
\draw[<-, color=blue!50!black, very thick] (0.100,-1.5) -- (-0.101,-1.5) node[above] {}; 
\draw[<-, color=blue!50!black, very thick] (-0.100,1.5) -- (0.101,1.5) node[below] {}; 
\fill (22.5:1.5) circle (3.3pt) node[left] {};
\draw (22.5:1.5) .. controls +(-0.3,0) and +(0,0.3) .. (0,0);
\fill (0,0) circle (3.3pt) node[left] (beta) {{\small $\alpha$}};
\fill[color=blue!50!black] (-22.5:1.5) circle (3.3pt) node[right] {{\small$\Psi$}};
\end{tikzpicture} 
%%%%%%%%%%%%%%%%%%%%%% 
\!\!
=
\Big\langle \beta_Q (\alpha) , \Psi \Big\rangle_Q 
\, , 
\end{align*}
and with this the relation~\eqref{eq:adjunctionBobu2} follows as 
\begin{align*}
\Big \langle \Psi,\; \beta_Q(\gamma_A \circ \alpha) \Big\rangle_Q 
& \eq^\eqref{eq:DefectParingSymmetry} 
\Big \langle \beta_Q(\gamma_A \circ \alpha),\; \Psi \Big\rangle_Q 
\eq^{\eqref{eq:adjunctionBobu1}}
\Big\langle \gamma_A \circ \alpha,\; \beta^Q(\Psi) \Big\rangle_{(a,A)} \\
& \eq^\eqref{eq:bulkPairingTwistedSymmetry} 
\Big\langle \beta^Q(\Psi),\; \alpha \Big\rangle_{(a,A)} \, .
\end{align*}
\end{proof}

The next result generalises \cite[Thm.\,6.5]{cr1210.6363} where the Cardy condition was proven for $\B = \LG$ and~$A$ a symmetric separable Frobenius algebra. The case $A = A_G$ with the symmetry condition  dropped was first shown for Landau-Ginzburg models in \cite[Thm.\,4.2.1]{pv1002.2116}. 

\begin{proposition}\label{prop:Cardy}
The Cardy condition holds for any separable Frobenius algebra~$A$: 
$$
\langle \beta^Q(\Phi) \, \beta^{P}(\Psi) \rangle_{(a,A)}
=
\tr({}_\Psi m_\Phi)
$$
where $\Phi \in \End_A(Q)$, $\Psi \in \End_A(P)$, and we take the $\C$-linear trace over $\Hom_A(Q,P)$. 
\end{proposition}

\begin{proof}
We calculate
\begin{align*}
&\Big\langle \beta^Q(\Phi) \, \beta^{P}(\Psi) \Big\rangle_{(a,A)} = 
\left\langle \,
%%%%%%%%%%%%%%%%%%%%%%
\begin{tikzpicture}[very thick,scale=0.6,color=blue!50!black, baseline]
\draw (0,0) circle (1.0);
\draw[->] (0.100,-1.0) -- (-0.101,-1.0) node[above] {}; 
\draw[->] (-0.100,1.0) -- (0.101,1.0) node[below] {}; 
\fill (202.5:1.0) circle (3.3pt) node[right] {{\small $\Phi$}};
\fill (90:1.0) circle (0pt) node[below] {{\small $Q$}};
\fill[color=green!50!black] (155:1.0) circle (3.3pt) node[right] {};
\draw[-dot-, color=green!50!black] (-2.5,-0.2) .. controls +(0,-1) and +(0,-1) .. (-1.5,-0.2);
\draw[color=green!50!black] (-1.5,-0.2) .. controls +(0,0.4) and +(-0.4,-0.2) .. (155:1.0);
\draw[color=green!50!black] (-2,-1.4) node[Odot] (unit) {};
\draw[color=green!50!black] (-2,-1) -- (unit);
%%%%
\begin{scope}[shift={(4.5,0)}]
\draw (0,0) circle (1.0);
\draw[->] (0.100,-1.0) -- (-0.101,-1.0) node[above] {}; 
\draw[->] (-0.100,1.0) -- (0.101,1.0) node[below] {}; 
\fill (202.5:1.0) circle (3.3pt) node[right] {{\small $\Psi$}};
\fill (90:1.0) circle (0pt) node[below] {{\small $P$}};
\fill[color=green!50!black] (155:1.0) circle (3.3pt) node[right] {};
\draw[-dot-, color=green!50!black] (-2.5,-0.2) .. controls +(0,-1) and +(0,-1) .. (-1.5,-0.2);
\draw[color=green!50!black] (-1.5,-0.2) .. controls +(0,0.4) and +(-0.4,-0.2) .. (155:1.0);
\draw[color=green!50!black] (-2,-1.4) node[Odot] (unit) {};
\draw[color=green!50!black] (-2,-1) -- (unit);
\end{scope}
%%%%
\draw[-dot-,color=green!50!black] (-2.5,-0.2) .. controls +(0,2.2) and +(0,2.2) .. (2,-0.2);
\draw[color=green!50!black] (-0.25,1.9) node[Odot] (unit2) {};
\draw[color=green!50!black] (-0.25,1.5) -- (unit2);
\end{tikzpicture} 
%%%%%%%%%%%%%%%%%%%%%% 
\right\rangle_{\raisemath{10pt}{\!\!\!a}}\\
&\eq^{\eqref{eq:FrobeniusProperty}}
\left\langle \,
%%%%%%%%%%%%%%%%%%%%%%
\begin{tikzpicture}[very thick,scale=0.6,color=blue!50!black, baseline]
\draw (0,0) circle (1.0);
\draw[->] (0.100,-1.0) -- (-0.101,-1.0) node[above] {}; 
\draw[->] (-0.100,1.0) -- (0.101,1.0) node[below] {}; 
\fill (202.5:1.0) circle (3.3pt) node[right] {{\small $\Phi$}};
\fill (90:1.0) circle (0pt) node[below] {{\small $Q$}};
\fill[color=green!50!black] (155:1.0) circle (3.3pt) node[right] {};
\draw[-dot-, color=green!50!black] (-2.5,-0.2) .. controls +(0,-1) and +(0,-1) .. (-1.5,-0.2);
\draw[color=green!50!black] (-1.5,-0.2) .. controls +(0,0.4) and +(-0.4,-0.2) .. (155:1.0);
\draw[color=green!50!black] (-2,-1.4) node[Odot] (unit) {};
\draw[color=green!50!black] (-2,-1) -- (unit);
%%%%
\begin{scope}[shift={(4,0)}]
\draw (0,0) circle (1.5);
\draw[->] (0.100,-1.5) -- (-0.101,-1.5) node[above] {}; 
\draw[->] (-0.100,1.5) -- (0.101,1.5) node[below] {}; 
\fill (202.5:1.5) circle (3.3pt) node[right] {{\small $\Psi$}};
\fill (270:1.5) circle (0pt) node[above] {{\small $P$}};
\fill[color=green!50!black] (110:1.5) circle (3.3pt) node[right] {};
\draw[color=green!50!black] (110:1.5) .. controls +(-6,-0.1) and +(0,1.2) .. (-6.5,-0.2);
\end{scope}
%%%%
\end{tikzpicture} 
%%%%%%%%%%%%%%%%%%%%%% 
\right\rangle_{\raisemath{10pt}{\!\!\!a}} \\
&\eq^{\eqref{eq:oriflip}}
%%%%%%%%%%%%%%%%%%%%%%
\begin{tikzpicture}[very thick,scale=0.6,color=blue!50!black, baseline]
\draw (0,0) circle (2.5);
\draw[<-] (0.100,-2.5) -- (-0.101,-2.5) node[above] {{\small $P$}}; 
\draw[<-] (-0.100,2.5) -- (0.101,2.5) node[below] {};
\draw (0,0) circle (1);
\draw[->] (0.100,-1) -- (-0.101,-1) node[below] {{}}; 
\draw[->] (-0.100,1) -- (0.101,1) node[below] {{\small $Q$}};
\fill (202.5:1) circle (3.3pt) node[right] {{\small $\Phi$}};
\fill (337.5:2.5) circle (3.3pt) node[right] {{\small $\Psi$}};
\fill[color=green!50!black] (155:1.0) circle (3.3pt) node[right] {};
\draw[-dot-, color=green!50!black] (-1.75,-0.2) .. controls +(0,-0.6) and +(0,-0.6) .. (-1.25,-0.2);
\draw[color=green!50!black] (-1.25,-0.2) .. controls +(0,0.3) and +(-0.3,-0.15) .. (155:1.0);
\draw[color=green!50!black] (-1.5,-1.1) node[Odot] (unit) {};
\draw[color=green!50!black] (-1.5,-0.7) -- (unit);
\draw[color=green!50!black] (-1.75,-0.2) .. controls +(0,1.5) and +(-1,-0.2) .. (70:2.5);
\fill[color=green!50!black] (70:2.5) circle (3.3pt) node {{}};
\end{tikzpicture}
%%%%%%%%%%%%%%%%%%%%%% 
\eq^{\text{Zorro}}_{\eqref{eq:rightPhi}}
%%%%%%%%%%%%%%%%%%%%%%
\begin{tikzpicture}[very thick,scale=0.6,color=blue!50!black, baseline]
\draw (180:2.5) arc (180:360:2.5);
\draw (-2.5,0) .. controls +(0,3) and +(0,0.75) .. (-0.83,2);
\draw (-0.83,2) .. controls +(0,-0.85) and +(0,-0.85) .. (0.83,2);
\draw (2.5,0) .. controls +(0,3) and +(0,0.75) .. (0.83,2);
\draw[<-] (0.100,-2.5) -- (-0.101,-2.5) node[above] {{\small $P$}}; 
\draw[<-] (-1.401,2.5) -- (-1.400,2.5) node[below] {};
\draw[<-] (1.280,2.5) -- (1.281,2.5) node[below] {};
\draw[<-] (-0.1,1.37) -- (0.1,1.37) node[below] {};
\draw (0,0) circle (1);
\draw[->] (0.100,-1) -- (-0.101,-1) node[below] {{}}; 
\draw[->] (-0.100,1) -- (0.101,1) node[below] {{\small $Q$}};
\fill (202.5:1) circle (3.3pt) node[right] {{\small $\Phi$}};
\fill (189:2.5) circle (3.3pt) node[left] {{\small $\Psi^\dagger$}};
\fill[color=green!50!black] (-1,0) circle (3.3pt) node (meet) {};
\draw[-dot-, color=green!50!black] (-1.5,-0.6) .. controls +(0,-0.5) and +(0,-0.5) .. (-2,-0.6);
\draw[color=green!50!black] (-1.5,-0.6) .. controls +(0,0.25) and +(-0.25,-0.25) .. (-1,0);
\fill[color=green!50!black] (-0.83,2) circle (3.3pt) node {};
\draw[color=green!50!black] (-2,-0.6) .. controls +(0,0.8) and +(-1.4,-0.4) .. (-0.83,2);
\draw[color=green!50!black] (-1.75,-1.4) node[Odot] (down) {}; 
\draw[color=green!50!black] (down) -- (-1.75,-1.05); 
\end{tikzpicture}
%%%%%%%%%%%%%%%%%%%%%% 
=
%%%%%%%%%%%%%%%%%%%%%%
\begin{tikzpicture}[very thick,scale=0.6,color=blue!50!black, baseline]
\draw (0,0) circle (2.5);
\draw[<-] (0.100,-2.5) -- (-0.101,-2.5) node[above] {{\small $P$}}; 
\draw[<-] (-0.100,2.5) -- (0.101,2.5) node[below] {};
\draw (0,0) circle (1);
\draw[->] (0.100,-1) -- (-0.101,-1) node[below] {{}}; 
\draw[->] (-0.100,1) -- (0.101,1) node[below] {{\small $Q$}};
\fill (202.5:1) circle (3.3pt) node[right] {{\small $\Phi$}};
\fill (189:2.5) circle (3.3pt) node[left] {{\small $\Psi^\dagger$}};
\fill[color=green!50!black] (-1,0) circle (3.3pt) node (meet) {};
\fill[color=green!50!black] (-2.5,0) circle (3.3pt) node (meet) {};
\draw[-dot-, color=green!50!black] (-1.5,-0.6) .. controls +(0,-0.5) and +(0,-0.5) .. (-2,-0.6);
\draw[color=green!50!black] (-1.5,-0.6) .. controls +(0,0.25) and +(-0.25,-0.25) .. (-1,0);
\draw[color=green!50!black] (-2,-0.6) .. controls +(0,0.25) and +(0.25,-0.25) .. (-2.5,0);
\draw[color=green!50!black] (-1.75,-1.4) node[Odot] (down) {}; 
\draw[color=green!50!black] (down) -- (-1.75,-1.05); 
\end{tikzpicture}
%%%%%%%%%%%%%%%%%%%%%%
\; ,
\end{align*}
where in the last step we use the definition (analogous to~\eqref{eq:pivvotal}) of the right $A$-module structure on $P^\dagger$. Adapting the argument at the end of the proof of \cite[Thm.\,6.5]{cr1210.6363} we continue to find
$$
%%%%%%%%%%%%%%%%%%%%%%
\begin{tikzpicture}[very thick,scale=0.6,color=blue!50!black, baseline]
\draw (0,0) circle (2.5);
\draw[<-] (0.100,-2.5) -- (-0.101,-2.5) node[above] {{\small $P$}}; 
\draw[<-] (-0.100,2.5) -- (0.101,2.5) node[below] {};
\draw (0,0) circle (1);
\draw[->] (0.100,-1) -- (-0.101,-1) node[below] {{}}; 
\draw[->] (-0.100,1) -- (0.101,1) node[below] {{\small $Q$}};
\fill (202.5:1) circle (3.3pt) node[right] {{\small $\Phi$}};
\fill (189:2.5) circle (3.3pt) node[left] {{\small $\Psi^\dagger$}};
\fill[color=green!50!black] (-1,0) circle (3.3pt) node (meet) {};
\fill[color=green!50!black] (-2.5,0) circle (3.3pt) node (meet) {};
\draw[-dot-, color=green!50!black] (-1.5,-0.6) .. controls +(0,-0.5) and +(0,-0.5) .. (-2,-0.6);
\draw[color=green!50!black] (-1.5,-0.6) .. controls +(0,0.25) and +(-0.25,-0.25) .. (-1,0);
\draw[color=green!50!black] (-2,-0.6) .. controls +(0,0.25) and +(0.25,-0.25) .. (-2.5,0);
\draw[color=green!50!black] (-1.75,-1.4) node[Odot] (down) {}; 
\draw[color=green!50!black] (down) -- (-1.75,-1.05); 
\end{tikzpicture}
%%%%%%%%%%%%%%%%%%%%%%
=
%%%%%%%%%%%%%%%%%%%%%%
\begin{tikzpicture}[very thick,scale=0.6,color=blue!50!black, baseline]
\draw (0,0) circle (1.5);
\draw[->] (0.100,-1.5) -- (-0.101,-1.5) node[below] {{\small $P^\dagger \otimes_A Q$}}; 
\draw[->] (-0.100,1.5) -- (0.101,1.5) node[below] {};
\fill (180:1.5) circle (3.3pt) node[right] {{\small $\Psi^\dagger \otimes_A \Phi$}};
\end{tikzpicture}
%%%%%%%%%%%%%%%%%%%%%%
=
\tr({}_\Psi m_\Phi)
\, . 
$$
\end{proof}

\subsubsection*{Summary}

Let us now summarise the foregoing discussion and compare the structure we obtained for generalised orbifolds $(a,A)$ with the structure of an open/closed TFT. For non-symmetric~$A$, neither of the bulk spaces $\Hcc$ and $\Hrr$ are guaranteed to form a commutative Frobenius algebra. This structure is rather ``shared'' between the two in the sense that $\Hcc$ is a commutative algebra, while $\Hrr$ has a nondegenerate pairing. This is in line with the general interpretation of the chiral sector as ``operators'', and the Ramond sector as ``states''. 

The situation is similar on the boundary, described by $A$-modules and the open string spaces $\Hom_A(Q,P)$ and $\Hom_A(Q,{}_{\gamma_A} P)$ in the chiral and Ramond sector, respectively. The category $\modu(A)$ does have a Serre functor which provides a perfect pairing of $\Hom_A(P,Q)$ with $\Hom_A(Q,{}_{\gamma_A} P)$, but in general there are no nondegenerate pairings for the individual sectors. 

Such mixing is also present for the boundary-bulk maps $\beta^Q: \End_A(Q) \rightarrow \Hrr$. On the other hand, we saw that the bulk-boundary maps~$\beta_Q$ respect the segregation of chiral and Ramond sectors, and that $\beta_Q, \beta^Q$ are ``twisted adjoint''. Most importantly, we proved a version of the Cardy condition -- involving the pairing on the Ramond bulk space and a trace in the chiral boundary sector. 

We thus conclude that as expected in general $(a,A)$ fails to provide an open/closed TFT, though the failure is relatively mild. In particular, if~$A$ is not only separable Frobenius but also symmetric, then $(a,A)$ does give an honest TFT.

\subsection{Defect actions}\label{subsec:DefectsFunctoriality}

The natural conclusion of this section is a discussion of defects. Given two orbifold theories $(a,A)$ and $(b,B)$, a defect between them is a 1-morphism $X\in \B(a,b)$ that has the structure of a $B$-$A$-bimodule: 
$$
%%%%%%%%%%%%%%%%%%%%%%
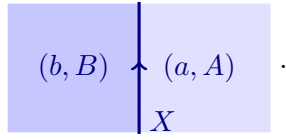
\begin{tikzpicture}[very thick,scale=0.7,color=blue!50!black, baseline=0cm]
\nicepalecolourscheme (0,-1.25) rectangle (2.5,1.25);
\nicehalfpalecolourscheme (-2.5,-1.25) rectangle (0,1.25);
\draw[line width=0] 
(0,1.25) node[line width=0pt] (A) {}
(0,-1.25) node[line width=0pt] (A2) {}; 
\draw[
	decoration={markings, mark=at position 0.55 with {\arrow{>}}}, postaction={decorate}
	]
 (0,-1.25) -- (0,1.25); 
 \draw[line width=0] 
(0,-1) node[line width=0pt, right] (Xbottom) {{\small $X$}}
(-2.1,0) node[line width=0pt, right] (Xbottom) {{\small $(b,B)$}}
(0.25,0) node[line width=0pt, right] (Xbottom) {{\small $(a,A)$}};
\end{tikzpicture}
%%%%%%%%%%%%%%%%%%%%%%
\, .
$$
If~$X'$ is another such defect we define the spaces of defect changing fields in the chiral primary and Ramond sectors by
$$
\Hom_{BA}(X,X') 
\, , \quad 
\Hom_{BA}(X, {}_{\gamma_B} X') 
\, , 
$$
respectively. This choice of defect condition changing fields in the Ramond sector reduces to the space~\eqref{eq:HARR} of RR ground states $\Hrr \cong \Hom_{AA}(A,{}_{\gamma_A} A)$ in the special case $X=A=B$. 
The unit $I_{(a,A)}$ of the theory $(a,A)$ is given by~$A$. 

Defects induce actions on bulk fields also in the orbifold theory. Indeed, the natural generalisation of~\eqref{eq:MotherDefectActionOnBulkFields} is
\begin{align}
\mathcal D_{\text{l}}(X) : \Hom(I_b,B) \ni \beta & \lmt 
%%%%%%%%%%%%%%%%%%%%%%
\begin{tikzpicture}[very thick,scale=0.75,color=blue!50!black, baseline]
\nicepalecolourscheme (0,0) circle (2.0);
\fill (-1.175,-1.175) circle (0pt) node[white] {{\small$a$}};
\nicecolourscheme (0,0) circle (1.25);
\fill (-0.65,-0.65) circle (0pt) node[white] {{\small$b$}};
\draw (0,0) circle (1.25);
\fill (-45:1.3) circle (0pt) node[right] {{\small$X$}};
\draw[<-, very thick] (0.100,-1.25) -- (-0.101,-1.25) node[above] {}; 
\draw[<-, very thick] (-0.100,1.25) -- (0.101,1.25) node[below] {}; 
\fill[color=green!50!black] (135:0) circle (2.5pt) node[left] {{\small$\beta$}};
\draw[color=green!50!black] (0,0) .. controls +(0,0.6) and +(-0.4,-0.4) .. (45:1.25);
\fill[color=green!50!black] (45:1.25) circle (2.5pt) node[right] {};
\fill[color=green!50!black] (25:1.25) circle (2.5pt) node[right] {};
\draw[color=green!50!black] (25:1.25) .. controls +(0.3,0.4) and +(0,-0.5) .. (40:2.3);
\fill[color=green!50!black] (40:2.3) circle (0pt) node[above] {{\small$A$}};
\end{tikzpicture} 
%%%%%%%%%%%%%%%%%%%%%% 
\in \Hom(I_a,A) 
\, , \label{eq:leftdefectorbiaction}
\\
\mathcal D_{\text{r}}(X) : \Hom(I_a, A) \ni \alpha & \lmt
%%%%%%%%%%%%%%%%%%%%%%
\begin{tikzpicture}[very thick,scale=0.75,color=blue!50!black, baseline]
\nicecolourscheme (0,0) circle (2.0);
\fill (1.175,-1.175) circle (0pt) node[white] {{\small$b$}};
\nicepalecolourscheme (0,0) circle (1.25);
\fill (0.65,-0.65) circle (0pt) node[white] {{\small$a$}};
\draw (0,0) circle (1.25);
\fill (-135:1.3) circle (0pt) node[left] {{\small$X$}};
\draw[->, very thick] (0.100,-1.25) -- (-0.101,-1.25) node[above] {}; 
\draw[->, very thick] (-0.100,1.25) -- (0.101,1.25) node[below] {}; 
\fill[color=green!50!black] (135:0) circle (2.5pt) node[right] {{\small$\alpha$}};
\draw[color=green!50!black] (0,0) .. controls +(0,0.6) and +(0.4,-0.4) .. (135:1.25);
\fill[color=green!50!black] (135:1.25) circle (2.5pt) node[right] {};
\fill[color=green!50!black] (155:1.25) circle (2.5pt) node[right] {};
\draw[color=green!50!black] (155:1.25) .. controls +(-0.3,0.4) and +(0,-0.5) .. (140:2.3);
\fill[color=green!50!black] (140:2.3) circle (0pt) node[above] {{\small$B$}};
\end{tikzpicture} 
%%%%%%%%%%%%%%%%%%%%%% 
\in \Hom(I_b,B) \, . \label{eq:rightdefectorbiaction}
\end{align}
As in~\eqref{eq:loopinvariant} one shows that these induce well-defined operators 
$$
\mathcal D_{\text{l}}(X): \mathcal H_{\text{RR}}^B \lra \mathcal H_{\text{RR}}^A
\, , \quad
\mathcal D_{\text{r}}(X): \mathcal H_{\text{RR}}^A \lra \mathcal H_{\text{RR}}^B
$$
by restriction to orbifold projected RR ground states. On the other hand, the defect actions do not in general map between orbifold projected (c,c) fields. 
More precisely, every defect acts as zero on (c,c) fields that do not happen to be simultaneous RR ground states, i.\,e.~elements of the unprojected orbifold space~\eqref{eq:HomIA} that are invariant under both $\picc$ and $\pirr$. This is a direct consequence of
$$
%%%%%%%%%%%%%%%%%%%%%%
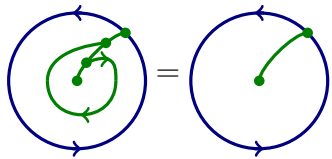
\begin{tikzpicture}[very thick,scale=0.6,color=green!50!black, baseline]
\draw[color=blue!50!black] (0,0) circle (1.5);
\draw[<-, color=blue!50!black, very thick] (0.100,-1.5) -- (-0.101,-1.5) node[above] {}; 
\draw[<-, color=blue!50!black, very thick] (-0.100,1.5) -- (0.101,1.5) node[below] {}; 
\fill (45:1.5) circle (3.3pt) node[left] {};
\draw (45:1.5) .. controls +(-0.3,0) and +(0,0.3) .. (0,0);
\fill (0,0) circle (3.3pt) node[left] (beta) {};
\fill (63:0.45) circle (3.3pt) node[left] {};
\fill (53:1.05) circle (3.3pt) node[left] {};
\draw[color=green!50!black](63:0.45) .. controls +(0.7,0.3) and +(0,0.2) .. (0:0.85);
\draw[->,color=green!50!black] (0.7,0.48) -- (0.71,0.48);
\draw[directedgreen](0:0.85) .. controls +(0,-1) and +(0,-1) .. (180:0.65);
\draw[color=green!50!black] (180:0.65) .. controls +(0,0.4) and +(-0.9,-0.1) .. (53:1.05);
\end{tikzpicture} 
%%%%%%%%%%%%%%%%%%%%%% 
= 
%%%%%%%%%%%%%%%%%%%%%%
\begin{tikzpicture}[very thick,scale=0.6,color=green!50!black, baseline]
\draw[color=blue!50!black] (0,0) circle (1.5);
\draw[<-, color=blue!50!black, very thick] (0.100,-1.5) -- (-0.101,-1.5) node[above] {}; 
\draw[<-, color=blue!50!black, very thick] (-0.100,1.5) -- (0.101,1.5) node[below] {}; 
\fill (45:1.5) circle (3.3pt) node[left] {};
\draw (45:1.5) .. controls +(-0.3,0) and +(0,0.3) .. (0,0);
\fill (0,0) circle (3.3pt) node[left] (beta) {};
\end{tikzpicture} 
%%%%%%%%%%%%%%%%%%%%%% 
$$
which in turn follows from the perfect pairing~\cite[Prop.\,4.6]{cr1210.6363} between (c,c) fields and RR ground states discussed in the previous section, to wit the special case where the (c,c) field is the identity. We thus conclude that in general defect actions are naturally defined only on Hochschild homology, i.\,e.~on RR ground states. 

A related aspect is that in the special case $X=A=B$ defect actions reduce to the RR projector~$\pirr$ of~\eqref{eq:piRR}: 
\be\label{eq:piARRdefectaction}
\pirr = \mathcal D_{\text{l}}(A) = \mathcal D_{\text{r}}(A) \, . 
\ee

\medskip

Defect actions on bulk fields continue to have all the expected properties. In particular, the invisible defect~$A$ acts as the identity on RR ground states, fusion is compatible with defect action, and wrapping a defect around a bulk field one way amounts to wrapping its adjoint the opposite way: 

\begin{proposition}
For $X: (a,A) \rightarrow (b,B)$ and $Y: (b,B) \rightarrow (c,C)$ we have 
\begin{enumerate}
\item $\mathcal D_{\text{l}}(A) = \mathcal D_{\text{r}}(A) = 1$ on $\Hrr$, 
\item 
$
\mathcal D_{\text{l}}(X) \circ \mathcal D_{\text{l}}(Y) = \mathcal D_{\text{l}}(Y\otimes_B X)
$, 
$
\mathcal D_{\text{r}}(Y) \circ \mathcal D_{\text{r}}(X) = \mathcal D_{\text{r}}(Y\otimes_B X)
$, 
\item 
$
\mathcal D_{\text{l}}(X) = \mathcal D_{\text{r}}(X^\dagger)
$.
\end{enumerate}
\end{proposition}

\begin{proof}
Part~(i) immediately follows from~\eqref{eq:piARRdefectaction}. The proof of part~(ii) proceeds completely analogously to the unorbifolded case (see~\cite[Prop.\,8.5]{cm1208.1481}) together with the argument for the appearance of the tensor product over~$B$ as in~\cite[Thm.\,6.5]{cr1210.6363}. Similarly, part~(iii) follows from pivotality and an application of the Zorro moves. 
\end{proof}

Finally we show that left and right defect actions are adjoint with respect to the bulk pairings~\eqref{eq:genorbbulkpairing} under the assumption~\eqref{eq:oriflip}. The interpretation is as in the unorbifolded case: wrapping a defect around one field insertion in a two-point correlator on the sphere is isotopic to wrapping the defect around the other insertion point with opposite orientation; see e.\,g.~\cite[Rem.\,8.4]{cm1208.1481} for an illustration. 
The precise statement is: 

\begin{proposition}
For $\alpha \in \Hrr$ and $\beta \in \mathcal{H}_{\text{RR}}^B$ we have 
\begin{align}
\Big\langle \beta, \mathcal D_{\text{r}}(X) (\alpha) \Big\rangle_{(b,B)}
& = 
\Big\langle \mathcal D_{\text{l}}(X) (\beta) , \alpha \Big\rangle_{(a,A)} 
\, , \label{eq:lassolemma1} \\
\Big\langle \alpha, \mathcal D_{\text{l}}(X) (\beta) \Big\rangle_{(a,A)}
& = 
\Big\langle \mathcal D_{\text{r}}(X) (\gamma^{-1}_A \circ \alpha) , \gamma^{-1}_B \circ \beta \Big\rangle_{(b,B)} \, . 
\label{eq:lassolemma2}
\end{align}
\end{proposition}

\begin{proof}
To prove~\eqref{eq:lassolemma1} we compute
\begin{align*}
\Big\langle \beta, \mathcal D_{\text{r}}(X) (\alpha) \Big\rangle_{(b,B)} 
& 
\eq^{\eqref{eq:genorbbulkpairing}}_{\eqref{eq:comoduleAction}}
\left\langle \, 
%%%%%%%%%%%%%%%%%%%%%%
\begin{tikzpicture}[very thick,scale=0.6,color=green!50!black, baseline]
\draw[color=blue!50!black] (0,0) circle (1.5);
\draw[->, color=blue!50!black, very thick] (0.100,-1.5) -- (-0.101,-1.5) node[above] {}; 
\draw[->, color=blue!50!black, very thick] (-0.100,1.5) -- (0.101,1.5) node[below] {}; 
\fill (0:0) circle (3.3pt) node[right] {{\small$\alpha$}};
\fill (135:1.5) circle (3.3pt) node[left] {};
\draw[color=green!50!black] (0,0) .. controls +(0,0.6) and +(0.4,-0.4) .. (135:1.5);
\fill (155:1.5) circle (3.3pt) node[left] {};
\draw[color=green!50!black] (-2,0) .. controls +(0,0.4) and +(-0.4,-0.2) .. (155:1.5);
\draw[-dot-] (-2,0) .. controls +(0,-1) and +(0,-1) .. (-3,0);
\draw[-dot-] (-3,0) .. controls +(0,1) and +(0,1) .. (-4,0);
\draw (-4,0) -- (-4,-1);
\fill (-4,-1) circle (3.3pt) node[right] {{\small$\beta$}};
\draw (-2.5,-1.2) node[Odot] (D) {}; 
\draw (-2.5,-0.8) -- (D);
\draw (-3.5,1.2) node[Odot] (D) {}; 
\draw (-3.5,0.8) -- (D);
\end{tikzpicture} 
%%%%%%%%%%%%%%%%%%%%%% 
\,  \right\rangle_{\raisemath{10pt}{\!\!\!b}}
\eq^{\eqref{eq:FrobeniusProperty}}
\left\langle 
%%%%%%%%%%%%%%%%%%%%%%
\begin{tikzpicture}[very thick,scale=0.6,color=green!50!black, baseline]
\draw[color=blue!50!black] (0,0) circle (1.5);
\draw[->, color=blue!50!black, very thick] (0.100,-1.5) -- (-0.101,-1.5) node[above] {}; 
\draw[->, color=blue!50!black, very thick] (-0.100,1.5) -- (0.101,1.5) node[below] {}; 
\fill (0:0) circle (3.3pt) node[right] {{\small$\alpha$}};
\fill (135:1.5) circle (3.3pt) node[left] {};
\draw[color=green!50!black] (0,0) .. controls +(0,0.6) and +(0.4,-0.4) .. (135:1.5);
\fill (155:1.5) circle (3.3pt) node[left] {};
\draw[color=green!50!black] (-2,0) .. controls +(0,0.4) and +(-0.4,-0.2) .. (155:1.5);
\fill (-2,0) circle (3.3pt) node[left] {{\small$\beta$}};
\end{tikzpicture} 
%%%%%%%%%%%%%%%%%%%%%% 
\,  \right\rangle_{\raisemath{10pt}{\!\!\!b}}
\\
& 
\eq^{\eqref{eq:oriflip}}
\left\langle \, 
%%%%%%%%%%%%%%%%%%%%%%
\begin{tikzpicture}[very thick,scale=0.6,color=green!50!black, baseline]
\draw[color=blue!50!black] (0,0) circle (1.5);
\draw[<-, color=blue!50!black, very thick] (0.100,-1.5) -- (-0.101,-1.5) node[above] {}; 
\draw[<-, color=blue!50!black, very thick] (-0.100,1.5) -- (0.101,1.5) node[below] {}; 
\fill (45:1.5) circle (3.3pt) node[left] {};
\fill (22.5:1.5) circle (3.3pt) node[left] {};
\draw (45:1.5) .. controls +(0.3,0) and +(0,0.3) .. (1.9,0.45);
\draw (22.5:1.5) .. controls +(-0.3,0) and +(0,0.3) .. (0,0);
\fill (0,0) circle (3.3pt) node[left] (beta) {{\small $\beta$}};
\fill (1.9,0.45) circle (3.3pt) node[right] (beta) {{\small $\alpha$}};
\end{tikzpicture} 
%%%%%%%%%%%%%%%%%%%%%% 
\!\!
\right\rangle_{\raisemath{10pt}{\!\!\!a}}
=
\Big\langle \mathcal D_{\text{l}}(X) (\beta) , \alpha \Big\rangle_{(a,A)} 
\, . 
\end{align*}
With this~\eqref{eq:lassolemma2} follows from 
\begin{align*}
&
\Big\langle \alpha, \mathcal D_{\text{l}}(X) (\beta) \Big\rangle_{(a,A)} 
\eq^{\eqref{eq:bulkPairingTwistedSymmetry}}
\Big\langle \gamma_A \circ \mathcal D_{\text{l}}(X) (\beta) , \alpha \Big\rangle_{(a,A)}
\eq^{\eqref{eq:AlgebraMorphism}}
\Big\langle \mathcal D_{\text{l}}(X) (\beta) , \gamma_A^{-1} \circ \alpha \Big\rangle_{(a,A)}
\\ 
&\qquad\qquad
\eq^{\eqref{eq:lassolemma1}}
\Big\langle  \beta , \mathcal D_{\text{r}}(X) (\gamma_A^{-1} \circ \alpha) \Big\rangle_{(b,B)}
\eq^{\eqref{eq:bulkPairingTwistedSymmetry}}
\Big\langle \gamma_B \circ \mathcal D_{\text{r}}(X) (\gamma^{-1}_A \circ \alpha) , \beta \Big\rangle_{(b,B)} 
\\ 
&\qquad\qquad
\eq^{\eqref{eq:AlgebraMorphism}}
\Big\langle \mathcal D_{\text{r}}(X) (\gamma^{-1}_A \circ \alpha) , \gamma^{-1}_B \circ \beta \Big\rangle_{(b,B)} 
\, . 
\end{align*}
\end{proof}

\begin{remark}\label{rem:twistedpivotal}
The basic assumption we made in this section is that the bicategory~$\B$ be pivotal. On the other hand, our motivating example of Landau-Ginzburg models gives rise to a bicategory~$\LG$ that is not quite pivotal, but only ``graded pivotal'' as discussed in~\cite[Sect.\,7]{cm1208.1481}. Similarly, the bicategory of B-twisted sigma models studied in~\cite{cw1007.2679} as well as the equivariant completion of~\cite{cr1210.6363} are ``pivotal up to a twist by the Serre functor''. We could have developed the theory in this more general setting; indeed, every result in this section can straightforwardly be strengthened to the twisted pivotal case by enriching string diagrams with the ``wiggly line calculus'' for Serre functors, explained in detail in the Landau-Ginzburg setting in~\cite[Sect.\,7--9]{cm1208.1481}. We chose not to present these merely technical, non-conceptual details in an attempt not to distract the reader from our construction's naturality and simplicity. 
\end{remark}

\subsubsection*{Acknowledgements} 

We thank Michael Kay, Dan Murfet, and Ingo Runkel for helpful discussions.

\appendix

\section{Some technical details}\label{TheAppendix}

\subsection{Equivalent descriptions of orbifold spaces}\label{app:HHPrime}

In this appendix we show that the bulk orbifold spaces $\Hil$ and $\Hil'$, defined in~\eqref{eq:orbiproj} and~\eqref{eq:orbiprojPrime}, respectively, are isomorphic to each other. We construct this isomorphism explicitly, which will allow us to map twisted states from one picture to another. Since obviously $\Hil=\Hil'$ for an abelian orbifold group $G$, the case of interest is when $G$ is nonabelian.

Let us fix a set of representatives $\{g\}$ for the conjugacy classes $\Cong$ in $G$. We will proceed by constructing an embedding $\alpha: H' \rightarrow H$ of the unprojected bulk space 
$$
H'=\bigoplus_{\{g\}} H_g
$$
into the space $H=\bigoplus_{g\in G}H_g$, such that $\alpha$ intertwines the orbifold projectors on~$H$ and~$H'$, i.\,e.~$\alpha P'=P \alpha$, and such that $\alpha(\Hil')=\Hil$. There exist several different embeddings satisfying these properties, but as we will show, they all agree on the invariant subspace $\Hil'$, giving rise to a canonical isomorphism $\Hil' \cong \Hil$.

To construct the embedding, for every element $g\in G$ we define the map
$$
\alphagKg: H_g \lra \bigoplus_{g' \in \Cong} H_{g'} 
\, , \quad 
|\phi\rangle^g \lmt \sum_{k\in K(g)} k |\phi\rangle^g \, ,
$$
where $K(g)$ denotes a set of representatives of $G/\Centg \cong \Cong$. Note that $G/\Centg$ can be equivalently described as
$$
G/\Centg=G/\!\!\sim \quad \text{with} \quad k_1 \sim k_2 \quad \text{iff} \quad k_1 g k_1^{-1}=k_2 g k_2^{-1}\, ,
$$
since $k_1 \sim k_2$ if and only if $k_1=k_2 l$ for some $l \in \Centg$. 
We now define the map
$$
\alphaKK: H' \lra H 
\, , \quad 
\alphaKK=\sum_{\{g\}}\alphagKg \, ,
$$
where for every element in $\{g\}$ we choose a set of representatives $K(g)$ of  $G/\Centg$, and we denote this choice collectively by $\{K\}$. Since $\alphaKK$ is injective we indeed obtain an embedding of $H'$ into $H$.

\begin{lemma}
For any choice of representatives $\{K\}$ the map $\alphaKK$ intertwines the projectors on $H$ and $H'$: 
$$
\alphaKK P'=P\alphaKK \, .
$$
\end{lemma}
\begin{proof}
It suffices to prove that $\alphagKg P'_g = P \alphagKg$ for every $g\in G$. We compute
\begin{align*}
P \alphagKg=\frac{1}{|G|}\sum_{\substack{h \in G,\\ k \in K(g)}} \!\! h k = \frac{1}{|G|} \sum_{\substack{h' \in G,\\ k \in K(g)}} \!\! h' = \frac{|\Cong|}{|G|}\sum_{h'\in G} h'= \frac{1}{|\Centg|} \sum_{\substack{k \in K(g),\\ l \in \Centg}} \!\! kl= \alphagKg P'_g \, ,
\end{align*}
where we used the substitution $h'= hk$, the fact that $|K(g)|=|\Cong|$, and that any $h\in G$ can be written as $h=kl$ for some $k\in K(g)$ and $l \in \Centg$.
\end{proof}
In particular, this implies $\alphaKK(\Hil')\subset \Hil$. While the embedding $\alphaKK$ depends on the choice of representatives $\{K\}$, its restriction to $\Hil'$ is independent of this choice, as the next lemma shows.

\begin{lemma}
For any two choices of representatives $\{K\}$ and $\{K'\}$, one has
$$
\alphaKK\big|_{\Hil'} =\alpha^{\{K'\}}\big|_{\Hil'} \, .
$$
\end{lemma}
\begin{proof}
Given an element $g\in G$, a vector $|\phi\rangle^g \in P'_g H'_g$, and two representatives $k\in K(g)$ and $k'\in K'(g)$ of a class in $G/\Centg$, we find
$$
k|\phi\rangle^g= k' |\phi\rangle^g \, ,
$$
since $k'=k l$ for some $l\in \Centg$, and $l |\phi\rangle^g= |\phi\rangle^g$. It then directly follows that
$$
\alphaKK |\phi\rangle = \sum_{\{g\}}\sum_{k\in K(g)} k |\phi\rangle= \sum_{\{g\}}\sum_{k'\in K'(g)} k' |\phi\rangle = \alpha^{\{K'\}}  |\phi\rangle \,
$$
for any $|\phi\rangle \in \Hil'$.
\end{proof}

This result allows us to define the map
$$
\alpha_0: \Hil' \lra \Hil 
\, , \quad \alpha_0=\alphaKK\big|_{\Hil'} \, ,
$$
which provides the desired isomorphism: 
\begin{proposition}
The map $\alpha_0$ is an isomorphism between $\Hil'$ and $\Hil$.
\end{proposition}
\begin{proof}
We already know that $\alpha_0$ is an embedding of $\Hil'$ into $\Hil$, so we only need to show that its image is the whole of $\Hil$. Let us take an arbitrary vector $|\phi\rangle \in H$. We wish to show that $P|\phi\rangle \in \im(\alpha_0)$. For this, we choose a set of representatives $K(g)$ of $G/\Centg$ for every element in $\{g\}$, and decompose $|\phi\rangle$ into homogeneous components as
$$
|\phi\rangle=\sum_{\{g\}}\sum_{k\in K(g)} |\phi \rangle^{kgk^{-1}} \, .
$$
Further, each $|\phi \rangle^{kgk^{-1}}\in H_{kgk^{-1}}$ can be expressed as $|\phi \rangle^{kgk^{-1}}=k |\phi_k' \rangle^g$ for some $|\phi_k' \rangle^g\in H_g$, so we can write
$$
|\phi\rangle=\sum_{\{g\}}\sum_{k\in K(g)} k |\phi_k' \rangle^g \, .
$$
Next we compute
\begin{align*}
P |\phi\rangle &= \frac{1}{|G|} \sum_{\{g\}}\sum_{\substack{h\in G, \\ k\in K(g)}} h k |\phi_k' \rangle^g = \frac{1}{|G|} \sum_{\{g\}}\sum_{\substack{h\in G, \\ k\in K(g)}} h |\phi_k' \rangle^g=  \frac{1}{|G|} \sum_{\{g\}} \sum_{\substack{k, k' \in K(g), \\ l \in \Centg}} k' l |\phi_k' \rangle^g\\
&=  \sum_{\{g\}} \sum_{k\in K(g)} k \, \frac{1}{|\Centg|} \sum_{l \in \Centg} l |\psi\rangle^g = \sum_{\{g\}} \sum_{k\in K(g)} k P'_g |\psi\rangle^g= \alpha_0 P' \Big(\sum_{\{g\}}|\psi\rangle^g\Big)
\end{align*}
where 
$$
|\psi\rangle^g=\frac{|\Centg|}{|G|}\sum_{k\in K(g)} |\phi'_k\rangle^g \, .
$$
Hence, $P|\phi\rangle\in \im(\alpha_0)$ as we wanted to show.
\end{proof}

\subsection{Orbifold projection via defects}\label{app:RRprojDetails}

Our orbifold projectors on (c,c) fields and RR ground states are given by~\eqref{eq:piccAG} and~\eqref{eq:piRRAG}, respectively. In this appendix we determine their action explicitly in the case of A-type minimal models (and hence also their tensor products). This reproduces the orbifold projectors in the conventional description reviewed in Section~\ref{subsec:conventionalorbi}.

Consider the minimal model with potential $W=x^d$ and orbifold group $G=\Z_d$, which acts as $x \mapsto \eta^g x$ with $\eta=\E^{2\pi \I/d}$ for $g\in \Z_d$. First, we focus on the projector on (c,c) fields. In the defect approach, the action of a group element $h\in G$ on a bulk field in the $g$-twisted sector $\phi_g \in \Hom(I, {}_g I)$ is implemented by
$$
\phi_g \lmt {}_{\raisemath{-1pt}{h}} (\lambda_{{}_{g h^{-1}} I}) \circ \big(1_{{}_h I} \otimes \big( \; {}_{\raisemath{-1pt}{g}} (\lambda_{{}_{h^{-1}} I}) \circ (\phi_g \otimes 1_{{}_{h^{-1}} I}) \circ \lambda^{-1}_{{}_{h^{-1}} I} \big) \big) \circ {}_{\raisemath{-2pt}{h}} \big(\lambda^{-1}_{{}_{h^{-1}} I}\big) \, .
$$
This action can be determined explicitly using the expression for the inverse left unit action $\lambda^{-1}$ given in~\eqref{eq:lambdainverse}. The spaces $\Hom(I, {}_g I)$ for $g \neq 0$ are in the present case one-dimensional~\cite{bg0503207}, spanned by 
$$
\phi_g=\theta^* - \frac{x^d-x'^d}{(x-x')(\eta^g x-x')} \cdot \theta \, ,
$$
while $\End(I)$ is spanned by monomials $x^l$ with $l < d-1$. We then compute
\begin{align*}
{}_{\raisemath{-1pt}{g}} (\lambda_{{}_{h^{-1}} I}) \circ (\phi_g \otimes 1_{{}_{h^{-1}} I}) \circ \lambda^{-1}_{{}_{h^{-1}} I} \; (e_i) &= \sum_j \left\{ \partial^{x,x'}_{[1]} d_{{}_{h^{-1}} I} (x,z') \right\}_{ji} \Big|_{x' \mapsto \eta^g x} \otimes e_j \,, \\
(\lambda_{{}_{h^{-1}} I}) \circ (x^l \otimes 1_{{}_{h^{-1}} I}) \circ \lambda^{-1}_{{}_{h^{-1}} I} \; (e_i) &= x^l e_i \,,
\end{align*}
where $x,x'$ are the variables inside the region enclosed by $A_G$ in \eqref{eq:piccAG}, while $z,z'$ denote variables outside of this region. The action of $h$ on the bulk fields is thus given by
\begin{align*}
\phi_g &\lmt \Big( \big(\partial^{x,x'}_{[1]} d_{{}_{h^{-1}} I} (x,z') \big) \big|_{x' \mapsto \eta^g x} \Big) \Big|_{x \mapsto \eta^h z}= \eta^{-h} \phi_g = \det(h)^{-1} \, \phi_g \,,\\
x^l &\lmt \eta^{hl} z^l \,.
\end{align*}

In Section~\ref{subsec:gentwist} we show that the orbifold projector on RR ground states is related to the projector on (c,c) fields by the Nakayama automorphism $\gamma_{A_G}= \sum_{h\in G} \det(h) \cdot 1_{{}_h I}$. It immediately follows that the action of an element $h\in G$ on an RR ground state is now given by
\begin{align*}
\phi_g &\lmt \det(h) \eta^{-h} \phi_g = \phi_g \,,\\
x^l &\lmt \det(h) \eta^{hl} z^l \,.
\end{align*}

\subsection{Disc one-point correlators}\label{app:RRchargeComputations}

In Section~\ref{subsec:defectorbi} we explained how disc correlators in Landau-Ginzburg orbifolds can be computed using the defect approach. Here we want to verify that for a single bulk insertion and no boundary insertions, \eqref{eq:orbiDisc} reproduces the RR charge formula~\eqref{eq:WalcherProposal} of \cite{w0412274}. For this, we will study the examples of minimal model branes, generalised permutation branes, and tensor products thereof, where we find that the two approaches indeed yield the same result (up to a suitable normalisation owed to our conventions). Other cases where we checked (using computer algebra) the agreement with~\eqref{eq:WalcherProposal} include all branes in D- and E-type minimal models and several examples of linear matrix factorisations~\cite{err0508}.

In terms of its constituent morphisms the diagram~\eqref{eq:orbiDisc} with $\alpha=\phi_g$ and $\Psi=1$ is given by
\be\label{eq:orbiDiscMorph}
\langle \phi_g \rangle = \text{ev}_Q \circ \big( 1_{Q^\dagger} \otimes \left(\gamma_g \circ {}_{g}(\lambda_Q) \circ (\phi_g \otimes 1_Q) \circ \lambda_Q^{-1}\right) \big)\circ  \tcoev_Q \,.
\ee
This can be evaluated explicitly using the formulas for $\lambda^{-1}$, $\tcoev$, and $\ev$, given in \eqref{eq:lambdainverse}, \eqref{eq:coevtilde}, and \eqref{eq:eval}, respectively. Since $Q$ is a boundary condition, the expressions for adjunction maps simplify to
\begin{align*}
\widetilde{\coev}_Q(1)&=\sum_i (-1)^{|e_i|} e_i^* \otimes e_i \, ,\\
\ev_X( e_i^* \otimes p e_j )&=  \Res \left[ \frac{ p \, \big\{ \partial_{x_1}d_Q\ldots \partial_{x_n}d_Q \big\}_{ij}  \, \operatorname{d}\!x }{\partial_{x_1}W, \ldots, \partial_{x_n} W} \right] .
\end{align*}
We now restrict ourselves to the case when the potential $W$ is of Fermat type and the orbifold group $G$ acts diagonally on the variables $x_i$. The twisted fields $\phi_g \in \Hom(I, {}_g I)$ then take the simple form (see \cite{br0707.0922})
$$
\phi_g=(-1)^{\binom{t}{2}} \, \varphi_g \, \prod_{i=r+1}^n \left(\theta^*_i-\frac{\partial_{[i]}^{x,x'} W}{g(x_i)-x'_i} \cdot \theta_i\right) ,
$$
where for convenience we introduced the factor $(-1)^{\binom{t}{2}}$ with $t=n-r$, and we reorder the variables~$x_i$ so that the first~$r$ are $g$-invariant. The term $\varphi_g$ is a polynomial in the ring $\C[\overline x_1, \ldots, \overline x_r]/(\partial_{\overline x_i} \overline W)$ with notation as in Section~\ref{subsec:defectorbi}. Using this form of $\phi_g$ we compute
$$
(\phi_g \otimes 1_Q) \circ \lambda_Q^{-1} (e_i)\big|_{\theta=0} = \sum_j \varphi_g \left\{ \partial^{x,x'}_{[t]}d_Q \ldots \partial^{x,x'}_{[1]}d_Q \right\}_{ji}  \otimes e_j \,,
$$
so we have
$$
{}_{g}(\lambda_Q) \circ (\phi_g \otimes 1_Q) \circ \lambda_Q^{-1} (e_i) = \sum_j \varphi_g \left\{ \partial^{x,x'}_{[t]}d_Q \ldots \partial^{x,x'}_{[1]}d_Q \right\}_{ji}\Big|_{x' \mapsto g(x)} \, e_j\,.
$$
We thus obtain the following expression for \eqref{eq:orbiDiscMorph} in the Fermat case:
\be\label{eq:orbiDiscFermat}
\langle \phi_g \rangle = \Res \left[ \frac{ \str \Big(\partial_{x_1}d_Q \ldots \partial_{x_n}d_Q \, \gamma_g \, \varphi_g \left( \partial^{x,x'}_{[t]}d_Q \ldots \partial^{x,x'}_{[1]}d_Q \right)\Big|_{x' \mapsto g(x)} \Big) \operatorname{d}\!x}{\partial_{x_1}W, \ldots, \partial_{x_n} W} \right] .
\ee
We will now show that this formula reproduces~\eqref{eq:WalcherProposal} in the case of minimal model branes and generalised permutation branes. Note that because of the form of \eqref{eq:WalcherProposal} and \eqref{eq:orbiDiscFermat} the results can be directly extended to tensor products and cones of these branes.

\subsubsection*{Minimal model branes}

Consider the example of an A-type minimal model given by a potential $W=x^d$ and orbifold group $G=\Z_d$, and a matrix factorisation $Q$ with differential
$$
d_Q=\begin{pmatrix}0 & x^m \\ x^{d-m} & 0\end{pmatrix}, \quad m \in \{1,\ldots, d-1\} \, .
$$
The generator of $\Z_d$ acts on the variable $x$ by $x \mapsto \eta x$ with $\eta=\E^{2\pi \I/d}$, and its action on $Q$ is represented by the matrix
$$
\gamma=\eta^p \begin{pmatrix}1 & 0\\ 0 & \eta^{m} \end{pmatrix} 
$$
with $p \in \{0,\ldots,d-1\}$ denoting the representation label of~$Q$. Let us now evaluate~\eqref{eq:orbiDiscFermat} in a twisted sector (i.\,e.~for $g\neq 0$ and $\varphi_g=1$). Using that
$$
\partial_x d_Q \, \gamma^g \, (\partial^{x,x'}_{[1]} d_Q) \big|_{x' \mapsto \eta x} =\eta^{g p}  \begin{pmatrix}m x^{d-2} \,\frac{1-\eta^{g(d-m)}}{1-\eta^g} \,\eta^{g m} & 0\\ 0 & (d-m) x^{d-2} \, \frac{1-\eta^{g m}}{1-\eta^g}\, \end{pmatrix},
$$
we obtain
\begin{align*}
\langle \phi_g \rangle &= \mathrm{Res}\left[\frac{\str\big( \partial_x d_Q \, \gamma^g \,  (\partial^{x,x'}_{[1]} d_Q) \big|_{x' \mapsto \eta^g x} \big) \operatorname{d}\!x }{\partial_x W}\right] \\
&=\Res \left[\frac{\eta^{g p} x^{d-2}\frac{m (\eta^{g m}-1)-(d-m)(1-\eta^{g m})}{1-\eta^g} \operatorname{d}\!x }{d x^{d-1}}\right] 
=-\eta^{g p}\frac{1-\eta^{g m}}{1-\eta^{g}}=-\frac{\str\gamma^g}{1-\eta^g} \, .
\end{align*}
This reproduces the result of~\cite{w0412274} up to normalisation of the bulk field, which in our case is
$$
|\phi_g|^2=\big\langle \phi_g , \phi_{g^{-1}} \big\rangle_{(W, A_G)}= \left\langle \frac{d x^{d-2}}{1-\eta^g}\right\rangle_W= \frac{1}{1-\eta^g}\Res \left[ \frac{d x^{d-2} \operatorname{d}\!x }{\partial_x W} \right]= \frac{1}{1-\eta^g} \, ,
$$
where $\langle -,- \rangle_{(W, A_G)}$ is the bulk pairing~\eqref{eq:genorbbulkpairing}. A normalisation factor of this type also appears in all other examples that we have studied.

\subsubsection*{Generalised permutation branes}

The other example we consider is $W=x_1^{d k_1}+x_2^{d k_2}$, with $k_1, k_2$ coprime, and the diagonal action of the orbifold group $G=\Z_H$ with $H=k_1 k_2 d$. We consider linear permutation branes~$Q$ with differential
$$
d_Q=\begin{pmatrix}0 & x_1^{k_1}-\mu x_2^{k_2} \\ \displaystyle \prod_{\mu'\neq \mu, \mu'^d=-1}(x_1^{k_1}-\mu' x_2^{k_2}) & 0\end{pmatrix}
$$
where $\mu^d=-1$. The generator of~$\Z_H$ is represented on~$Q$ by the matrix
$$
\gamma=\omega^p \begin{pmatrix}1 & 0\\ 0 & \omega^{k_1 k_2} \end{pmatrix} 
\quad\text{with}\quad 
\omega=\E^{2 \pi \I/H} \, , \quad
p \in \{ 0, \ldots, H-1 \} \, . 
$$
In a sector with one twisted variable ($t=1$), the correlator $\langle \phi_g \rangle$ vanishes identically, because the expression in the supertrace is odd. For two twisted variables ($t=2$) after a straightforward computation one obtains 
$$
\langle \phi_g \rangle=-\frac{\str \gamma^g}{(1-\eta_1^g)(1-\eta_2^g)}
$$
where $\eta_l=\E^{2\pi \I/(k_l d)}$ for $l\in\{1,2\}$.


\begin{thebibliography}{HKK+}

\bibitem[Abo]{a1001.4593}
M.~Abouzaid, \textsl{A geometric criterion for generating the Fukaya category}, 
\doi{10.1007/s10240-010-0028-5}{Publications math\'{e}matiques de l'IH\'{E}S \textbf{112} (2010), 191--240},
\href{http://arxiv.org/abs/1001.4593}{[arXiv:1001.4593]}. 

\bibitem[ADD]{add0401}
S.~K.~Ashok, E.~Dell'Aquila, and D.-E.~Diaconescu, \textsl{Fractional {B}ranes in
  {L}andau-{G}inzburg {O}rbifolds}, \bhttpurl{projecteuclid.org/DPubS?service=UI&version=1.0&verb=Display&handle=euclid.atmp/1098389089}{Adv. Theor. Math. Phys. \textbf{8} (2004),
  461--}{513}, \href{http://www.arxiv.org/abs/hep-th/0401135}{[hep-th/0401135]}.

\bibitem[BD]{bdCobordismHypothesis}
J.~Baez and J.~Dolan, \textsl{Higher-Dimensional Algebra and Topological Quantum Field Theory}, \doi{10.1063/1.531236}{J.~Math.~Phys.~\textbf{36} (1995), 6073--6105},
  \href{http://www.arxiv.org/abs/math.QA/9503002}{[math.QA/9503002]}. 

\bibitem[BFK]{bfk1105.3177}
M.~Ballard, D.~Favero, and L.~Katzarkov, \textsl{A category of kernels for graded matrix factorisations and its implications for {H}odge theory}, \href{http://arxiv.org/abs/1105.3177}{[arXiv:1105.3177]}.

\bibitem[BBP]{Baumgartl:2012uh}
  M.~Baumgartl, I.~Brunner, and D.~Plencner,
  \textsl{D-brane Moduli Spaces and Superpotentials in a Two-Parameter Model},
  \bdoi{10.1007/JHEP03(2012)039}{JHEP \textbf{1203} (2012),}{039}, 
\href{http://arxiv.org/abs/1201.4103}{[arXiv:1201.4103]}.

\bibitem[Bor]{bor94}
F.~Borceux, \textsl{Handbook of categorical algebra 1}, volume 50 of \textsl{Encyclopedia of Mathematics and its Applications}, Cambridge University Press, Cambridge, 1994.

\bibitem[BCP]{BCPdiscretetorsion}
I.~Brunner, N.~Carqueville, and D.~Plencner, \textsl{Discrete torsion defects}, 
\href{http://arxiv.org/abs/1404.7497}{[arXiv:1404.7497]}.

\bibitem[BG]{bg0503207}
I.~Brunner and M.~R.~Gaberdiel, \textsl{Matrix factorisations and permutation branes}, \doi{10.1088/1126-6708/2005/07/012}{JHEP \textbf{0507} (2005), 012},
  \href{http://www.arxiv.org/abs/hep-th/0503207}{[hep-th/0503207]}. 

\bibitem[BHLS]{bhls0305}
I.~Brunner, M.~Herbst, W.~Lerche, and B.~Scheuner, \textsl{Landau-{G}inzburg
  {R}ealization of {O}pen {S}tring {TFT}}, \doi{10.1088/1126-6708/2006/11/043}{JHEP \textbf{0611} (2003), 043},
  \href{http://www.arxiv.org/abs/hep-th/0305133}{[hep-th/0305133]}.

\bibitem[BR1]{br0707.0922}
I.~Brunner and D.~Roggenkamp, \textsl{B-type defects in {L}andau-{G}inzburg
  models}, \doi{10.1088/1126-6708/2007/08/093}{JHEP \textbf{0708} (2007), 093},
  \href{http://arxiv.org/abs/0707.0922}{[arXiv:0707.0922]}.

\bibitem[BR2]{br0712.0188}
I.~Brunner and D.~Roggenkamp, \textsl{Defects and Bulk Perturbations of Boundary Landau-Ginzburg Orbifolds}, \doi{10.1088/1126-6708/2008/04/001}{JHEP \textbf{0804} (2008), 001},
  \href{http://arxiv.org/abs/0712.0188}{[arXiv:0712.0188]}.

\bibitem[CW]{cw1007.2679}
A.~{C\u ald\u araru} and S.~Willerton, \textsl{The Mukai pairing, I: a categorical approach},
\href{http://nyjm.albany.edu/j/2010/16-6.html}{New York Journal of Mathematics \textbf{16} (2010), 61--98}, 
  \href{http://arxiv.org/abs/0707.2052}{[arXiv:0707.2052]}.

\bibitem[CM1]{khovhompaper}
N.~Carqueville and D.~Murfet, \textsl{Computing {K}hovanov-{R}ozansky homology and defect fusion}, \href{http://arxiv.org/abs/1108.1081}{[arXiv:1108.1081]}. 

\bibitem[CM2]{cm1208.1481}
N.~Carqueville and D.~Murfet, \textsl{Adjunctions and defects in Landau-Ginzburg models}, \href{http://arxiv.org/abs/1208.1481}{[arXiv:1208.1481]}. 

\bibitem[CM3]{cm1303.1389}
N.~Carqueville and D.~Murfet, \textsl{A toolkit for defect computations in Landau-Ginzburg models}, \href{http://arxiv.org/abs/1303.1389}{[arXiv:1303.1389]}. 

\bibitem[CR1]{cr0909.4381}
N.~Carqueville and I.~Runkel,
\textsl{On the monoidal structure of matrix bi-factorisations}, \doi{10.1088/1751-8113/43/27/275401}{J. Phys. A: Math. Theor. \textbf{43} (2010), 275401},
\arxiv{0909.4381}{[arXiv:0909.4381]}.

\bibitem[CR2]{cr1006.5609}
N.~Carqueville and I.~Runkel, \textsl{Rigidity and defect actions in
  Landau-Ginzburg models}, \doi{10.1007/s00220-011-1403-x}{Comm. Math. Phys. \textbf{310} (2012), 135--179}, 
  \href{http://arxiv.org/abs/1006.5609}{[arXiv:1006.5609]}.
  
\bibitem[CR3]{cr1210.6363}
N.~Carqueville and I.~Runkel, \textsl{Orbifold completion of defect bicategories}, 
\href{http://arxiv.org/abs/1210.6363}{[arXiv:1210.6363]}.

\bibitem[CFG]{cfg0511078}
C.~Caviezel, S.~Fredenhagen, and M.~R.~Gaberdiel, \textsl{The RR charges of A-type Gepner models}, \doi{10.1088/1126-6708/2006/01/111}{JHEP \textbf{0601} (2006), 111},
  \href{http://www.arxiv.org/abs/hep-th/0511078}{[hep-th/0511078]}. 

\bibitem[DM]{dm1102.2957}
T.~Dyckerhoff and D.~Murfet, \textsl{Pushing forward matrix factorisations},
  \href{http://arxiv.org/abs/1102.2957}{[arXiv:1102.2957]}.

\bibitem[ERR]{err0508}
H.~Enger, A.~Recknagel, and D.~Roggenkamp, \textsl{Permutation branes and linear
  matrix factorisations}, \doi{10.1088/1126-6708/2006/01/087}{JHEP \textbf{0601} (2006), 087},
  \href{http://www.arxiv.org/abs/hep-th/0508053}{[hep-th/0508053]}.

\bibitem[FFRS]{ffrs0909.5013}
J.~Fr\"ohlich, J.~Fuchs, I.~Runkel and C.~Schweigert,
\textsl{Defect lines, dualities, and generalised orbifolds}, 
\bdoi{10.1142/9789814304634_0056}{Proceedings of the XVI Interna-}{tional Congress on Mathematical Physics, Prague, August 3--8, 2009}, \href{http://arxiv.org/abs/0909.5013}{[arXiv:0909.5013]}.

\bibitem[FRS]{tft1}
J.~Fuchs, I.~Runkel and C.~Schweigert,
\textsl{TFT construction of RCFT correlators. I: Partition functions},
Nucl.~Phys.~B {\bf 646} (2002), 353--497, \href{http://arxiv.org/abs/hep-th/0204148}{[hep-th/0204148]}.

\bibitem[FS]{fs0901.4886}
J.~Fuchs and C.~Stigner, 
\textsl{On Frobenius algebras in rigid monoidal categories},
 \bhref{http://ajse.kfupm.edu.sa/articles/332C_P.12.pdf}{The Arabian Journal for Science and Engineering \textbf{33-2C}}{ (2008), 175--191},
  \href{http://arxiv.org/abs/0901.4886}{[arXiv:0901.4886]}.

\bibitem[Gan]{g1304.7312}
S.~Ganatra, \textsl{Symplectic cohomology and duality for the wrapped Fukaya category}, 
\href{http://arxiv.org/abs/1304.7312}{[arXiv:1304.7312]}.

\bibitem[GQ]{Gepner:1986hr}
  D.~Gepner and Z.~Qiu,
  \textsl{Modular invariant partition functions for parafermionic field theories},
  \doi{10.1016/0550-3213(87)90348-8}{Nucl. Phys. B {\bf 285} (1987), 423--453}.

\bibitem[HL]{hl0404}
M.~Herbst and C.~I.~Lazaroiu, \textsl{Localization and traces in open-closed
  topological {L}andau-{G}inzburg models}, \doi{10.1088/1126-6708/2005/05/044}{JHEP \textbf{0505} (2005), 044},
  \bhref{http://www.arxiv.org/abs/hep-th/0404184}{[hep-}{th/0404184]}.

\bibitem[HKK+]{mirror}
K.~Hori, S.~Katz, A.~Klemm, R.~Pandharipande, R.~Thomas, C.~Vafa, R.~Vakil, and
  E.~Zaslow, \textsl{Mirror symmetry}, 
    \bhref{http://www.claymath.org/publications/Mirror_Symmetry}{Clay Mathematics Mono-}{graphs, V.~1,
  American Mathematical Society, 2003}.

\bibitem[IV]{IntriligatorVafa1990}
K.~A.~Intriligator and C.~Vafa, \textsl{Landau-Ginzburg orbifolds}, \bdoi{10.1016/0550-3213(90)90535-L}{Nucl. Phys.}{B \textbf{339} (1990), 95--120}.

\bibitem[JS1]{JSGoTCI}
A.~Joyal and R.~Street, \textsl{The geometry of tensor calculus I}, \bdoi{10.1016/0001-8708(91)90003-P}{Advances in}{Math. \textbf{88} (1991), 55--112}.
  
\bibitem[JS2]{JSGoTCII}
A.~Joyal and R.~Street, \textsl{The geometry of tensor calculus II}, 
draft available at 
\href{http://maths.mq.edu.au/~street/GTCII.pdf}{http://maths.mq.edu.au/\textasciitilde street/GTCII.pdf}. 

\bibitem[KKL+]{Kachru:2000an}
  S.~Kachru, S.~Katz, A.~Lawrence, and J.~McGreevy,
  \textsl{Mirror symmetry for open strings}, 
  \doi{10.1103/PhysRevD.62.126005}{Phys. Rev. D {\bf 62} (2000), 126005}, 
    \href{http://www.arxiv.org/abs/hep-th/0006047}{[hep-th/0006047]}.

\bibitem[KL1]{kl0210}
A.~Kapustin and Y.~Li, \textsl{D-branes in {L}andau-{G}inzburg {M}odels and
  {A}lgebraic {G}eometry}, \doi{10.1088/1126-6708/2003/12/005}{JHEP \textbf{0312} (2003), 005},
  \href{http://www.arxiv.org/abs/hep-th/0210296}{[hep-th/0210296]}.

\bibitem[KL2]{kl0305}
A.~Kapustin and Y.~Li, \textsl{Topological {C}orrelators in {L}andau-{G}inzburg {M}odels with
  {B}oundaries}, Adv. Theor. Math. Phys. \textbf{7} (2004), 727--749,
  \href{http://www.arxiv.org/abs/hep-th/0305136}{[hep-th/0305136]}.
  
\bibitem[KR]{kr0405232}
A.~Kapustin and L.~Rozansky, \textsl{On the relation between open and closed
  topological strings}, \doi{10.1007/s00220-004-1227-z}{Commun. Math. Phys. \textbf{252} (2004), 393--414},
  \bhref{http://arxiv.org/abs/hep-th/0405232}{[hep-}{th/0405232]}.

\bibitem[Lau]{ladia}
A.~D.~Lauda, \textsl{An introduction to diagrammatic algebra and categorified quantum $\mathfrak{sl}_2$}, Bulletin of the Institute of Mathematics Academia Sinica (New Series), Vol. \textbf{7} (2012), No. 2, 165--270, \href{http://arxiv.org/abs/1106.2128}{[arXiv:1106.2128]}.

\bibitem[Laz1]{l0010269}
C.~I.~Lazaroiu, \textsl{On the structure of open-closed topological field theory in two dimensions}, \doi{10.1016/S0550-3213(01)00135-3}{Nucl.\ Phys.\ B \textbf{603} (2001), 497--530},
  \href{http://www.arxiv.org/abs/hep-th/0010269}{[hep-th/0010269]}.

\bibitem[Laz2]{l0312}
C.~I.~Lazaroiu, \textsl{On the boundary coupling of topological
  {L}andau-{G}inzburg models}, \doi{10.1088/1126-6708/2005/05/037}{JHEP \textbf{0505} (2005), 037},
  \arxiv{hep-th/0312286}{[hep-th/0312286]}.

\bibitem[LVW]{lvw1989}
W.~Lerche, C.~Vafa, and N.~Warner, \textsl{Chiral {R}ings in {$N=2$}
  {S}uperconformal {T}heories}, 
  \doi{10.1016/0550-3213(89)90474-4}{Nucl. Phys.~B \textbf{324} (1989), 427}.

\bibitem[Lur]{l0905.0465}
J.~Lurie, \textsl{On the Classification of Topological Field Theories},
  \href{http://arxiv.org/abs/0905.0465}{[arXiv:0905.0465]}.  

\bibitem[MS]{ms0609042}
G.~W.~Moore and G.~Segal, \textsl{D-branes and {K}-theory in {2D} topological
  field theory}, \href{http://arxiv.org/abs/hep-th/0609042}{[hep-th/0609042]}.

\bibitem[Mur]{m0912.1629}
D.~Murfet, \textsl{Residues and duality for singularity categories of isolated {G}orenstein singularities},
  \href{http://arxiv.org/abs/0912.1629}{[arXiv:0912.1629]}.  

\bibitem[NR]{NovakRunkel}
S.~Novak and I.~Runkel, \textsl{State sum construction of two-dimensional topological quantum field theories on spin surfaces}, 
\href{http://arxiv.org/abs/1402.2839}{[arXiv:1402.2839]}.

\bibitem[PV]{pv1002.2116}
A.~Polishchuk and A.~Vaintrob, \textsl{Chern characters and {H}irzebruch-{R}iemann-{R}och formula for matrix factorisations}, 
\bdoi{doi:10.1215/00127094-1645540}{Duke Math.\ J.\ {\bf 161}}{(2012), 1863--1926},
\href{http://arxiv.org/abs/1002.2116}{[arXiv:1002.2116]}. 

\bibitem[RvdB]{rvdb9911242}
I.~Reiten and M.~Van den Bergh, \textsl{Noetherian hereditary categories satisfying Serre duality}, \bdoi{10.1090/S0894-0347-02-00387-9 }{Journal of the American Mathematical Society \textbf{15}}{(2002), 295--366}, 
    \href{http://arxiv.org/abs/math/9911242}{[math.RT/9911242]}.
    
\bibitem[Shk1]{s0702590}
D.~Shklyarov, \textsl{On Serre duality for compact homologically smooth DG algebras}, 
  \href{http://arxiv.org/abs/math/0702590}{[math.RA/0702590]}.  

\bibitem[Shk2]{s0710.1937}
D.~Shklyarov, \textsl{Hirzebruch-Riemann-Roch theorem for DG algebras}, 
  \href{http://arxiv.org/abs/0710.1937}{[arXiv:0710.1937]}.  

\bibitem[SY]{syFrobeniusAlgebrasI}
A.~Skowro\'{n}ski and K.~Yamagata, \textsl{Frobenius Algebras I: Basic Representation Theory}, \bdoi{10.4171/102}{EMS Textbooks in Mathematics, European Mathe-}{matical Society Publishing House, Zurich, 2012}.

\bibitem[Vaf]{v1989}
C.~Vafa, \textsl{String vacua and orbifoldized LG models}, 
\bdoi{10.1142/S0217732389001350}{Mod. Phys. Lett.}{A~\textbf{4} (1989), 1169}.

\bibitem[Wal]{w0412274}
J.~Walcher, \textsl{Stability of Landau-Ginzburg branes}, \bdoi{10.1063/1.2007590}{J.~Math.~Phys.~\textbf{46}}{(2005), 082305}, 
    \href{http://arxiv.org/abs/hep-th/0412274}{[hep-th/0412274]}.

\end{thebibliography}
\end{document}